\documentclass[12pt]{article}
\usepackage{amsmath}
\usepackage{graphicx}
\usepackage{enumerate}

\usepackage{natbib} 

\newcommand{\blind}{1}

\usepackage{amsmath,amssymb,amsfonts,amsthm,bbm}

\usepackage{mathtools}   % for \norm
\usepackage{stmaryrd} % for \llbracket     1 \rrbracket

\usepackage{enumitem}   % for customized enumerate
\usepackage{romannum}

\usepackage[titletoc,title]{appendix}  % Appendix A

\usepackage[table,xcdraw,dvipsnames]{xcolor}   % for \textcolor, \colorlet, ...
\usepackage{hyperref}
\newcommand\myshade{85}
\colorlet{mylinkcolor}{YellowOrange}
\colorlet{mycitecolor}{Aquamarine}
\colorlet{myurlcolor}{violet}

\hypersetup{
  linkcolor  = mylinkcolor!\myshade!black,
  citecolor  = mycitecolor!\myshade!black,
  urlcolor   = myurlcolor!\myshade!black,
  colorlinks = true,
}

\usepackage{caption}
\usepackage{subcaption}

\usepackage[normalem]{ulem}
\usepackage{setspace}

\usepackage{graphicx}
\usepackage{comment}
\graphicspath{{figs/}}

\renewcommand{\hat}{\widehat}
\renewcommand{\tilde}{\widetilde}

\newcommand{\bfm}[1]{\ensuremath{\boldsymbol{#1}}} % bm

\def\bzero{\bfm 0}

\def\ba{\bfm a}   \def\bA{\bfm A}  
\def\bb{\bfm b}   \def\bB{\bfm B}  
   \def\bC{\bfm C}  
   \def\bD{\bfm D}  
   \def\bE{\bfm E}  \def\EE{\mathbb{E}}
\def\bff{\bfm f}  \def\bF{\bfm F}  \def\FF{\mathbb{F}}
   \def\bG{\bfm G}  
   \def\bH{\bfm H}  
   \def\bI{\bfm I}

   \def\bM{\bfm M}  
     
   \def\bO{\bfm O}  
   \def\bP{\bfm P}  \def\PP{\mathbb{P}}
   \def\bQ{\bfm Q}  
   \def\bR{\bfm R}  \def\RR{\mathbb{R}}
   \def\bS{\bfm S}  \def\SS{\mathbb{S}}
     
\def\bu{\bfm u}   \def\bU{\bfm U}  \def\UU{\mathbb{U}}
   \def\bV{\bfm V}  
\def\bw{\bfm w}   \def\bW{\bfm W}  
\def\bx{\bfm x}   \def\bX{\bfm X}  \def\XX{\mathbb{X}}
\def\by{\bfm y}     
   \def\bZ{\bfm Z}  \def\ZZ{\mathbb{Z}}

 \def\cA{{\cal  A}}
 \def\cB{{\cal  B}}

 \def\cM{{\cal  M}}
\def\calN{{\cal  N}} \def\cN{{\cal  N}}

 \def\cT{{\cal  T}}

\newcommand{\bfsym}[1]{\ensuremath{\boldsymbol{#1}}}

 \def\bbeta{\bfsym \beta}
              \def\bGamma{\bfsym \Gamma}
            \def\bDelta {\bfsym {\Delta}}

 \def\btheta{\bfsym {\theta}}           \def\bTheta {\bfsym {\Theta}}
 \def\beps{\bfsym \varepsilon}          
 \def\bepsilon{\bfsym \varepsilon}
 
              \def\bSigma{\bfsym \Sigma}
         \def\bLambda {\bfsym {\Lambda}}

          \def\bPhi{\bfsym {\Phi}}
\def\bpsi{\bfsym {\psi}}          \def\bPsi{\bfsym {\Psi}}

\def\bLam {\bLambda}

  \def\mF{\mathcal{F}}
  
\def\mG{\mathcal G}
\def\mU{\mathcal U}
\def\mF{\mathcal F}
\def\mC{\mathcal C}

\providecommand{\abs}[1]{\left\lvert#1\right\rvert}
\providecommand{\norm}[1]{\left\lVert#1\right\rVert}

\providecommand{\angles}[1]{\left\langle #1 \right\rangle}
\providecommand{\paren}[1]{\left( #1 \right)}

\providecommand{\braces}[1]{\left\{ #1 \right\}}

\providecommand{\defeq}{:=}

\usepackage{mathtools}
\DeclarePairedDelimiterX{\infdivx}[2]{(}{)}{%
  #1 \; \delimsize\| \; #2%
}

\DeclareMathOperator*{\argmin}{argmin}

\DeclareMathOperator{\diag}{diag}

\DeclareMathOperator{\rank}{rank}

\DeclareMathOperator{\sign}{sign}
\DeclareMathOperator{\supp}{supp}

\DeclareMathOperator{\var}{var}

\DeclareMathOperator{\Tr}{Tr}
\DeclareMathOperator{\vect}{vec}
\DeclareMathOperator*{\plim}{\mathit{p}lim}

\newcommand*\xbar[1]{%
  \hbox{%
    \vbox{%
      \hrule height 0.4pt % The actual bar
      \kern0.5ex%         % Distance between bar and symbol
      \hbox{%
        \kern-0em%      % Shortening on the left side
        \ensuremath{#1}%
        \kern-0em%      % Shortening on the right side
      }%
    }%
  }%
} 
\newtheorem{definition}{Definition}
\newtheorem{assumption}[definition]{Assumption}

\newtheorem{lemma}[definition]{Lemma}
\newtheorem{proposition}[definition]{Proposition}
\newtheorem{theorem}[definition]{Theorem}
\newtheorem{corollary}[definition]{Corollary}

\newtheorem{example}[definition]{Example}

\theoremstyle{definition}
\newtheorem{remark}{Remark}

\newcommand{\bigO}[1]{ \mathcal{O} \left( #1 \right) }

\newcommand{\Op}[1]{{\mathcal{O}_p} \left( #1 \right) }
\newcommand{\op}[1]{{\rm o_p} \left( #1 \right) }

\definecolor{royalpurple}{rgb}{0.47, 0.32, 0.66}
\definecolor{greenfresh}{HTML}{00897B}
\definecolor{bluefresh}{HTML}{1E88E5}
\definecolor{redfresh}{HTML}{E53935}

\definecolor{royalpurple}{rgb}{0.47, 0.32, 0.66}

\def\beq{\begin{equation}}
\def\eeq{\end{equation}}

\def\bet{\begin{theorem}}
\def\eet{\end{theorem}}

\def\bel{\begin{lemma}}
\def\eel{\end{lemma}}

\def\eps{\varepsilon}
\def\lam {\lambda}

\def\bLam {\bLambda}

\def\cond{\;|\;}

\usepackage{setspace}
\usepackage{etoolbox}
\BeforeBeginEnvironment{equation*}{\begin{singlespace}\vspace*{-\baselineskip}}
	\AfterEndEnvironment{equation*}{\end{singlespace}\noindent\ignorespaces}
\BeforeBeginEnvironment{equation}{\begin{singlespace}\vspace*{-\baselineskip}}
	\AfterEndEnvironment{equation}{\end{singlespace}\noindent\ignorespaces}
\BeforeBeginEnvironment{align}{\begin{singlespace}\vspace*{-\baselineskip}}
	\AfterEndEnvironment{align}{\end{singlespace}\noindent\ignorespaces}
\BeforeBeginEnvironment{align*}{\begin{singlespace}\vspace*{-\baselineskip}}
	\AfterEndEnvironment{align*}{\end{singlespace}\noindent\ignorespaces}
\BeforeBeginEnvironment{eqnarray}{\begin{singlespace}\vspace*{-\baselineskip}}
	\AfterEndEnvironment{eqnarray}{\end{singlespace}\noindent\ignorespaces}
\BeforeBeginEnvironment{eqnarray*}{\begin{singlespace}\vspace*{-\baselineskip}}
	\AfterEndEnvironment{eqnarray*}{\end{singlespace}\noindent\ignorespaces}
\usepackage[framemethod=TikZ]{mdframed} % box environment & flexible box designs

\usepackage{fancybox}

\usepackage[linesnumbered,ruled,vlined]{algorithm2e}
\SetCommentSty{mycommfont}

\SetKwInput{KwInput}{Input}                % Set the Input
\SetKwInput{KwOutput}{Output}              % set the Output
\SetKwInput{KwInitialize}{Initialize}    % set initialization

\usepackage{listings}             % Include the listings-package
\begin{document}
\pagenumbering{arabic}

\def\spacingset#1{\renewcommand{\baselinestretch}%
{#1}\small\normalsize} \spacingset{1}

%---------------------------------------------------
%
% Titile pagex
%
\if1\blind
{
\title{\bf Factor Augmented Matrix Regression}
\author{
Elynn Y. Chen$^\flat$\footnote{Co-first author} \hspace{3ex}
Jianqing Fan$^\natural$\footnote{Corresponding author.} \hspace{3ex}
Xiaonan Zhu$^{\natural *}$ \\ \normalsize
\medskip
$^\flat$New York University \hspace{3ex}
$^{\natural}$ Princeton University
}
\maketitle
} \fi

\if0\blind
{
\bigskip
\bigskip
\bigskip
\begin{center}
{\LARGE\bf Title}
\end{center}
\medskip
} \fi

\def\r#1{\textcolor{red}{\bf #1}}
\def\b#1{\textcolor{blue}{\bf #1}}

\bigskip
\begin{abstract}
\spacingset{1.68}
We introduce \underline{F}actor-\underline{A}ugmented \underline{Ma}trix \underline{R}egression (FAMAR) to address the growing applications of matrix-variate data and their associated challenges, particularly with high-dimensionality and covariate correlations. 
FAMAR encompasses two key algorithms. The first is a novel non-iterative approach that efficiently estimates the factors and loadings of the matrix factor model, utilizing techniques of pre-training, diverse projection, and block-wise averaging. The second algorithm offers an accelerated solution for penalized matrix factor regression. Both algorithms are supported by established statistical and numerical convergence properties. Empirical evaluations, conducted on synthetic and real economics datasets, demonstrate FAMAR's superiority in terms of accuracy, interpretability, and computational speed. Our application to economic data showcases how matrix factors can be incorporated to predict the GDPs of the countries of interest, and the influence of these factors on the GDPs.
\end{abstract}

\noindent%
{\it Keywords:}  Matrix factor models; Matrix regression; Factor-augmented regression; Diversified projections; High-dimensionality.
\vfill

%%---------------------------------------------------
%%
%% Main Text
%5

\newpage
\spacingset{1.98} % DON'T change the spacing!

\addtolength{\textheight}{.1in}%

\section{Introduction}  \label{sec:intro}

Matrix-variate data are commonly encountered in diverse domains nowadays, such as finance, economics, spatio-temporal data, healthcare, images, and social networks \citep{chen2020semiparametric,chen2023community}. In many problems, especially in high-dimensional regimes, it is common that matrix covariates assume low-rank structure or can be well approximated by low-rank matrices as exemplified by factor models \citep{wainwright2019high,chen2021statistical}. The outcomes of interest are directly influenced by latent low-rank matrix factors, not by the observable matrix covariates themselves. For instance, autism diagnosis relies on identifying specific patterns within brain connectivity networks, underscoring a direct link to reduced-dimensionality factors \citep{nogay2020machine}; regional GDP is shaped by the condensed economic indicators found within spatial-temporal housing data \citep{kopoin2013forecasting}; and the GDP growth of countries is determined by underlying economic factors, demonstrating a direct relationship with these distilled elements \citep{banerjee2005leading}.

To incorporate the matrix factor structure in matrix regressions, 
we introduce the \underline{F}actor \underline{A}ugmented \underline{Ma}trix \underline{R}egression (FAMAR) for observed random variables $\braces{\paren{y_i, \bX_i},\;i\in[n]}$, in which $y_i$ is a scalar response of interest and $\bX_i$ is a $p_1\times p_2$ covariate matrix, defined as 
\begin{equation} \label{eqn:famar}
\begin{cases}
& y_i = \angles{\bA^*,\bF_i} + \angles{\bB^*,\bU_i} + \varepsilon_i, \qquad \text{\it Matrix Factor Regression (MFR)} \\
& \bX_i = \bR^*\bF_i\bC^{*\top} + \bU_i, \qquad\qquad\hspace{3.4ex} \text{\it Matrix Factor Model (MFM)}
\end{cases}
\end{equation}
where $\bF_i$ is the ${k_1\times k_2}$ latent factor matrix of lower dimensions $k_j\ll p_j$ for $j=1,2$, $\bU_i$ is the unobserved idiosyncratic matrix, and $\bR^*$ of dimension $p_1\times k_1$ and $\bC^*$ of dimension $p_2\times k_2$ are the loading matrices of the Matrix Factor Model (MFM) \citep{chen2021statistical}.  The impact of covariate matrix $\bX_i$ on outcome $y_i$ is through the latent factors $\bF_i$ and unobserved idiosyncratic matrix $\bU_i$ with
matrices $\bA^*$ of dimension $k_1 \times k_2$ and $\bB^*$ of dimension $p_1\times p_2$ as the regression
coefficient matrices of the Matrix Factor Regression (MFR). As the first component in MFR is the low-dimensional factor effect, the coefficient $\bA^*$ can be estimated without any structures. We consider both sparse and low-rank structure for the regression coefficient matrix $\bB^*$ since it is high-dimensional.

As practitioners harness the power of various data, model \eqref{eqn:famar} has wide applications in economics and finance. For example, $y_i$ represents the monthly GDP of a country of interest, and $\bX_i\in\RR^{p_1\times p_2}$ consists of the $p_2$ economic variables collected from $p_1$ different countries. 
Model \eqref{eqn:famar} is able to separate economic common factor $\bF_i$ and idiosyncratic shocks $\bU_i$ and provide economists with their respective impacts on the GDP of one country of interest. 
A salient feature of model \eqref{eqn:famar} is that the association between the response $y_i$ and the covariate matrix $\bX_i$ is established directly through the factor $\bF_i$ and idiosyncratic matrix $\bU_i$. 
This structure makes covariate inputs much weakly correlated and relates to various traditional models, as illustrated below.

\begin{example}[Matrix regression with MFM covariates]
Matrix regression is written as
\begin{equation} \label{eqn:mr}
\begin{cases}
& y_i = \angles{\bB^*,\bX_i} + \eps_i, \hspace{6ex} \text{\it Matrix Regression (MR)} \\
& \bX_i = \bR^*\bF_i\bC^{*\top} + \bU_i, \hspace{3ex} \text{\it Matrix Factor Model (MFM)}
\end{cases}
\end{equation}
Model \eqref{eqn:mr} is a special case of model \eqref{eqn:famar} with $\bA^* = {\bR^*}^\top\bB^*\bC^*$.   
\end{example}

While we directly observe $\bX_i$, using the matrix regression model \eqref{eqn:mr} might seem straightforward. The influence of the latent factor $\bF_i$ in model \eqref{eqn:mr} is represented as ${\bR^*}^\top\bB^*\bC^*$. However, the general FAMAR model introduces $\bfm\Phi^* := \bA^* - \bR^\top\bB^*\bC$, which quantifies the additional impact of the latent factor $\bF_i$ in model \eqref{eqn:famar} that is not accounted for by the observed predictor $\bX_i$ in model \eqref{eqn:mr}.

\begin{example}[Reduce-rank matrix regression]
Low-rank matrix regression $y_i = \angles{\bB^*,\bX_i} + \eps_i$ corresponds to a special case of \eqref{eqn:famar} with low-rank coefficient matrix $\bB^*$, matrix factor $\bF_i= 0$ and $\bU_i = \bX_i$ when $\bX_i$ does not admit a factor structure, and $\bfm\Phi^* = 0$ when it does. 
\end{example}

\begin{example}[Sparse matrix regression]
Sparse matrix regression $y_i = \angles{\bB^*,\bX_i} + \eps_i$ is a special case of model \eqref{eqn:famar} with sparse coefficient matrix $\bB^*$, and $\bX_i$ has not low-rank component, i.e., $\bF_i = 0$ and $\bU_i = \bX_i$.
\end{example}

There are several advantages of considering the more general model \eqref{eqn:famar}. First, in reality, especially when the columns and rows of $\bX_i$ are highly correlated, the leading factors are likely to make additional contributions to the outcome, beyond a fixed portion represented by $\bA^* = {\bR^*}^\top\bB^*\bC^*$ in model~\eqref{eqn:mr}. 
To this end, model \eqref{eqn:famar} involves incorporating the primary factors into a low-rank regression that extends the linear space defined by $\bX_i$ in useful ways by noting that using $\bF_i$ and $\bU_i$ as covariates are the same as the usage of $\bF_i$ and $\bX_i$. In addition, even when model~\eqref{eqn:mr} is correctly specified, model \eqref{eqn:famar} is helpful in dealing with the colinearity of $\bX_i$ due to the factor structure \citep{fan2020factor}. 
This is also confirmed in our empirical analysis. When covariates $\bX_i$ are highly correlated, direct low-rank regression on $\bX_i$'s often overestimates the rank due to their dependency structure, reducing the effectiveness of regularization and the interpretability of estimation.
FAMAR enhances performance by utilizing the space spanned by $(\bF_i, \bU_i)$, which is the same as that span by $(\bF_i, \bX_i)$, with augmented feature $\bF_i$, capitalizing further on underlying information. Additionally, its factor model approach decorrelates the covariance matrix, leading to coefficient estimators that are both more interpretable and stable.
%situations when $\bX_i$ is not fully or correctly observed. 

Factor augmentations in the vector contexts have been studied under various settings and shown to be powerful \citep{fan2020factor, zhou2021measuring, fan2022factor,fan2023latent}.
In contrast, little work has been done on the matrix-variate problems. 
Matrix-variate settings are intrinsically different because of their multidimensional structure, where low rankness comes into play. 
Linear and generalized matrix regression have been studied in \cite{zhou2014regularized,wang2017generalized,fan2019generalized}.
However, model~\eqref{eqn:famar} in the present paper is different in that the association between $y_i$ and $\bX_i$ is not direct as those in the literature but rather through the latent variables $\bF_i$ and $\bU_i$. 
We present in the real data analysis that such an association provides better prediction and interpretation of the effect of common economic factors and idiosyncratic variables.
In addition, existing estimation methods solving model~\eqref{eqn:famar}~MFM entail spectral or iterative optimization methods \citep{chen2019constrained,chen2021statistical,liu2022identification,yu2022projected}. 
We propose a new non-iterative method, consisting of simple procedures of ``Pre-Trained Projection'' and ``Block-Average-an-OLS'', to obtain good estimators of $\bF_i$ and $\bU_i$ from the MFM, which are of independent interest. 

The rest of this paper is organized as follows.
Section \ref{sec:FARMA} develops the FAMAR method in details, including a new method for estimating the MFM parameters in model \eqref{eqn:famar} in Section \ref{sec:div_proj} and an efficient algorithm for solving MFR parameters in model \eqref{eqn:famar} in Section \ref{sec:comp}.
Section \ref{sec:theory} establishes their theoretical properties.
Sections \ref{sec:simul} and \ref{sec:real} present empirical results on synthetic and real datasets.
Section \ref{sec:conc} concludes and discusses future works.
All proofs are relegated to the supplementary material.

\section{Factor-Augmented Matrix Regression}\label{sec:FARMA}

When the factors $\bF_i$ and the residuals $\bU_i$ are directly observable, unknown coefficients in model \eqref{eqn:famar} can be obtained by a general matrix regression \citep{wainwright2019high}. 
In the more challenging setting where $\bF_i$ and $\bU_i$ are not directly observed, we need first to estimate $\bF_i$ from the observed matrix variate $\bX_i$. 
Section \ref{sec:div_proj} introduces a novel non-iterative estimation approach to obtain estimated $\hat{\bF}_i$ and $\hat\bU_i$ in model \eqref{eqn:famar} that achieve optimal convergence rates.
With these, to facilitate the estimation of $\bB^*$ under the high-dimensional setting, we could impose either low-rank or sparse structure on $\bB^*$. 

\paragraph{FAMAR with Low-Rank $\bB^*$.} 
The low-rank structure on $\bB^*$ can be achieved by solving {\em nuclear norm regularized matrix regression} with an accelerated algorithm introduced in Section \ref{sec:comp}:
\begin{equation} \label{eqn:convex}	\big(\hat{\bA},\hat{\bB}\big)\in\argmin_{\bA,\bB} \left\{\frac{1}{2n}\sum_{i=1}^n \left( y_i-\angles{\bA,\hat{\bF}_i}-\angles{\bB,\hat{\bU}_i} \right)^2 + \lambda_n\|\bB\|_*\right\}, 
\end{equation}
where $\|\bB\|_*$ denotes the nuclear norm of the matrix coefficient. 

\paragraph{FAMAR with Sparse $\bB^*$.} 
The overall sparsity structure on $\bB^*$ can be induced by using the $\ell_1$ norm on vectorized $\bB^*$, akin to factor augmented regression in vector form \citep{fan2020factor,fan2023latent}. Specifically, the overall sparsity on $\bB^*$ can be achieved by solving the following optimization problem 
\begin{equation} \label{eqn:sparse_opt}
\big(\hat{\bA}^*,\hat{\bB}^*\big)\in\argmin_{\bA,\bB} \left\{\frac{1}{2n}\sum_{i=1}^n \left( y_i-\angles{\bA,\hat{\bF}_i}-\angles{\bB,\hat{\bU}_i} \right)^2 + \lambda_n\|\vect (\bB)\|_{1}\right\}.
\end{equation}	
We also note that FAMAR provides a broader range of sparsity options compared to its vector-based counterpart, thanks to the multi-dimensional nature of the coefficient matrix $\bB^*$. This allows for the imposition of row-wise or column-wise sparsity in $\bB^*$, an option not available in vector-based models. 
We establish theoretical guarantees for the overall sparsity structure with LASSO in Section \ref{sec:theory:famar:sparse}. Theoretical analysis on FAMAR with row-wise or column-wise sparsity in $\bB^*$ can be carried out by applying, e.g., \cite{yuan2006model, obozinski2011support, hastie2015statistical} in the theoretical frameworks presented therein and are left for future research. 

\subsection{Pre-Trained Projection and Block-Wise Averaging}\label{sec:div_proj}

We present an innovative non-iterative estimation approach that utilizes pre-training and exploits the unique structure inherent in Kronecker products. 
It consists of two major steps, namely ``Pre-Trained Projection'', where we estimate the latent factor matrix $\bF_i$ by a one-time projection, and ``Block-Averaging-an-OLS'', where we exploit the unique Kronecker structure and estimate by simple averaging in a block-wise fashion. 

\medskip
\noindent
\textbf{Estimating matrix factors by ``Pre-Trained Projection''.}
The idea of {\em pre-trained projection} is inspired by the diversified projection proposed by \cite{fan2022learning} for the vector factor models.  
Here for the matrix factor model, we introduce {\em ``row and column diversified projection matrices''} $\bW_1$ and $\bW_2$ in Definition \ref{def:W}.
Accordingly, we can obtain a crude estimator of the matrix factor
\begin{equation}\label{eqn:FX}
\hat{\bF}_i = \frac{1}{p_1 p_2}{\bW}_1^\top\bX_i\bW_2.
\end{equation}
Following from the matrix factor structure of $\bX_i$, we have
\begin{equation}\label{eqn:FFhat}
\hat{\bF}_i = \bH_1 \bF_i \bH_2^\top + \frac{1}{p_1 p_2}{\bW_1}^\top\bU_i\bW_2,
\end{equation}
where $\bH_1 = (p_1)^{-1}{\bW}_1^\top\bR^*$ and $\bH_2 = (p_2)^{-1}{\bW}_2^\top\bC^*$ are transformation matrices. A rotated version of $\bF_i$ (i.e. $\bH_1 \bF_i \bH_2^\top$) can be estimated by $\hat{\bF}_i$ as long as the second idiosyncratic term in \eqref{eqn:FFhat} is averaged away and the first signal term dominates the average of the noise in the second term. Formally, we characterize the desirable projection matrices as follows. 

\begin{definition}[Diversified projection matrices]\label{def:W}
Let $c_1$, $c_2$ be universal positive constants. Matrices $\bW_1\in \RR^{p_1\times k_1}$ and  $\bW_2\in \RR^{p_2\times k_2}$ are said to be diversified projection matrices for MFM in \eqref{eqn:famar} if they satisfy
\begin{enumerate}[label={(\alph*)}]
\item $\|\bW_1\|_{\max} < c_1$ and $ \|\bW_2\|_{\max} <c_2$, where $\|\bW\|_{\max}$ is the element-wise $\ell_\infty$-norm.
\item The matrices $\bH_1 = (p_1)^{-1}{\bW}_1^\top\bR^*\in\RR^{k_1\times k_1}$ and $\bH_2 = (p_2)^{-1}{\bW}_2^\top\bC^*\in\RR^{k_2\times k_2}$ satisfy $\nu_{\min }\big(\bH_1\big) \gg p_1^{-1/2}$ and $\nu_{\min }\big(\bH_2\big) \gg p_2^{-1/2}$, respectively.
\item $\bW_1, \bW_2$ are independent of $\{\bX_i\}_{i=1	}^n$.
\end{enumerate}
\end{definition}

Diversified projection matrices can be achieved in various ways \citep{fan2022learning,fan2022factor}, such as those based on observed characteristics, initial transformation, and the Hadamard projection. 
We consider the idea of {\em ``pre-trainning''}, which is a data-driven method of constructing $\bW_1$ and $\bW_2$ by $\alpha$-PCA \citep{chen2021statistical} through data splitting.
Specifically, with a separate independent set of samples ${\bX}_j^\prime,\  j\in[n^\prime]$, we obtain $\bW_1$ and $\bW_2$ as the estimations of $\bR^*$ and $\bC^*$ utilizing the method of $\alpha$-PCA.   As condition \ref{def:W}(b) does not require a consistent estimation, the pre-trained sample size $n'$ can be a negligible fraction of sample size $n$, as to be formally demonstrated.
Then, with the sample of interest ${\bX}_i,\  i\in[n]$, we obtain an estimator of the matrix factor $\{\hat{\bF}_i\}$, $i\in[n]$ using equation \eqref{eqn:FX} with $\bW_1$ and $\bW_2$ as the diversified projection matrices. 

\medskip
\noindent
\textbf{Estimating the loading matrix by ``Block-Averaging-an-OLS''.}
To estimate $\bR^*$ and $\bC^*$, our proposed method averages over blocks of an OLS estimator of the vectorized MFM \eqref{eqn:famar}. 
The details and rationales are explained as follows. 
The vectorized model of MFM \eqref{eqn:famar} with all samples stacked in rows can be written as $\XX = \FF {\bGamma^*}^\top + \UU$, or equivalently, 
\begin{equation}\label{eqn:fm1}
\XX = \big[ \FF (\bH_2\otimes \bH_1)^\top\big]  \big[ \bGamma^*(\bH_2^{-1}\otimes\bH_1^{-1})\big]^\top + \UU \defeq \mathring{\FF}\;\mathring{\bGamma}^\top + \UU, 
\end{equation}
where $\bH_1$ and $\bH_2$ are the transformation matrices defined in Definition \ref{def:W},
$\FF$ is an $n\times k_1 k_2$ matrix whose $i$-th row is the vectorized $\bF_i$, $\XX$ and $\UU$ are $n\times p_1 p_2$ matrices defined similarly, $\bGamma^*$ is the true $p_1 p_2 \times k_1 k_2$ loading matrix satisfying $\bGamma^* = \bC^*\otimes\bR^*$, 
\begin{equation}
\mathring{\FF}\defeq\FF (\bH_2\otimes \bH_1)^\top
\quad\text{and}\quad
\mathring{\bGamma}\defeq\bGamma^*(\bH_2^{-1}\otimes\bH_1^{-1}) 
\end{equation}
are the rotated truth that our pre-trained projected estimators are targeting at.

Now with the vectorized estimator 
$\hat{\FF}:=\big(\operatorname{vec}(\hat{\bF}_1), \ldots, \operatorname{vec}(\hat{\bF}_n)\big)^\top$ from the {\em ``pre-trained projection''} that well approximates $\mathring{\FF}$, the rotated loading matrix $\mathring{\bGamma}$ can be well-approximated by an OLS: 
\begin{equation}\label{eqn:est-Gamma}
\tilde{\bGamma} = \argmin_{\bGamma} \frac{1}{2n}\big\|\XX - \hat{\FF} \bGamma\big\|_F^2 
= \XX^\top\hat{\FF}(\hat{\FF}^\top\hat{\FF})^{-1}.
\end{equation}
Our estimation target $\mathring{\bGamma} = \mathring{\bC} \otimes\mathring{\bR}$ naturally admits a Kronecker product structure with 
\begin{equation}\label{eqn:rotated-truth}
\mathring{\bR} =\bR^*\bH_1^{-1}, \quad\text{and}\quad
\mathring{\bC} = \bC^*\bH_2^{-1},
\end{equation}
which are the rotated truth.  To take advantage of such a structure, we use block averaging to improve the quality of estimation.
Specifically, $\mathring{\bGamma}$ can be divided into $p_2 k_2$ blocks of $p_1\times k_1$ matrices and the $(i,j)$-th block is matrix $\mathring{c}_{ij}\mathring{\bR}$, where $\mathring{c}_{ij}$ is the $(i,j)$-th element of $\mathring{\bC}$. 
Denote $\bE_{pq}\in\RR^{p\times p q}$ as the matrix consisting of $q$ blocks of $\bI_p$, we have that 
\[
\frac{1}{p_2k_2} \bE_{p_1p_2}\mathring{\bGamma}^{\top} \bE_{k_1k_2}^\top = \left(\frac{1}{p_2k_2} \sum_{i=1}^{p_2} \sum_{j=1}^{k_2} \mathring{c}_{ij}\right) \mathring{\bR}
=: \mathring c_s \mathring \bR
\]
and 
\[
\frac{1}{p_1k_1} \bE_{p_2p_1} \bS_{p_2p_1} \mathring\bGamma^{\top} \bS_{k_2k_1}^\top \bE_{k_2k_1}^\top = \left(\frac{1}{p_1k_1} \sum_{i=1}^{p_1} \sum_{j=1}^{k_1} \mathring{r}_{ij}\right) \mathring{\bC}
=: \mathring r_s \mathring \bC,  
\]
where $\bS_{pq}$ is the shuffle matrix defined by taking slices of the $\bI_{pq}$:
\[
\mathbf{S}_{pq}=\left[\begin{array}{c}
\mathbf{I}_{pq}(1: q: pq,:) \\
\mathbf{I}_{pq}(2: q: pq,:) \\
\vdots \\
\mathbf{I}_{pq}(q: q: pq,:)
\end{array}\right].
\]
MATLAB colon notation is used here to indicate submatrices, where  $j:i:k$ creates a regularly-spaced vector using $i$ as the increment between elements.
This Kronecker structure enables us to refine $\tilde\bGamma$ by block-wise averaging. 
We define the averaging estimator
\vspace{-0.2in}
{
\begin{equation}\label{eqn:blw-avg}
\hat\bR := \frac{1}{p_2k_2} \bE_{p_1p_2}\tilde\bGamma^{\top} \bE_{k_1k_2}^\top ,
\quad\text{and}\quad
\hat\bC = \frac{1}{p_1k_1} \bE_{p_2p_1} \bS_{p_2p_1}\tilde\bGamma^{\top} \bS_{k_2k_1}^\top \bE_{k_2k_1}^\top. 
\end{equation}
Therefore, two constants $\mathring c_s$ and $\mathring r_s$ can be estimated by $\hat c_s:= \frac{1}{p_2k_2} \sum_{i=1}^{p_2} \sum_{j=1}^{k_2} \hat{c}_{ij}$ and $\hat r_s:= \frac{1}{p_1k_1} \sum_{i=1}^{p_1} \sum_{j=1}^{k_1} \hat{r}_{ij}$, respectively, where $\hat c_{ij}$ (resp. $\hat r_{ij}$) is the $(i,j)$-th element of $\hat\bC$ (resp. $\hat\bR$).
Accordingly, the estimator of the idiosyncratic matrix $\bU_i$ is given by
\begin{equation}\label{eqn:idio}
\hat{\bU}_i := \bX_i - \frac{1}{\hat c_s \hat r_s} \hat{\bR}\hat{\bF}_i\hat{\bC}^\top,
\end{equation}
}
where $\hat{\bF}_i$ is obtained by ``pre-trained projection'' and $\hat{\bR}$ and $\hat{\bC}$ are estimated by ``block-averaging-an-OLS''.  Note that $\hat{\bF}_i$, $\hat{\bR}$, and $\hat{\bC}$ estimate the affine transformed version of $\bF_i$, $\bR^*$, and $\bC^*$, while $\hat \bU_i$ directly estimates $\bU_i$. 

\subsection{Accelerated Algorithm for Solving Matrix Factor Regression} \label{sec:comp}

Using $\hat{\bF}_i$ and $\hat{\bU}_i$ estimated through {\em pre-trained projection} and {\em block-averaging-an-OLS}, we estimate the model coefficients $\bB^*$ and $\bA^*$ for the MFR \eqref{eqn:famar} by applying Nesterov's accelerated iterative algorithm \citep{nesterov2013gradient,ji2009accelerated} to solve the optimization program \eqref{eqn:convex}. We explore both low-rank and sparse settings for $\bB^*$ in our subsequent analysis.

\paragraph{Accelerated algorithm for low-rank coefficient matrix $\bB$.}
The process is detailed in Algorithm \ref{alg:accelarated} as follows. We define the objective function from \eqref{eqn:convex} as:
\begin{equation*}
F_n(\bA,\bB) = f_n(\bA,\bB) + \lambda_n\norm{\bB}_*, 
\end{equation*}
where 
\begin{equation*}
f_n(\bA,\bB) := (2n)^{-1}\sum_{i=1}^n\left( y_i-\angles{\bA,\hat{\bF}_i}-\angles{\bB,\hat{\bU}_i} \right)^2.
\end{equation*}
We now drop the $\hat{\cdot}$ in $\hat{\bF}_i$ and $\hat{\bU}_i$ and use ${\bF}_i$ and ${\bU}_i$ instead for clear presentation. 
The accelerated algorithm keeps two sequences $\{(\bB^{(k)}, \bA^{(k)})\}$ and $\{(\bP^{(k)}, \bQ^{(k)})\}$ and updated them iteratively. 
At the beginning of the $k$-th step, $\{(\bA^{(k-1)}, \bB^{(k-1)})\}$ is the approximate solution from the last step, and $\{(\bP^{(k)}, \bQ^{(k)})\}$ is the search point for the current step.
The accelerated algorithm performs the gradient descent update at the search point $\{(\bP^{(k)}, \bQ^{(k)})\}$. 
Specifically, we define the update equations as
\begin{equation*} 
p_{L^{(k)}}(\bA, \bB) \defeq \cT_{\lam_n/L^{(k)}}\paren{\bB-\frac{1}{L^{(k)}}\nabla_B f_n(\bA, \bB)}, \quad 
q_{L^{(k)}}(\bA, \bB) \defeq \bA-\frac{1}{L^{(k)}}\nabla_A f_n(\bA, \bB),
\end{equation*}
where $L^{(k)}$ is an appropriate step size, $\cT_{\lam}(\bC)$ is an operator defined as $\cT_{\lam}(\bC)=\bU\bD_{\lam}\bV^{\top}$ on a matrix $\bC$ with SVD $\bC=\bU\bD\bV^{\top}$ and $\bD_{\lam}$ is diagonal with the soft thresholding $(d_{\lam})_{ii} = \sign(d_{ii})\max\{0,\abs{d_{ii}} - \lam\}$.
The solution for the $k$-the step is updated as 
\begin{equation*}
\bB^{(k)} \leftarrow p_{L^{(k)}}(\bQ^{(k)}, \bP^{(k)}),\quad
\bA^{(k)} \leftarrow q_{L^{(k)}}(\bQ^{(k)}, \bP^{(k)}).
\end{equation*}
At last, the search point for the next step, $(\bP^{(k+1)},\bQ^{(k+1)})$, is constructed as a linear combination of the latest two approximate solutions at step $(k)$ and $(k-1)$. Specifically, 
\begin{equation}
(\bP^{(k+1)},\bQ^{(k+1)}) \leftarrow (\bB^{(k)}, \bA^{(k)}) + \frac{\alpha^{(k)} -1}{\alpha^{(k+1)}}\big((\bB^{(k)}, \bA^{(k)})-(\bB^{(k-1)}, \bA^{(k-1)})\big),
\end{equation}
where we set $\alpha^{(1)}=1$ and update it iteratively by $\alpha^{(k+1)} = 0.5\times(1 +\sqrt{1+4(\alpha^{(k)})^2})$. 
To choose an appropriate step size $L^{(k)}$, we start from an initial estimate $L_0$ and increase this estimate with a multiplicative factor $\gamma>1$ repeatedly until the convergence condition 
\begin{equation}\label{eqn:conv-cond}
\begin{aligned}
F_n\big(q_{L}(\bQ^{(k)}, \bP^{(k)}),p_{L}(\bQ^{(k)},\bP^{(k)})\big) & \geq f_n(\bQ^{(k)},\bP^{(k)}) + H_{L,B}(p_{L}(\bQ^{(k)}, \bP^{(k)});\bQ^{(k)},\bP^{(k)}) \\
& + H_{L,A}(q_{L}(\bQ^{(k)}, \bP^{(k)});\bQ^{(k)},\bP^{(k)}) + \lambda_n\| p_{L}(\bP^{(k)})\|_{*}, 
\end{aligned}
\end{equation}
is satisfied for $L=L^{(k)}$. The $H_{L,A}(\bA';\bA,\bB)$ and $H_{L,B}(\bB';\bA,\bB)$ are defined, respectively, as
\begin{equation}\label{eqn:H_L_A_B}
\begin{aligned}
H_{L,A}(\bA'; \bA,\bB) &\defeq \angles{\bA'-\bA,\nabla_A f_n(\bA,\bB)}+\frac{L}{2}\|\bA'-\bA\|_F^2 \\
H_{L,B}(\bB'; \bA,\bB) &\defeq \angles{\bB'-\bB,\nabla_B f_n(\bA,\bB)}+\frac{L}{2}\|\bB'-\bB\|_F^2. 
\end{aligned}
\end{equation}
Lemma \ref{thm:convergence} in the supplementary shows that Algorithm \ref{alg:accelarated} converges under condition \eqref{eqn:conv-cond}.

An alternative numerical method involves initially estimating $\bA^*$ using linear regression on the vectorized form of MFM \eqref{eqn:famar}, followed by estimating $\bB^*$ as a general low-rank matrix regression problem using the residuals from the first step, leveraging the fact that $\hat\bF_i$'s and $\hat\bU_i$'s are uncorrelated. For the low-rank setting, we maintain the joint update of both parameter matrices in our description to align with the convergence proof provided.
In contrast, for the sparse setting, we adopt this two-step method, enabling the direct application of existing accelerated algorithms for sparse vector regression.

\begin{algorithm}[ht!]
\caption{Accelerated Projected Gradient Descent}
\label{alg:accelarated}
\SetKwInOut{Input}{Input}
\SetKwInOut{Output}{Output}
\Input{Initial values $(\bB^{(0)}, \bA^{(0)})$, step size $L^{(0)}$, and $\gamma>1$ which is the multiplicative factor for searching for $L$ satisfying a convergence condition.}

\tcc{******* Initialization ******* }

Set initial search points $(\bP^{(1)}, \bQ^{(1)}) = (\bB^{(0)}, \bA^{(0)})$, and $\alpha^{(1)}=1$ which is used to generate weights for search points. 

Set $k\leftarrow 1$.

\Repeat{$\frac{\|\bB^{(k)}-\bB^{(k-1)}\|_F}{1+\|\bB^{(k-1)}\|_F} + \frac{\|\bA^{(k)}-\bA^{(k-1)}\|_F}{1+\|\bA^{(k-1)}\|_F}<\epsilon$}{

\tcc{******* Estimate step size $L^{(k)}$ to ensure convergence condition ******* }

Set: $\bar{L}\leftarrow L^{(k-1)}$.

Calculate: 

$p_{\bar{L}}(\bQ^{(k)}, \bP^{(k)}) \leftarrow\cT_{\lambda_n/\bar{L}}\big(\bP^{(k)}-\frac{1}{\bar{L}} \nabla_B f_n(\bQ^{(k)}, \bP^{(k)})\big)$.

$q_{\bar{L}}(\bQ^{(k)}, \bP^{(k)})\leftarrow\bQ^{(k)}-\frac{1}{\bar{L}} \nabla_A f_n(\bQ^{(k)}, \bP^{(k)})$. 

\Repeat{Condition \eqref{eqn:conv-cond} is satisfied}{

Set: $\bar{L}\leftarrow\gamma\bar{L}$. 

Calculate:

$p_{\bar{L}}(\bQ^{(k)}, \bP^{(k)}) \leftarrow\cT_{\lambda_n/\bar{L}}\big(\bP^{(k)}-\frac{1}{\bar{L}} \nabla_B f_n(\bQ^{(k)}, \bP^{(k)})\big)$.

$q_{\bar{L}}(\bQ^{(k)}, \bP^{(k)}) \leftarrow \bQ^{(k)}-\frac{1}{\bar{L}} \nabla_A f_n(\bQ^{(k)}, \bP^{(k)})$.
}

Set: $L^{(k)}\leftarrow\bar{L}$. 

\tcc{******* Update iterates by projected gradient descent ******* }

Update approximates: 

$\bB^{(k)} \leftarrow \cT_{\lambda_n/L^{(k)}}\big(\bP^{(k)}-\frac{1}{L^{(k)}} \nabla_B f_n(\bQ^{(k)}, \bP^{(k)})\big)$.

$\bA^{(k)} \leftarrow \bQ^{(k)}-\frac{1}{L^{(k)}} \nabla_A f_n(\bQ^{(k)}, \bP^{(k)})$.

Update weighting sequence: $\alpha^{(k+1)} = 0.5\times\left(1 +\sqrt{1+4(\alpha^{(k)})^2}\right)$. 

Update search points: $(\bP^{(k+1)},\bQ^{(k+1)}) \leftarrow (\bB^{(k)}, \bA^{(k)}) + \frac{\alpha^{(k)} -1}{\alpha^{(k+1)}}\big((\bB^{(k)}, \bA^{(k)})-(\bB^{(k-1)}, \bA^{(k-1)})\big)$.

$k\leftarrow k+1$.
}

\Output{$\bA^{(k)}$ and $\bB^{(k)}$}
\end{algorithm}

\paragraph{Accelerated algorithm for sparse coefficient matrix $\bB$.}
Define $\hat{\PP}:=\hat{\FF}(\hat{\FF}^\top \hat{\FF})^{-1}\hat{\FF}^\top$ as the projection matrix, and let $\tilde\by := (\bI_n - \hat{\PP})\by$ represent the residuals of the response vector $\by$ after projection onto the column space of $\hat{\FF}$, which contains vectorized factors. Given that $\hat{\UU} = (\bI_n - \hat{\PP})\XX$ results in $\hat{\FF}^\top \hat{\UU} = \bzero$, the solution to \eqref{eqn:sparse_opt} can be straightforwardly determined as:
\begin{align}
& \vect(\hat\bA) =  \big(\hat{\FF}^\top \hat{\FF}\big)^{-1}\hat{\FF}^\top \by, \nonumber \\
&\vect(\hat\bB) \in \argmin_{\bb\in\RR^{p_1 p_2}}  \left\{\frac{1}{2n}\sum_{i=1}^n \left\|\tilde{y}_i - \angles{\bB,\hat{\bU}_i}\right\|_2^2 + \lambda_n\|\vect (\bB)\|_{1}\right\}. \label{eqn:sparse_B_opt}
\end{align}
Nesterov’s accelerated gradient algorithms \citep{simon2013sparse,yang2015fast,yu2015high} can be effectively utilized to solve convex problem \eqref{eqn:sparse_B_opt}. For a comprehensive review of additional computational methods for sparsity, refer to Chapter 3.5 in \cite{fan2020statistical}.

\section{Theoretical Results}\label{sec:theory}

\subsection{Analysis of the Pre-Trained Projection}

We start with some assumptions that are necessary for our theoretical development.

\begin{assumption} \label{asum:FU}
Without loss of generality, we assume that $\bX_i$, $i\in[n]$ are mean zero, and thus $\bF_i$ and $\bU_i$ are mean zero. Throughout this paper, we assume that $\{\bX_i\}_{i=1}^n$ are observed, and latent $\{(\bF_i,\bU_i)\}_{i=1}^n$ are i.i.d. copies of $(\bF,\bU)$. 
\end{assumption}
Additionally, we adhere to Assumptions \ref{asum:FU_subG} and \ref{asu:F} outlined in the appendix. These are standard assumptions in factor models, as discussed in \cite{chen2021statistical}. They are presented in the appendix due to constraints on page length.

Identification issue is inherent with latent factor models: for any invertible matrices $\bO_1\in\RR^{k_1\times k_1}$ and $\bO_2\in\RR^{k_2\times k_2}$, the triples $(\bR^*,\bF_i,\bC^*)$ and $(\bR^*\bO_1, \bO_1^{-1}\bF_i\bO_2^{-1}, \bO_2\bC^*)$ are equivalent under the MFM model \eqref{eqn:famar}. 
The following assumptions are commonly used to separate $\bF_i$ and $\bU_i$ and identify one targeting value of $(\bR^*, \bC^*)$.

\begin{assumption}
\begin{enumerate}[label={(\alph*)}]\label{asm:RC}
\item There exist universal constants $c_1$, $c_2$, $C_1$, and $C_2$ such that 
\[
\begin{aligned}
&c_1\leq \lambda_{min}\big(\frac{1}{p_1}{\bR^*}^\top{\bR^*}\big) \leq \lambda_{max}\big(\frac{1}{p_1}{\bR^*}^\top{\bR^*}\big) \leq C_1,\\
&c_2\leq \lambda_{min}\big(\frac{1}{p_2}{\bC^*}^\top{\bC^*}\big) \leq \lambda_{max}\big(\frac{1}{p_2}{\bC^*}^\top\bC^*\big) \leq C_2.
\end{aligned}
\]
{ \item $\EE\left[\frac{1}{p_1}\bF_i^\top{\bR^*}{\bR^*}^\top\bF_i\right]=\bI_{k_2}$, and $\EE\left[\frac{1}{p_2}\bF_i{\bC^*}^\top{\bC^*}\bF_i^\top\right]=\bI_{k_1}$.}
\end{enumerate}
\end{assumption}

The following proposition shows that the pre-trained projection matrices $\bW_1$ and $\bW_2$ proposed in Section \ref{sec:div_proj} satisfies  Definition \ref{def:W} as diversified projection matrices.
\begin{proposition}\label{prop:H}
Suppose $\bX_i$, $i\in[n]$ are i.i.d.~copies of $\bX$ defined in the MFM \eqref{eqn:famar}. Under Assumption \ref{asum:FU}, \ref{asm:RC}, and \ref{asum:FU_subG} in the appendix, there exist universal constants $c_1, c_2,$ and $c_3$ such that the matrices $\bH_1=\frac{1}{p_1}\bW_1^\top\bR^*$ and $\bH_2=\frac{1}{p_2}\bW_2^\top\bC^*$ satisfy
\[
c_1 -c_2\left(\sqrt{t(k_1\vee k_2)}\sqrt{\frac{\log k_j + t}{n}} + \frac{1}{\sqrt{p_j}} \right) \leq \nu_{min}\big(\bH_j\big) \leq \nu_{max}\big(\bH_j\big) \leq c_3
\]
for both $j=1,2$ with probability at least $1-11e^{-t}$.
\end{proposition}

Proposition \ref{prop:H} indicates that when $n^\prime \gg k_1k_2 +t$ and $p_1, p_2 \gg 1$, the conditions of diversified projection matrices in Definition \ref{def:W} are satisfied by the pre-trained projection matrices we proposed in Section \ref{sec:div_proj} with $\nu_{\min}(\bH_j) \asymp 1$, for $j=1,2$. 
In particular, by taking $t=\log n$, this condition holds with probability at least $1-11/n$.
With that,
the convergence rate of pre-trained estimator $\hat{\bF}_i$ to the true $\bF_i$ up to transformation matrices is given below.
\begin{proposition}\label{lem:Ferr_nuc}
Suppose that assumptions in Proposition \ref{prop:H} and Assumption \ref{asu:F} (a) hold.
Invoking Proposition \ref{prop:H} with $t=\log n$, we have, as $p_1\wedge p_2 \rightarrow \infty$ and $n$ is finite or $n\rightarrow\infty$, 
for any $i \in[n]$,
\[
\|\hat{\bF}_i-\bH_1\bF_i\bH_2^\top\|_2 = \Op{\frac{1}{\sqrt{p_1 p_2}}}.
\]
\end{proposition}

\begin{remark}\label{rem:F_err}
The optimality of the rate presented in Proposition \ref{lem:Ferr_nuc} can be understood by considering an oracle scenario where $\bR^*$ and $\bC^*$ are known.
This scenario reduces the problem to estimating $\bF_i$ in a linear regression framework by vectorizing the matrix factor model:
\[
\vect(\bX_i) = \bGamma^* \vect(\bF_i) + \vect(\bU_i)\in\RR^{p_1p_2}.
\]
Here, $\vect(\bF_i)$, with dimension $k_1k_2$, acts as the unknown ``parameter'' vector to be estimated. Standard OLS analysis indicates that the convergence rate of this estimator is $1/\sqrt{p_1p_2}$, aligning with the rate achieved in Proposition \ref{lem:Ferr_nuc}.
\end{remark}

\subsection{Properties of the Block-Wise Averaged Estimators}

Proposition \ref{thm:ols-pre-trained} in Section \ref{append:ols-after-diversify} of the appendix establishes the asymptotic normality of $\tilde\bGamma$ obtained by least squares in \eqref{eqn:est-Gamma}. 
Next we establish the asymptotic normality of the block-wise averaged estimator \eqref{eqn:blw-avg} of the loading matrices and element-wise convergence rate of the estimator \eqref{eqn:idio} of the idiosyncratic component. 
All proofs are relegated to Section \ref{append:mfm-proofs} in the appendix. 

\begin{theorem}\label{thm:MCMR}
Under the condition of Proposition \ref{thm:ols-pre-trained} and ${n}\ll {p_1p_2}\cdot(p_1\wedge p_2)$, we have
\begin{equation}  
\sqrt{np_2k_2} \left(\hat\bR - \mathring{c}_s \mathring{\bR}\right)
\xrightarrow{\ d\ } \cM\cN \big(\bzero, \bSigma_{R,p}, \bSigma_{R,k}\big)
\end{equation}
and
\begin{equation}
\sqrt{np_1k_1} \left(\hat\bC - \mathring{r}_s \mathring{\bC}\right)
\xrightarrow{\ d\ } \cM\cN \big(\bzero, \bSigma_{C,p}, \bSigma_{C,k}\big).
\end{equation}
where $\mathring{c}:=\sum_{i=1}^{p_2}\sum_{j=1}^{k_2}\mathring{c}_{ij}/( {p_2k_2})$, $\mathring{r}:=\sum_{i=1}^{p_1}\sum_{j=1}^{k_1}\mathring{r}_{ij}/( {p_1k_1})$ are scalars, $\mathring{\bR}$ and $\mathring{\bC}$ are rotated truth defined in \eqref{eqn:rotated-truth}, and $\bSigma_{R,p}$, $\bSigma_{C,q}$, $\bSigma_{R,k}$, $\bSigma_{C,k}$ are covariance matrices of dimension $p_1\times p_1$, $p_2\times p_2$, $k_1\times k_1$, and $k_2\times k_2$, respectively, with finite entries. 
\end{theorem}

\begin{remark}
The optimality of the rate for $\hat\bR$ presented in Theorem \ref{thm:MCMR} can be interpreted by considering an oracle scenario where $\bF_i$ and $\bC^*$ are known. This scenario reduces the problem to estimating $\bF_i$ using a linear regression model by stacking the matrix factor model along the second axis:
\[
[\bX_1\; \cdots\;\bX_n] = \bR \; [\bF_1\bC^{*\top}\; \cdots\;\bF_n\bC^{*\top}] + [\bE_1\; \cdots\;\bE_n].
\]
In this model, $\bR$ of dimension $p_1\times k_1$, is treated as the unknown ``parameter'' vector. According to standard OLS analysis, the element-wise convergence rate of this estimator is $1/\sqrt{n p_2}$, which is consistent with the rate determined in Theorem \ref{thm:MCMR}.
\end{remark}

\begin{theorem}\label{thm:U_element}
Under the same condition of Theorem \ref{thm:MCMR}, we have, for each $i\in[n]$, $j\in[p_1]$, and $k\in[p_2]$, that 
\begin{equation}
\left|\hat{u}_{i,jk} - {u}_{i,jk}\right| = \Op{\frac{1}{\sqrt{np_1}}+ \frac{1}{\sqrt{np_2}} + \frac{1}{\sqrt{p_1p_2}} },
\end{equation}
where $\hat{u}_{i,jk}$ and ${u}_{i,jk}$ are the $(j,k)$-th element of $\hat\bU_i$ and $\bU_i$, respectively. 
\end{theorem}
\begin{remark}
The three components that shows up in the convergence rate in Theorem \ref{thm:U_element} comes from the errors in estimating $\bC^*$, $\bR^*$, and $\bF_i$, respectively. 
\end{remark}

\subsection{Statistical Convergence of FAMAR Coefficients}

We focus our theoretical analysis on latent FAMAR \eqref{eqn:famar} where $(\bF_i, \bU_i)$ are not directly observable. 
Theoretical results for observable FAMAR are provided in Section \ref{append:famar-observed-factor} for low-rank $\bB^*$ and in Section \ref{append:sparse-B-observed-F} for overall sparse $\bB^*$ in the supplemental material. 

\subsubsection{FAMAR with Low-Rank $\bB^*$} \label{sec:theory:famar:lowrank}

To study the property of $\hat{\bA}$ and $\hat{\bB}$ obtained by solving \eqref{eqn:convex}, we introduce the adjusted-restricted strong convexity condition (adjusted-RSC), which is a modification of the strong convexity condition (RSC) \citep{negahban2011estimation}.

\begin{definition}[Adjusted-restricted strong convexity]\label{def:rsc}
We let $\Pi_{B^{\perp}}$ denote the projection operator onto the matrix subspace $\{ \bM \cond {\bV_1^*}^{\top}\bM = \bzero \text{ or }  {\bM} {\bV_2^*}  = \bzero\}$ where $\bV_1^*$ and $\bV_2^*$ are the left and right singular space of $\bB^*$ respectively.  
Any $\Delta_B\in\RR^{p_1\times p_2}$ can be decomposed as
\begin{equation}\label{eqn:space-decomp}
\bfm\Delta^{\prime \prime}_B=\Pi_{B^{\perp}}(\bfm\Delta),
\quad\text{and}\quad \bfm\Delta^{\prime}_B=\bfm\Delta-\bfm\Delta^{\prime \prime}_B.
\end{equation}
Define a set in $\RR^{k_1\times k_2}\times \RR^{p_1\times p_2}$ as
\begin{equation}
\mathcal{C}:=\left\{\bfm\Delta=  (\bDelta_A, \bDelta_B) \mid
\bfm{\Delta}_A\in \mathbb{R}^{k_1 \times k_2}, 
%\bfm{\Delta}_A { has the same dimension as } \bA,
\bfm\Delta_B\in \mathbb{R}^{p_1 \times p_2} ,\|\bfm\Delta_B^{\prime \prime}\|_* \leq 3\|\bfm\Delta_B^{\prime}\|_* + \|\bfm\Delta_A\|_*\right\}.
\end{equation}
We say that the operator $\mG_n:\RR^{k_1\times k_2}\times \RR^{p_1\times p_2}\rightarrow \RR^{n}$ satisfies adjusted-RSC over the set $\mC$ if there exists some scalar-valued $\kappa(\mG_n)>0$ such that
\[
\frac{1}{2n}\|\mG_n(\bDelta_A, \bDelta_B)\|^2 \geq \kappa(\mG_n)\paren{\|\bDelta_A\|_F^2 + \|\bDelta_B\|_F^2} \ \text{ for all } (\bDelta_A, \bDelta_B)\in\mC.
\]
\end{definition}

\begin{theorem}\label{thm:consistency2}
Suppose that the projection matrices $\bW_1$ and $\bW_2$ are generated by the pre-trained projection in Section \ref{sec:div_proj} with $n^\prime \gg \log(p_1\vee p_2)$, $n\ll p_1^2p_2^2$. Then, under the same condition of Proposition \ref{thm:ols-pre-trained}, the solutions $(\hat\bA, \hat\bB)$ to the optimization problem \eqref{eqn:convex} satisfy 
\begin{equation} \label{eqn:error-A}
\left\|\bH_1^\top\hat\bA\bH_2 - \bA^*\right\|_F =  \Op{ \frac{\rho}{\sqrt{np_1p_2}}},
\end{equation}
where $\rho := \|\bA^*\|_F$ is the Frobenious norm of the true $k_1\times k_2$ coefficient matrix. In addition, suppose that there exists a positive constant $C_0$ such that $\frac{p_1\vee p_2}{n}<C_0$, and $\mG_n$ satisfies adjusted-RSC with constant $\kappa(\mG_n)$ over the set $\mC$, where $\kappa(\mG_n) > 2C_1(C_0+1)$ for some constant $C_1$.
Then we have
\begin{equation} \label{eqn:error-B}
\|\hat{\bB}-\bB^*\|_F \leq (2\kappa(\mG_n)^{-1})\cdot\lambda_n\cdot\big(2\sqrt{k_1\wedge k_2} + 6\sqrt{2\operatorname{rank}(\bB^*)}\big),
\end{equation}
when the tuning parameter $\lambda_n$ is chosen such that 
$\lambda_n\geq 2\|\xbar\mG_n^*(\bar\bepsilon)\|_2 /n$, where  $\xbar{\mG}_n^*({\beps}) := \sum_{i=1}^n \varepsilon_i \cdot \operatorname{diag}\big(\bH_1^{-1}\hat{\bF}_i(\bH_2^\top)^{-1},\; \hat{\bU}_i\big)$ and  $\xbar\varepsilon_i = y_i - \angles{\bH_1^{-1}\hat{\bF}_i(\bH_2^\top)^{-1},\bA^*}-\angles{\hat{\bU}_i,\bB^*}$.
\end{theorem}

\begin{remark}
In FAMAR model \eqref{eqn:famar}, the signal strength from the factors $\langle \bA^*, \bF_i \rangle$ should be comparable to the signal strength from the idiosyncratic component $\langle \bB^*, \bU_i \rangle$. 
Since the predictors $\bF_i\in\RR^{k_1\times k_2}$ and $\bU_i\in\RR^{p_1\times p_2}$ both have bounded entries, but differ in dimensions, it is reasonable to expect $\|\bA^*\|_F \asymp \bigO{p_1p_2}$ in the low-rank setting when $\|\bB^*\|_F \asymp \bigO{p_1p_2}$.
This aligns with the matrix regression model \eqref{eqn:mr} in Example 2, where the size of each element in $\bA^*$ inherently has order $\bigO{p_1p_2}$, given the dimensions of $\bR^* \in \RR^{p_1 \times k_1}$, $\bB^* \in \RR^{p_1 \times p_2}$, and $\bC^* \in \RR^{p_2 \times k_2}$. 
Consequently, the convergence rate of $\hat\bA$ in \eqref{eqn:error-A} is interpreted for the relative error 
$$\frac{\left\|\bH_1^\top\hat\bA\bH_2 - \bA^*\right\|_F}{\|\bA^*\|_F} = \Op{ \frac{1}{\sqrt{np_1p_2}}}.$$
This contrasts sharply with factor-augmented sparse vector or matrix regression models where $\bB^*$ is sparsely supported on a set $S$, $\norm{\bB^*}_F\asymp \bigO{\abs{S}}$, and $\|\bA^*\|_F \asymp \bigO{k_1 k_2}$.
\end{remark}

The convergence rate of $\hat{\bB}$ described in Theorem \ref{thm:consistency2} can be further clarified when FAMAR reduces to matrix regression with MFM covariates \eqref{eqn:mr}, or equivalently $\bA^* = \bR^\top\bB^*\bC$, and $\beps$ is a sub-Gaussian random vector. This scenario is elaborated in the subsequent proposition.

\begin{proposition}\label{prop:err_order}
Under FAMAR model \eqref{eqn:famar} with $\bA^* = \bR^\top\bB^*\bC$ and the conditions of Theorem \ref{thm:consistency2}, if $\EE[\varepsilon_i [\bF_i\; \bU_i]] = \bfm{0}$, and $\varepsilon_i$ is $\sigma$-sub-Gaussian, 
then when the tuning parameter $\lambda_n$ is chosen such that $\lambda_n\asymp\sqrt{\log(p_1+p_2)/n} + (\log n/n)^{1/4}\sqrt{p_1p_2/n}$, we have
\[
\begin{aligned}
\|\hat{\bB}-\bB^*\|_F = \mathcal O_p\big[\big(2\sqrt{k_1\wedge k_2} &+ 6\sqrt{2\operatorname{rank}(\bB^*)}\big)\big(\sqrt{\frac{\log(k_1+k_2+p_1+p_2)}{n}}\\
&+ \paren{\frac{\log n}{n}}^{1/4} \sqrt{(k_1\wedge k_2)k_1k_2} + \frac{(\log n)^{1/4}}{n^{3/4}}\sqrt{p_1p_2} \big)\big].
\end{aligned}
\]
In particular, since $k_1\asymp k_2 \asymp \rank(\bB^*)\asymp 1$, the result is simplified to
\[
\|\hat{\bB}-\bB^*\|_F = \Op{\sqrt{\frac{\log(p_1+p_2)}{n}}  + \frac{(\log n)^{1/4}\sqrt{p_1p_2}}{n^{3/4}}}.
\]
\end{proposition}

\begin{remark}
The two primary components contributing to the error bound rate for $\hat \bB$ are: (i) the interaction between regression noise $\bepsilon$ and the true covariates $\bF_i$'s and $\bU_i$'s, quantified as $\|n^{-1}\sum_{i=1}^n \varepsilon_i\;{\rm blockdiag}[\bF_i\;\bU_i]\|_2$, and (ii) the interaction of regression noise $\bepsilon$ with the estimation errors of $\bU_i$'s, quantified as $\|n^{-1}\sum_{i=1}^n \varepsilon_i(\hat\bU_i - \bU_i)\|_2$.
\end{remark}

\subsubsection{FAMAR with Sparse $\bB^*$} \label{sec:theory:famar:sparse}

When the true $\bB^*$ exhibits certain sparse structures, such as global, row- or column-wise sparsity, the methodologies and theoretical frameworks developed for high-dimensional vector regression with sparse or group sparse structures can be effectively adapted to the FAMAR model with sparse $\bB^*$. We have established theoretical results for overall sparsity using LASSO. Analogous results could be developed for overall sparsity employing non-concave penalties like SCAD and MCP, or for group sparsity, leveraging the relevant literature, e.g. \cite{yuan2006model, obozinski2011support, hastie2015statistical}. 
We leave these developments for future research to maintain focus and adhere to page limits.

To study the property of $\hat{\bA}$ and $\hat{\bB}$ obtained by solving \eqref{eqn:sparse_opt}, we consider the vectorized version with the restricted strong convexity condition (RSC) in Assumption \ref{asu:Bsparse_Fknown} and adapt the proofs in \cite{fan2020factor}. 
The analysis on $\bA$ is exactly the same as Theorem \ref{thm:consistency2} based on the fact that $\hat{\FF}^\top \hat{\UU} = \bzero$, we include it in the following theorem for completeness.
The consistency of $\hat\bB$ is established in Theorem \ref{thm:Bsparse_Funknown}. 

\begin{assumption}\label{asu:Bsparse_Fknown}
We denote $\btheta^*:= (\vect (\bB^*)^\top, \vect (\bA^*)^\top)^\top\in\RR^{p_1p_2+K_1K_2}$, and $\bbeta^*:= \btheta_{[p_1p_2]} = \vect (\bB^*)$ and denote $S:= \supp (\btheta^*)$, $S_1 = \supp (\bbeta^*)$, and $S_2 = [p_1p_2+K_1K_2]\setminus S$.
Let $\bW := [\UU \;\; \FF]$ represent the $n\times (p_1 p_2 + k_1k_2)$ true covariate matrix, we assume that 
\begin{itemize}
\item[(a)] (Restricted strong convexity) There exist $\kappa_2>\kappa_\infty>0$ such that $\|n(\bW_S^\top \bW_S)^{-1}\|_\infty \leq \frac{1}{2\kappa_\infty}$ and $\|n(\bW_S^\top \bW_S)^{-1}\|_2 \leq \frac{1}{2\kappa_2}$, where $\bW_S$ is the matrix of columns of $\bW$ supported on $S$;
\item[(b)] (Irrepresentable condition) $\|\bW_{S_2}^\top \bW_S (\bW_S^\top\bW_S)^{-1}\|_\infty \leq 1-\tau$ for some $\tau\in(0,1)$.
\end{itemize}
\end{assumption}

\begin{theorem}\label{thm:Bsparse_Funknown}
The solutions $(\hat\bA, \hat\bB)$ to the optimization problem \eqref{eqn:sparse_opt} satisfy 
\[
\left\|\bH_1^\top\hat\bA\bH_2 - \bA^*\right\|_F =  \Op{ \frac{\rho}{\sqrt{np_1p_2}}},
\]
where $\rho := \|\bA^*\|_F$, if the same conditions hold for the factors as in Theorem \ref{thm:consistency2}. In addition, 
suppose Assumption \ref{asu:Bsparse_Fknown} and \ref{asu:Baparse_W} hold. 
Define 
\[
\xbar\delta = \norm{n^{-1}\xbar{\bW}_S^{\top}\paren{\xbar{\bW}\btheta^* - \by}}_{\infty},
\]
let $\xbar{\bW} := [\xbar{\UU} \;\; \xbar{\FF}]$ represent the $n\times (p_1 p_2 + k_1k_2)$ rotated estimated covariate matrix and $\xbar{\bW}_S$ represent sub-matrix of $\xbar\bW$ supported on $S$.
If $\lambda>\frac{14\xbar\delta}{\tau}$, then we have $\supp(\hat\bB) \subseteq \supp(\bB^*)$ and 
\[
\big\|\widehat{\bB}-{\bB}^*\big\|_{\infty} \leq \frac{12\lambda}{5 \kappa_{\infty}}, \\
\quad \big\|\widehat{\bB}-{\bB}^*\big\|_2 \leq \frac{4\lambda\sqrt{|S|}}{\kappa_2}, \\
\quad \big\|\widehat{\bB}-{\bB}^*\big\|_1 \leq  \frac{12\lambda|S|}{5 \kappa_{\infty}}.
\]
\end{theorem}
The proof is presented in Section \ref{append:sparse-B-estimated-F} in the appendix. 
By taking $\lambda \asymp \xbar\delta$, we have 
\begin{equation*}
\|\hat\bB-\bB^*\|_\infty/\xbar\delta = \Op{1}, \quad
\|\hat\bB-\bB^*\|_2/\xbar\delta = \Op{\sqrt{|S|}}, \quad\text{and}\quad 
\|\hat\bB-\bB^*\|_1/\xbar\delta = \Op{|S|}.
\end{equation*}

The convergence rate of $\hat{\bB}$ described from Theorem \ref{thm:Bsparse_Funknown} can be further clarified when FAMAR reduces to matrix regression with MFM covariates \eqref{eqn:mr}, that is $\bA^* = \bR^\top\bB^*\bC$, and $\beps$ is a sub-Gaussian random vector. This scenario is elaborated in the subsequent corollary.

\begin{corollary}\label{thm:Bsparse_Funknown2}
Suppose the conditions for Theorem \ref{thm:Bsparse_Funknown} hold. In addition, we assume that the true matrix regression model is $y_i = \angles{\bX_i,\bB^*}+\varepsilon_i$, i.e. $\bA^* = \bR^{*\top}\bB^*\bC^*$, and that $\EE[\varepsilon_i(\bF_i\ \bU_i)]=0$ and $\varepsilon_i$ is $\sigma$-sub-Gaussian. Then 
\begin{equation}\label{eqn:err-B-sparse-1}
\|\hat\bB-\bB^*\|_\infty
= \Op{\frac{|S|\log(p_1p_2)}{n} + |S|\sqrt{\frac{\log p_1p_2}{n}} \left(\frac{p_1}{p_2} \vee \frac{p_2}{p_1} \right) ^{1/2}}.
\end{equation} 
\end{corollary}

\begin{remark}
When $\bB^*$ is a balanced matrix, that is, $p_1 \asymp p_2$, error bound \eqref{eqn:err-B-sparse-1} can be further reduce to 
\[
\|\hat\bB-\bB^*\|_\infty
= \Op{|S|\sqrt{\frac{\log p_1p_2}{n}}},
\]
which is optimal for vector sparse regression for vectorized $\bB^*$. 
\end{remark}

\section{Simulations} \label{sec:simul}

We study the finite-sample performance of the proposed FAMAR through simulations.

\medskip
\noindent
\textbf{FAMAR Matrix Loading's Asymptotically Normality.}
We test the asymptotic normality for the estimated loading matrices $\hat\bR$ and $\hat\bC$ according to Theorem \ref{thm:MCMR}.
To that end, we conduct experiments as follows:  $n=1000,\ \ n^\prime=500,\ k_1=3,\ k_2=2,\ \rank(\bB^*)=2,\ p_1=20,\ p_2=30$. The true loadings $\bR^*$ and $\bC^*$ are generated such that their elements are i.i.d. from $\cN(2,4)$. The true factors $\bF_i$ and residuals $\bU_i$ have their elements i.i.d. from $\cN(0,1)$ and $\cN(0,0.04)$, respectively. The experiments are repeated 10,000 times, and in each repetition, we compute $\sqrt{np_2k_2} \big(\hat{\bR} - \mathring{c}_s \mathring{\bR}\big)$ and $\sqrt{np_1k_1} \big(\hat{\bC} - \mathring{r}_s \mathring{\bC}\big)$. 
For each elements $(j,k)$, $j\in[p_1]$, $k\in[k_1]$ (resp. $(j^\prime,k^\prime)$, $j^\prime\in[p_2]$, $k^\prime\in[k_2]$), we get 10,000 realizations of $\sqrt{np_2k_2} \big(\hat{\bR} - \mathring{c}_s \mathring{\bR}\big)_{jk}$ (resp. $\sqrt{np_1k_1} \big(\hat{\bC} - \mathring{r}_s \mathring{\bC}\big)_{j^\prime k^\prime}$). Let $\hat s_{jk}$ (resp. $\hat s_{j^\prime k^\prime}^\prime$) be the sample standard deviation of the sequence. 
According to Theorem \ref{thm:MCMR}, $\sqrt{np_2k_2} \big(\hat{\bR} - \mathring{c}_s \mathring{\bR}\big)_{jk}/\hat s_{jk}$ and $\sqrt{np_1k_1} \big(\hat{\bC} - \mathring{r}_s \mathring{\bC}\big)_{j^\prime k^\prime}/\hat s_{j^\prime k^\prime}^\prime$ should distribute closely to standard Gaussian. 
This result is validated empirically by histograms for all the elements of the matrices in Appendix \ref{app:RC_normality}. 
Figure \ref{RC_normality_show} showcases histograms of four randomly picked elements. 
\begin{figure}[ht!]
\centering
\begin{subfigure}[b]{0.45\textwidth}
\centering
\includegraphics[width=\textwidth]{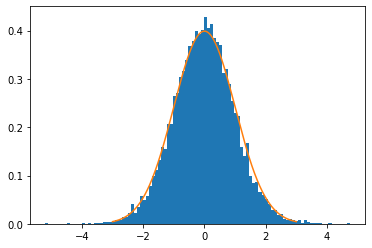}
\caption{$\sqrt{np_2k_2} \big(\hat{\bR} - c_s \mathring{\bR}\big)_{10,1}/\hat s_{10,1}$}
\end{subfigure}
\begin{subfigure}[b]{0.45\textwidth}
\centering
\includegraphics[width=\textwidth]{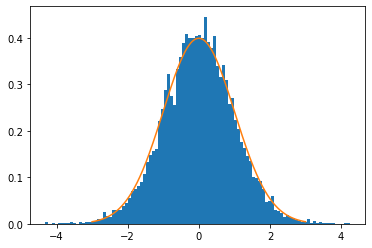}
\caption{$\sqrt{np_2k_2} \big(\hat{\bR} - c_s \mathring{\bR}\big)_{19,3}/\hat s_{19,3}$}
\end{subfigure}
\hfill
\begin{subfigure}[b]{0.45\textwidth}
\centering
\includegraphics[width=\textwidth]{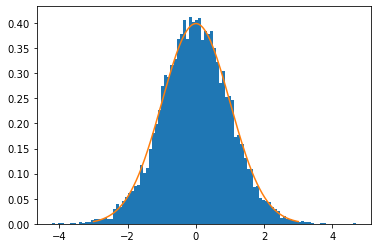}
\caption{$\sqrt{np_1k_1} \big(\hat{\bC} - r_s \mathring{\bC}\big)_{3,1}/\hat s_{3,1}^\prime$}
\end{subfigure}
\begin{subfigure}[b]{0.45\textwidth}
\centering
\includegraphics[width=\textwidth]{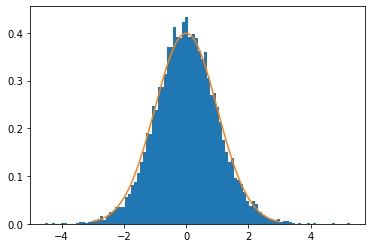}
\caption{$\sqrt{np_1k_1} \big(\hat{\bC} - r_s \mathring{\bC}\big)_{13,2}/\hat s_{13,2}^\prime$}
\end{subfigure}
\caption{Elementwise normality check for the estimated loadings $\hat\bR$ and $\hat\bC$ according to Theorem \ref{thm:MCMR}. The blue histograms are empirical densities, and the orange lines are the standard Gaussian density functions.}
\label{RC_normality_show}
\end{figure}

Next, we investigate the accuracy of the estimated factors, idiosyncratic variable, and FAMAR model coefficients. All experiments are repeated 100 times to reduce the effect of randomness.
The simulations are carried out under two settings of parameters:
\begin{enumerate}[label={{\sc Setting} \Roman*.}, wide, labelwidth=!, labelindent=0pt]
\item Fix $n=1000,\ n_{new}=1000,\ n^\prime=500,\ k_1=k_2=2,\ \rank(\bB^*)=2,$ and change $\ p_1=p_2$ in $\{20,30,40,50,60,70,80,90,100\}$.
\item Fix $p_1=80,\  p_2=50, \ n_{new}=1000, \ n^\prime=500,\ k_1=2,\ k_2=4,\ \operatorname{rank}(\bB^*)=2,$ and change $n$ in $\{500,1000,1500,2000,2500,3000,3500,4000,4500,5000\}$.
\end{enumerate}
More settings of simulation exploration are available in Appendix \ref{append:extra-experiments}.
For all cases, the true loading matrices $\bR^*$ and $\bC^*$ are generated element-wise from $\calN(2,4)$. The true factors $\bF_i$ and residuals $\bU_i$ are generated such that each element is drawn from $\calN(0,1)$ and $\calN(0,0.04)$, respectively. Thereafter, the $\bX_i$ are generated by the matrix factor model $\bX_i =  \bR^*\bF_i(\bC^*)^\top + \bU_i$. 
For $\bB^*$, we first generate a $p_1\times p_2$ matrix, with its individual elements sampled from $\calN(0.5,0.25)$, and then we take its best low-rank ($\bB^*$) approximation by SVD. The true response variables are defined by $y_i = \angles{\bA^*, \bF_i} + \angles{\bB^*, \bU_i} + \eps_i$, where $\bA^*:= 0.5\cdot (\bR^*)^\top \bB^*\bC^*$ and $\eps_i$ are drawn from $\calN(0,1)$.
By setting the variance of loading matrices much higher than that of the residual, we generate $\bX_i$ with highly correlated rows and columns.
The experiments are repeated 100 times to reduce the impact of randomness. The tuning parameter $\lambda$ for the nuclear norm penalty is chosen by cross-validation.

\medskip
\noindent
\textbf{FAMAR Matrix Factor and Idiosyncratic Estimation.}

Figures \ref{fig:pFU_f2r2} and \ref{fig:nFU_f2r2} plot the performance of the proposed algorithm in estimating $\bF$ and $\bU$ under {\sc Setting I.}~and {\sc Setting II.}, respectively. 
The estimation errors are all very small and decrease as either $n$ or $p_1$, $p_2$ increase, except that the error for $\hat{\bF}$ is not affected by increasing number of $n$ as it uses $n'$ pre-training sample size.
Those observations are consistent with our theoretical results. The comparison of the estimation on $\hat\bU$ with and without the block-averaging method showcases the privilege of the averaging algorithm. 

\begin{figure}[ht!]
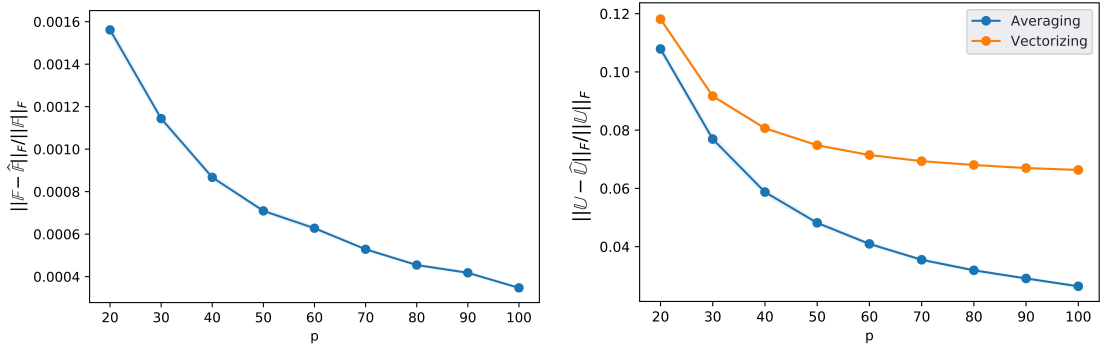

\centering
\begin{subfigure}[b]{0.45\textwidth}
\centering
\includegraphics[width=\textwidth]{figs/coefA7all_F_err_p_k22.pdf}
\end{subfigure}
\begin{subfigure}[b]{0.45\textwidth}
\centering
\includegraphics[width=\textwidth]{figs/coefA7all_U_err_p_k22.pdf}
\end{subfigure}
\caption{Performance of estimated $(\widehat{\bF},\widehat{\bU})$ under fixed $n=1000$ and increasing $p_1=p_2$ from $20$ to $100$. For factors, we report the relative error ${\|\FF(\bH_2\otimes\bH_1)^\top - \widehat\FF\|_F} /{\|\FF(\bH_2\otimes\bH_1)^\top\|_F}$, where $\FF\in\RR^{n\times k_1 k_2}$ is the matrix of vectorized true factors, and $\widehat\FF$ is the estimation. For simplicity of notation, we omit the rotation matrix $\bH_2\otimes\bH_1$ in the figure label. Similar relative errors are reported for the estimation of idiosyncratic components, also comparing the performance with and without the block-average technique. The solid lines are the means among the 100 repetitions, the shadows show the standard deviations, and the vertical lines present the quartiles (which are too small to be visible on the graphs).}
\label{fig:pFU_f2r2}
\end{figure}

\begin{figure}[ht!]
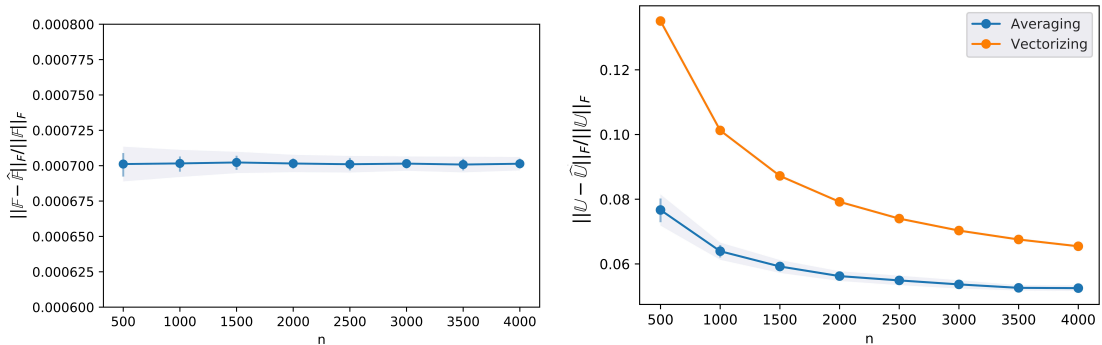

\centering
\begin{subfigure}[b]{0.45\textwidth}
\centering
\includegraphics[width=\textwidth]{figs/coefA7all_F_err_n_k24.pdf}
\end{subfigure}
\begin{subfigure}[b]{0.45\textwidth}
\centering
\includegraphics[width=\textwidth]{figs/coefA7all_U_err_n_k24.pdf}
\end{subfigure}
\caption{Performance of estimated $(\widehat{\FF},\widehat{\UU})$ when $n$ increases from $500$ to $4000$ with fixed $(p_1,p_2)=(70,50)$,  The experiments are repeated 100 times. Labels are the same as that for Figure \ref{fig:pFU_f2r2}. The solid lines are the means, the shadows show the standard deviation, and the vertical lines present the quartiles.
$\hat{\FF}$ is not affected by increasing number of $n$ as it uses $n'$ pre-training sample size, which is consistent with Theorem \ref{lem:Ferr_nuc}.}
\label{fig:nFU_f2r2}
\end{figure}

\medskip
\noindent
\textbf{FAMAR Coefficient Estimation and Prediction.}
Figures~\ref{fig:pFAMAR_f2r2} and \ref{fig:nFAMAR_f2r2} compare our proposed FAMAR with two current benchmarks, namely, the vanilla matrix regression model \eqref{eqn:mr} with nuclear penalty on $\bB$, and the oracle factor augmented regression model \eqref{eqn:famar} with observed $\bF$ and $\bU$, under {\sc Setting I.} and {\sc Setting II.}, respectively. 
We could clearly see that our proposed FAMAR estimation tracks the performances of the oracle factor augmented regression closely in almost all tasks. 
Our proposed method outperforms the existing matrix regression with nuclear norm penalty in both estimation of coefficient $\bB^*$ and prediction of new response $y$.
Additionally, by generating a low-rank $\hat\bB$, as illustrated in the lower-right plots of Figures~\ref{fig:pFAMAR_f2r2} and \ref{fig:nFAMAR_f2r2}, FAMAR offers improved interpretability and stability compared to direct regression on $\bX_i$'s. When the covariates $\bX_i$ are highly correlated, applying low-rank regression directly to $\bX_i$'s tends to overestimate the rank due to their dependency structure, which hampers the effectiveness of regularization and the interpretability of the model.

The superior performance of FAMAR stems from two key advantages. Firstly, since the space spanned by $(\bF_i, \bU_i)$ is identical to that spanned by $(\bF_i, \bX_i)$, FAMAR facilitates feature augmentation with $\bF_i$, effectively leveraging latent information. Secondly, the factor model decorrelates the covariance matrix, resulting in more stable and interpretable coefficient estimators.

\begin{figure}[ht!]
\centering
\begin{subfigure}[b]{0.45\textwidth}
\centering
\includegraphics[width=\textwidth]{figs/coefA7all_A_err_p_k22.pdf}
\end{subfigure}
\begin{subfigure}[b]{0.45\textwidth}
\centering
\includegraphics[width=\textwidth]{figs/coefA7all_B_err_p_k22.pdf}
\end{subfigure}
\hfill
\begin{subfigure}[b]{0.45\textwidth}
\centering
\includegraphics[width=\textwidth]{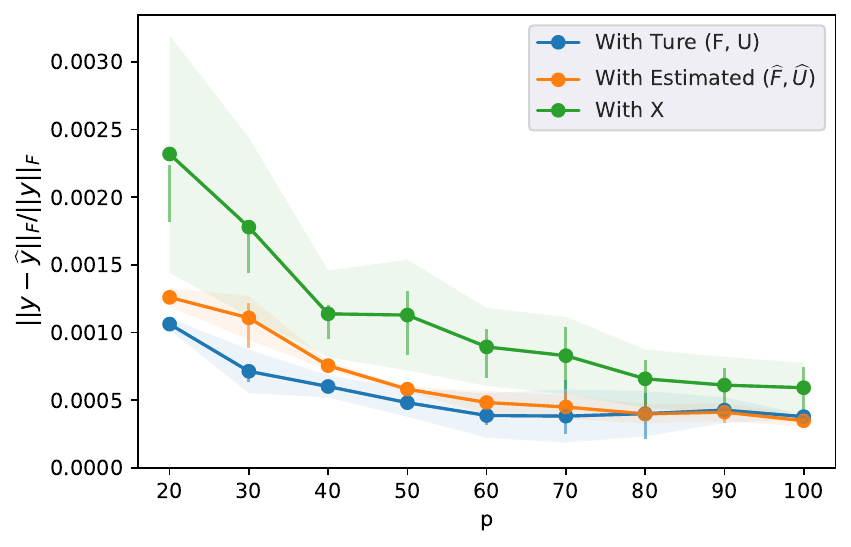}
\end{subfigure}
\begin{subfigure}[b]{0.45\textwidth}
\centering
\includegraphics[width=\textwidth]{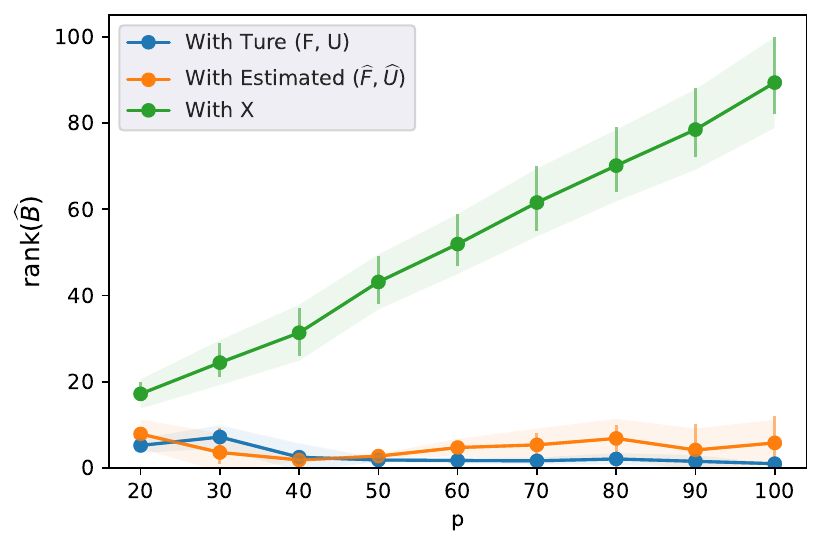}
\end{subfigure}
\hfill
\caption{
Fix $n=1000$, increase $p_1=p_2$ from $20$ to $100$. The blue, orange, and green lines correspond to running the regression with true $(\bF,\bU)$, estimated $(\widehat{\bF},\widehat{\bU})$, and $\bX$, respectively. 
The solid lines are the means over 100 repetitions, the shadows show the standard deviation, and the vertical lines present the quartiles. The lower-right figure shows the rank of $\hat\bB$ with $\lambda$ chosen by cross-validation under different models.}
\label{fig:pFAMAR_f2r2}
\end{figure}

\begin{figure}[htpb!]
\centering
\begin{subfigure}[b]{0.45\textwidth}
\centering
\includegraphics[width=\textwidth]{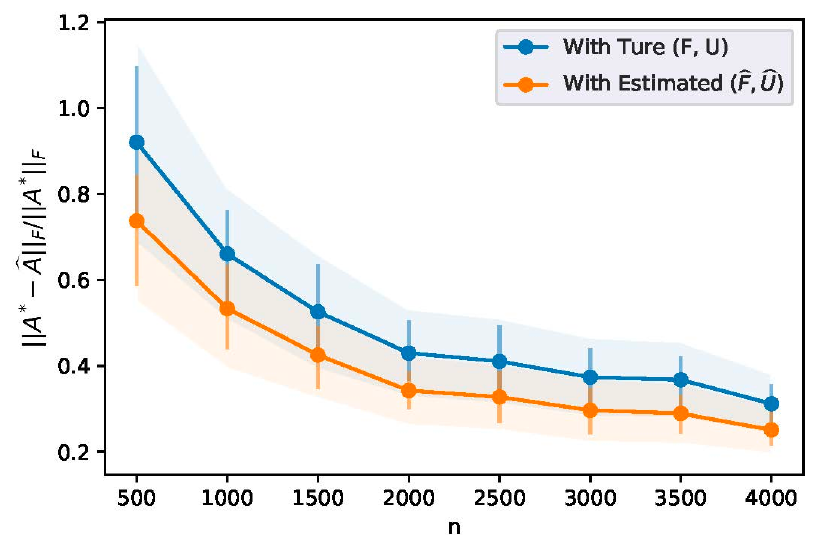}
\end{subfigure}
\begin{subfigure}[b]{0.45\textwidth}
\centering
\includegraphics[width=\textwidth]{figs/coefA7all_B_err_n_k24.pdf}
\end{subfigure}
\hfill
\begin{subfigure}[b]{0.45\textwidth}
\centering
\includegraphics[width=\textwidth]{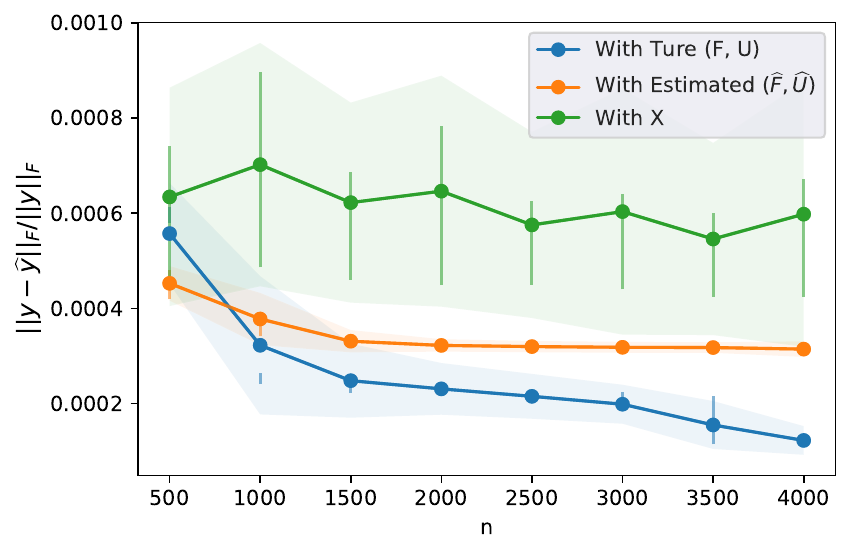}
\end{subfigure}
\begin{subfigure}[b]{0.45\textwidth}
\centering
\includegraphics[width=\textwidth]{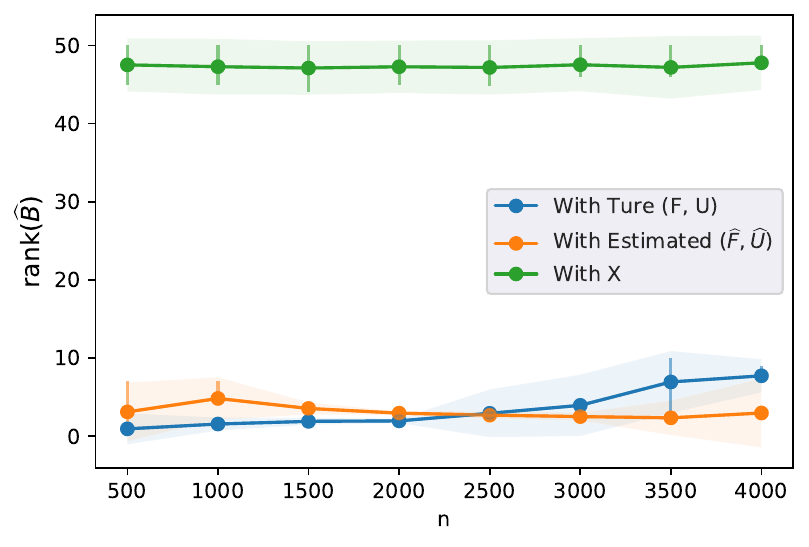}
\end{subfigure}
\hfill
\caption{Fix $(p_1,p_2)=(70,50)$, increase $n$ from $500$ to $4000$. The blue, orange, and green lines correspond to doing the regression with true $(\bF,\bU)$, estimated $(\widehat{\bF},\widehat{\bU})$, and $\bX$, respectively. The solid lines are the means over 100 repetitions, the shadows show the standard deviation, and the vertical lines present the quartiles.}
\label{fig:nFAMAR_f2r2}
\end{figure}

\section{Real Data Analysis} \label{sec:real}

We demonstrate the advantages of FAMAR with an real application to the multinational macroeconomic indices collected from OECD \citep{chen2019constrained}. 
The dataset contains 10 quarterly macroeconomic indexes of 14 countries from 1990.Q4 to 2016.Q4 for 105 quarters. 
By implementing the proposed FAMAR, we accomplish various empirical objectives, including the prediction of economic indicators, extraction of economic factors, and estimation of factor impacts. 
Subsequent sections showcase the results.

\medskip
\noindent
\textbf{Economic prediction.}
Here we showcase one-step ahead prediction of individual countries GDP using all the economic indicators from all countries in the last period. 
The covariate $\bX_i$ is the $10$ (index) $\times 14$ (country) matrix in one quarter and the response $y_i$ is the GDP of a country of interest in the next quarter. 
We take the estimated factor dimension as $(k_1,k_2)=(4,4)$ according to the analysis in \cite{chen2021statistical}.

We use a rolling window to do a one-step ahead prediction on the GDP of the United States (USA), United Kingdom (GBR), France (FRA), and Germany (DEU). 
We set the window width as 68, the first $n^\prime=32$ samples of which are set to pre-train the projection matrices $\bW_1$ and $\bW_2$, while the following $n=36$ samples are used to estimate the loading matrices $\bR^*$, $\bC^*$ and the regression coefficient matrix $\bB^*$. In each window, the covariate matrices $\bX_i$ are normalized by subtracting the mean and dividing the standard deviation of the training set. The response variable $y_i$ is also demeaned by subtracting the sample mean of each training set.

The out-of-sample $R^2$ of GDP prediction on some representative countries are reported in Table \ref{tab:oecd}.  
It is clear that FAMAR outperforms the vanilla matrix regression with a nuclear-norm penalty with much higher out-of-sample $R^2$. Besides the comparison with the baseline, we also present predictions based only on the estimated factors $\hat\bF_i$ or the estimated idiosyncratic components $\hat\bU_i$, respectively, to showcase their individual contributions to the process. As expected, it appears that factors play a predominant role in prediction, whereas the idiosyncratic components capture comparatively less valuable information in the context. This aligns with our expectation, for factors should capture most of the main and common features, and the residuals are regarded as idiosyncratic components corresponding to each feature, which are less representative and may vary a lot across time. 
\begin{table}[ht!]
\centering
\caption{Out-of-sample $R^2$ of predicting GDP}
\label{tab:oecd}
\begin{tabular}{l|c|c|c|c}
\hline
& FAMAR & $\hat\bF$ & $\hat\bU+\|\bB\|_*$  & $\bX+\|\bB\|_*$ \\ \hline
USA & 0.1540 & 0.1751 & -0.0512 &  -0.1066           \\ \hline
GBR & 0.2057 & 0.1727 & 0.0287 & -0.0403           \\ \hline
FRA & 0.1961 & 0.1540 & 0.0789 & -0.0411           \\ \hline
DEU & 0.1814 & 0.1460 & 0.1197 & 0.0050           \\ \hline
\end{tabular}
\end{table}

\medskip
\noindent
\textbf{Economic and Country Factors and Loadings.}
Figure \ref{RC_heatmap} shows the heatmaps for estimates of the row and column loading matrices for the $4\times 4$ factor model. 
The loading matrices are normalized so that the $\ell_2$ norm of each column is one. They are also varimax-rotated to reveal a clear structure.
From Figure \ref{RC_heatmap}(a) of the row factor loading matrix $\hat\bR$, we observed clearly that the rotated row factor 1-4 loaded heavily on central/western European, north American, Northern European, and Oceania countries, respectively.  
Accordingly, the four rotated row factors can be interpreted as the Central/Western European (C/W-EU), North American (NA), Northern European (N-EU), and Oceanian (OC) factors, mostly following geographical partitions.
In a similar fashion, the four rotated column factors are (i) Economic output and industrial production (GDP\&P); (ii) Consumer Price Index (CPI); (iii) International Trade (IT); (iv) Interest Rates (IR). Again, the factors agree with common economic knowledge.

\begin{figure}[ht!]
\centering
\begin{subfigure}[b]{0.48\textwidth}
\centering
\includegraphics[width=\textwidth]{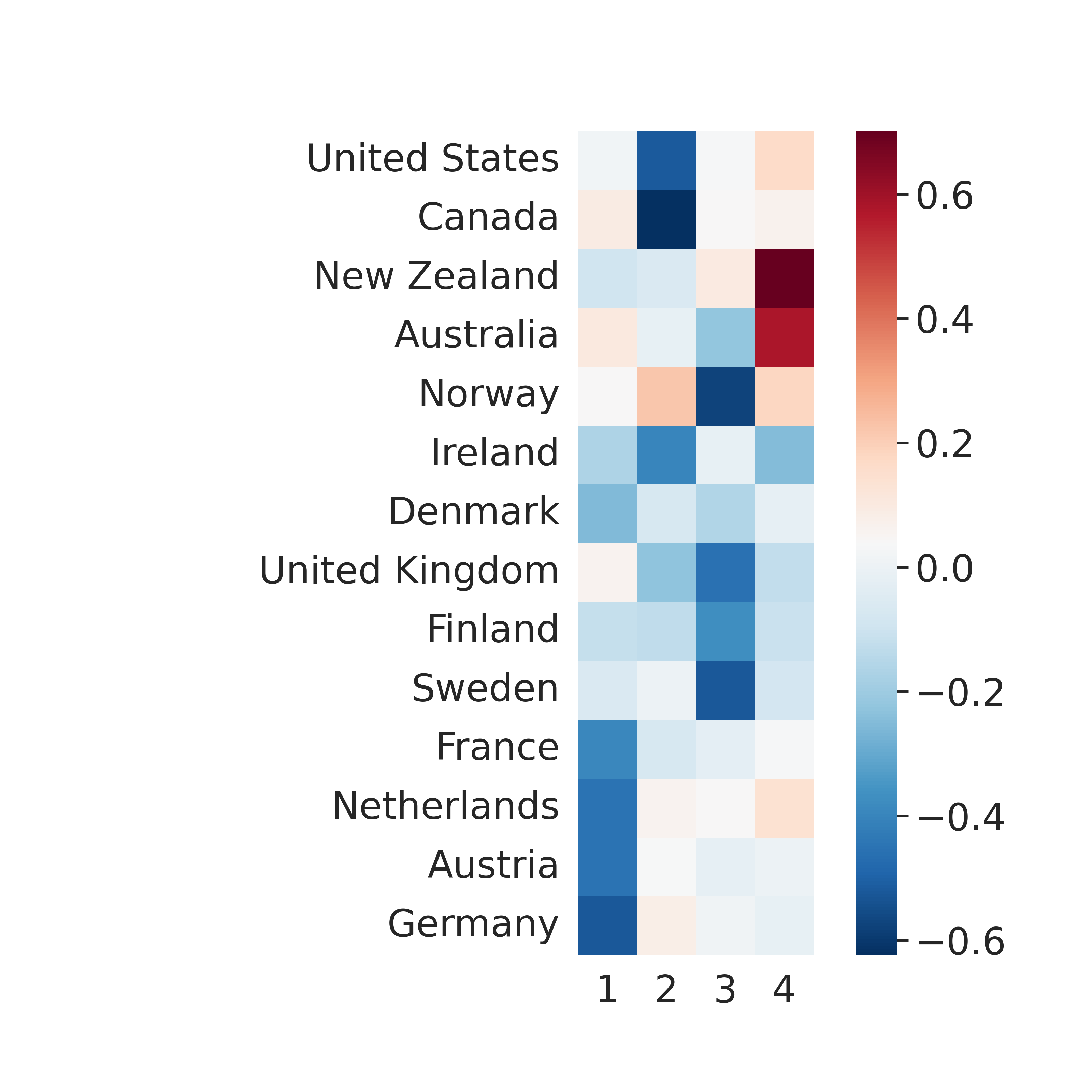}
\caption{$\hat\bR$ after varimax rotation}
\end{subfigure}
\begin{subfigure}[b]{0.48\textwidth}
\centering
\includegraphics[width=\textwidth]{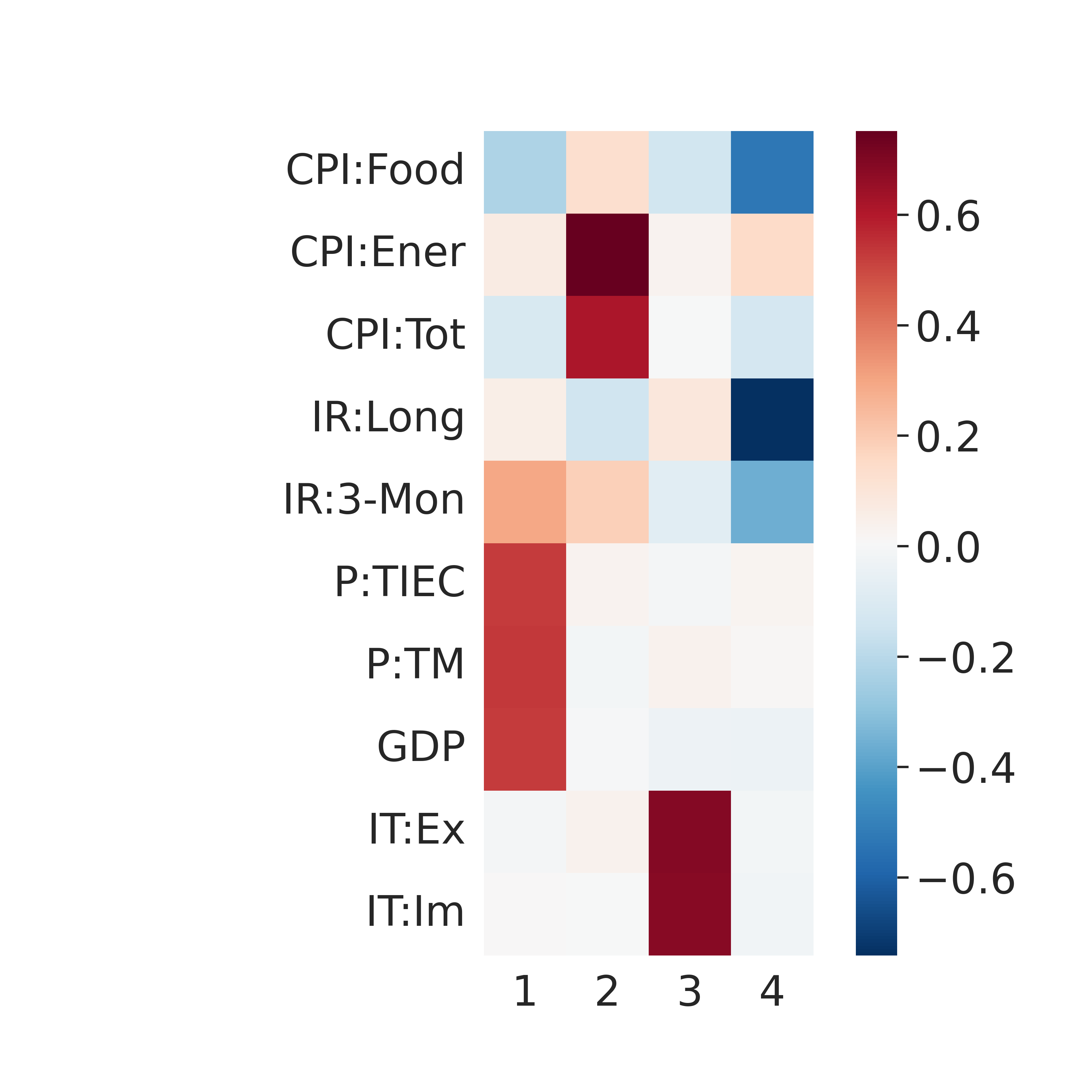}
\caption{$\hat\bC$ after varimax rotation}
\end{subfigure}
\hfill
\caption{Estimations of row and column loading matrices (varimax rotated) of matrix factor model for multinational macroeconomic indices.}
\label{RC_heatmap}
\end{figure}

\begin{figure}[ht!]
\centering
\begin{subfigure}[b]{0.45\textwidth}
\centering
\includegraphics[width=\textwidth]{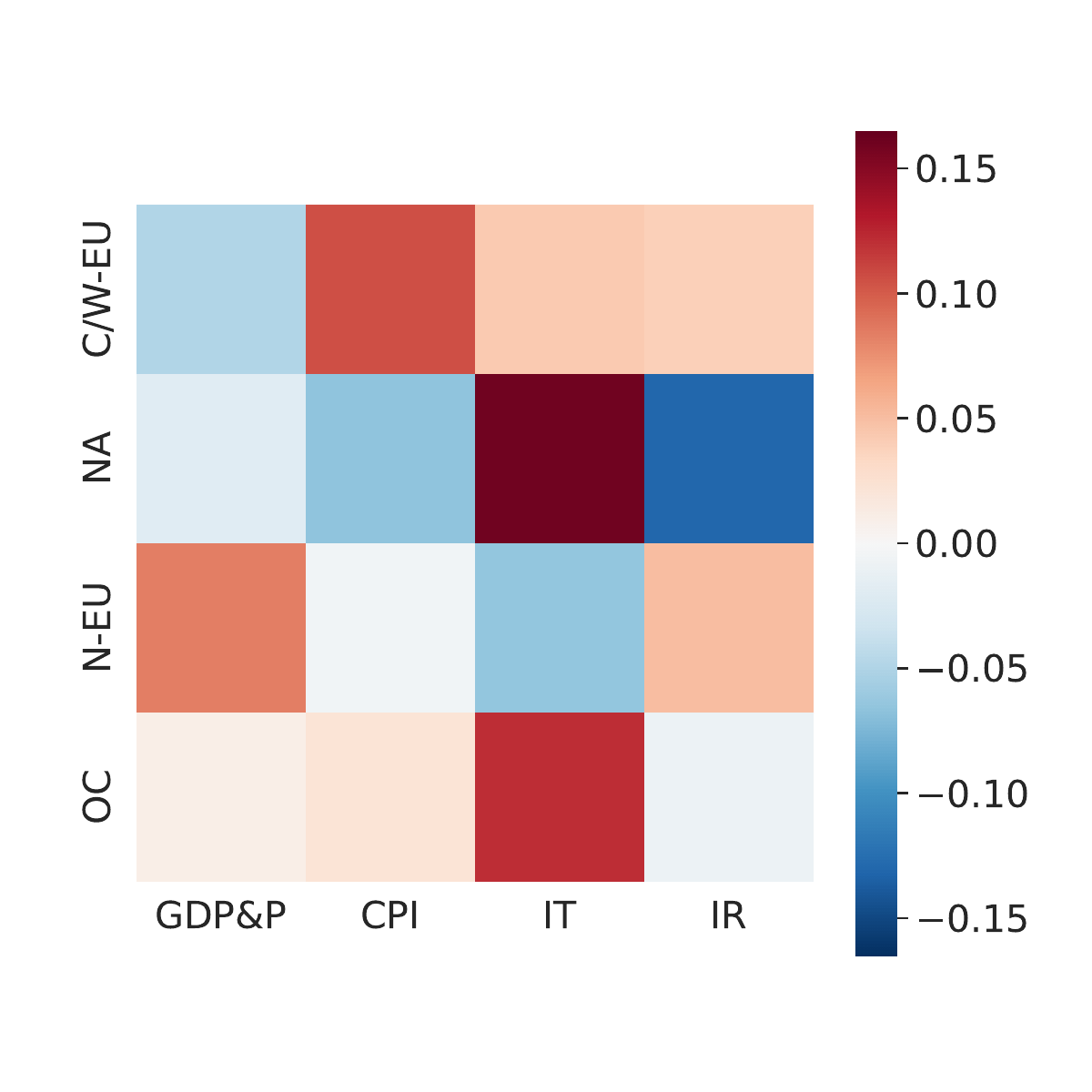}
\caption{USA}
\end{subfigure}
\begin{subfigure}[b]{0.45\textwidth}
\centering
\includegraphics[width=\textwidth]{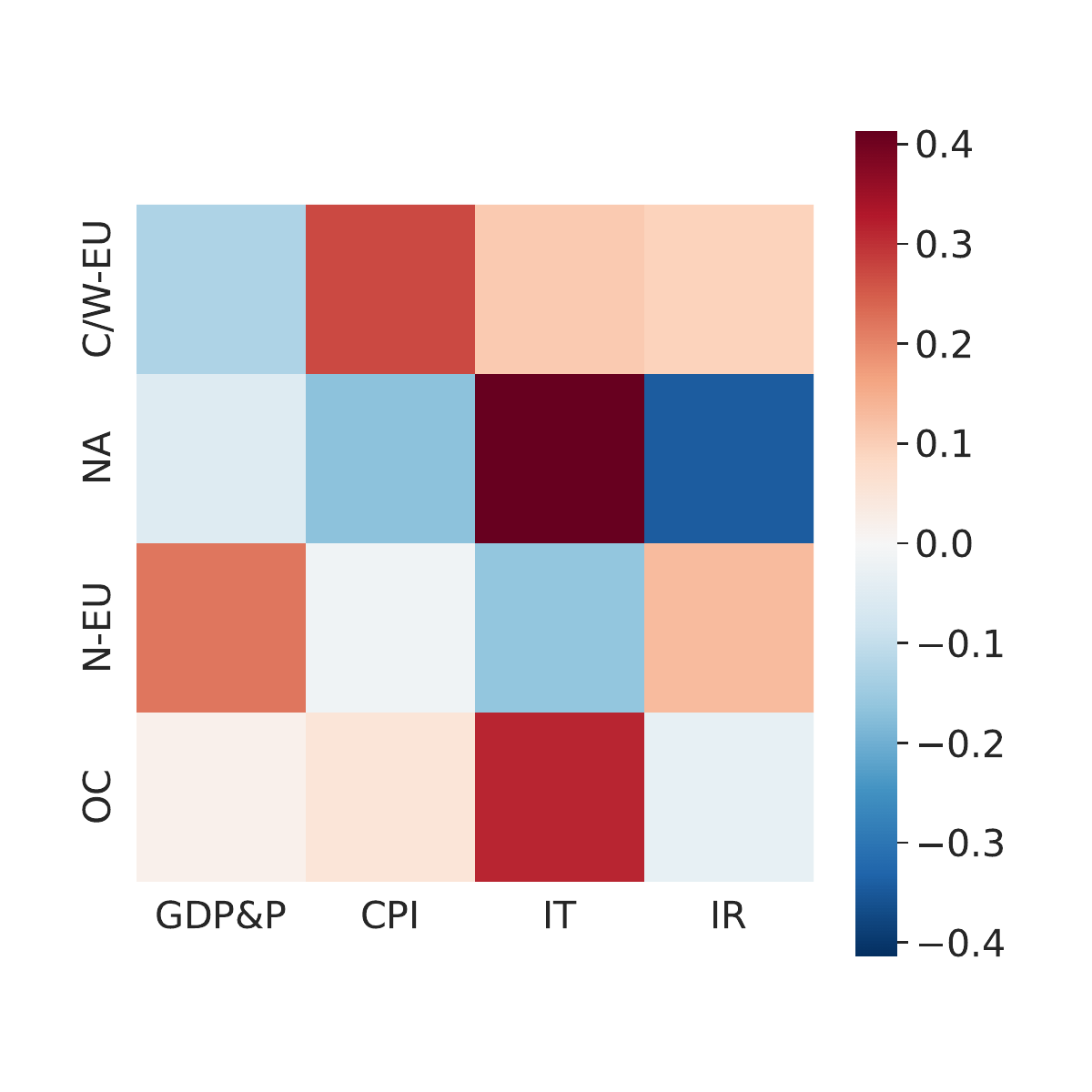}
\caption{GBR}
\end{subfigure}
\hfill
\begin{subfigure}[b]{0.45\textwidth}
\centering
\includegraphics[width=\textwidth]{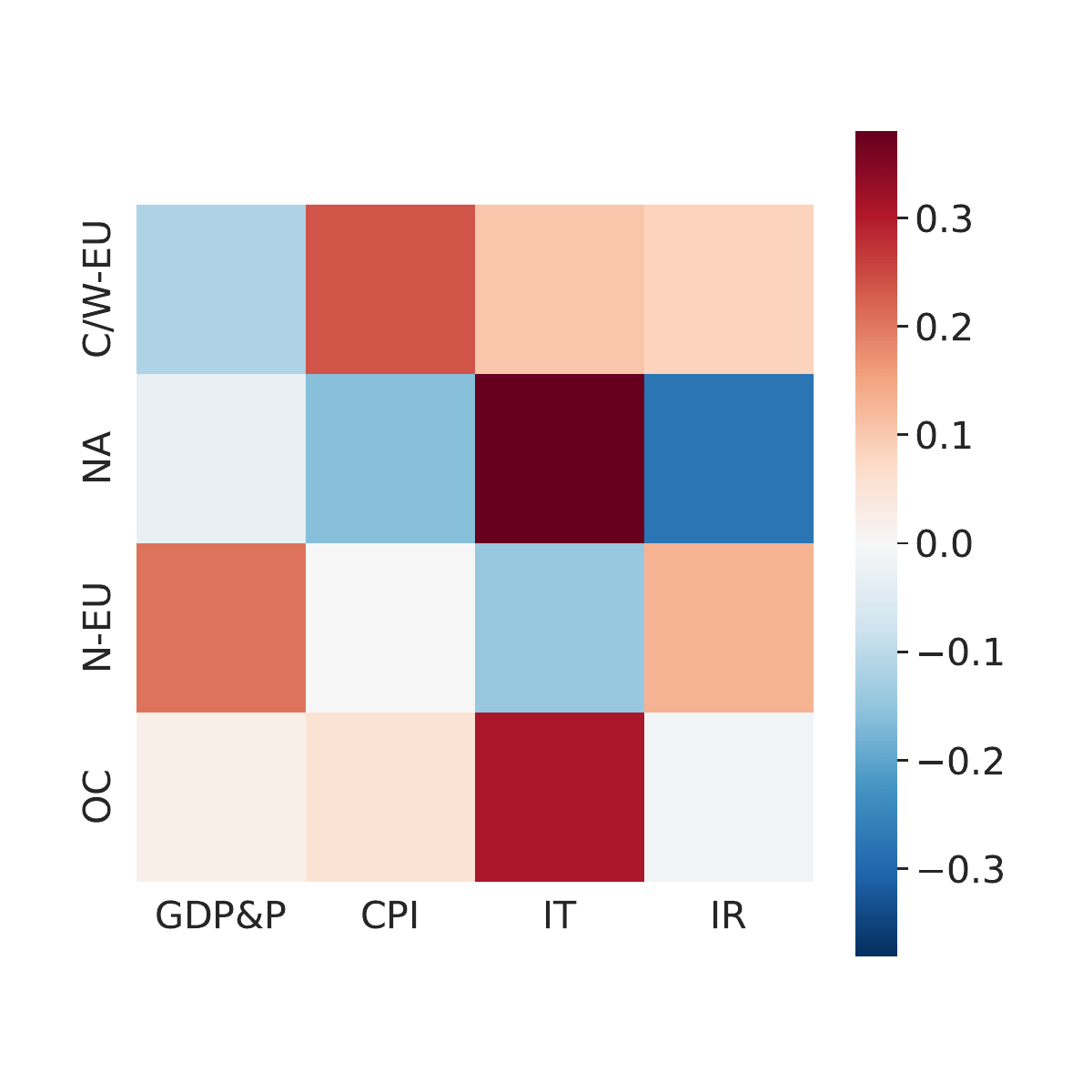}
\caption{FRA}
\end{subfigure}
\begin{subfigure}[b]{0.45\textwidth}
\centering
\includegraphics[width=\textwidth]{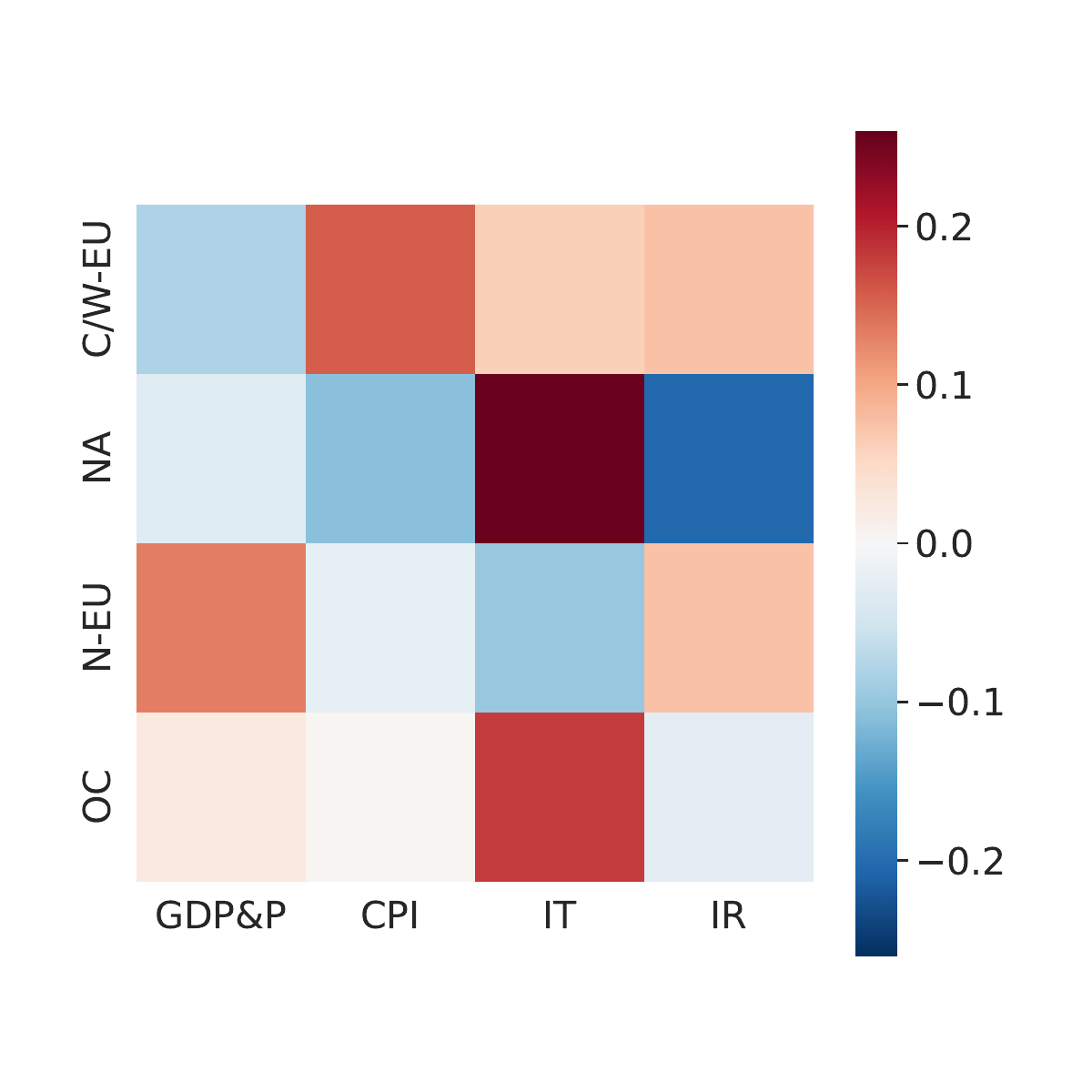}
\caption{DEU}
\end{subfigure}
\hfill
\caption{Heatmap of $\hat\bA$ corresponding to the  $\hat\bF_i$'s after rotation in Figure \ref{RC_heatmap}}
\label{A_heatmap}
\end{figure}

\medskip
\noindent
\textbf{Factors Impacts.}
To explore the impact of latent factors on the GDP of different countries, Figure \ref{A_heatmap} presents the heat map of $\hat\bA$ over the period. The heatmaps of $\hat\bA$ exhibit notable similarities across the four examined countries. A discernible pattern reveals several coordinates characterized by substantial magnitudes. Specifically, the four most significant coefficients, in terms of magnitude, correspond to the pairs: (IT, NA), (IR, NA), (IT, OC), and (CPI, C/W-EU). This implies that international trade and interest rates in North America, international trade in Oceania, and the CPI of Central and Western Europe play pivotal roles in forecasting the GDP of the countries.

\section{Conclusion}
\label{sec:conc}

We introduce Factor-Augmented Matrix Regression (FAMAR) as a powerful framework to tackle the growing complexity of matrix-variate data and its inherent high-dimensionality challenges. FAMAR comprises two pivotal algorithms, each designed to address distinct aspects of the problem. The first algorithm presents an innovative, non-iterative approach that efficiently estimates the factors and loadings within the matrix factor model. Leveraging advanced techniques such as pre-training, diverse projection, and block-wise averaging, it streamlines the estimation process. The second algorithm, equally crucial, offers an accelerated solution for nuclear-norm penalized matrix factor regression. Both algorithms are underpinned by robust statistical and numerical convergence guarantees, ensuring their reliability in practical applications.

Empirical assessments across a range of synthetic and real datasets substantiate the efficacy of FAMAR. Notably, it outperforms existing methods in terms of prediction accuracy, factor interpretability, and computational efficiency. FAMAR's versatility and superior performance make it a valuable addition to the toolkit of researchers and practitioners working with matrix-variate data. Its potential applications extend across diverse domains, from economics and finance to machine learning and data mining, where handling high-dimensional matrix data efficiently and effectively is paramount. As we move forward in the era of big data, FAMAR stands as a testament to the ongoing innovation in statistical modeling and data analysis, promising new avenues for extracting valuable insights from complex datasets.

%%---------------------------------------------------
%%
%% Bibliography
%%
\spacingset{1.18}
\bibliographystyle{agsm}
\bibliography{bib/main,bib/matreg,bib/optimization}

%%---------------------------------------------------
%%
%% Appendix in a seperate file
%%
\newpage
\setcounter{page}{1}
\begin{appendices}

\begin{center}
{\large\bf SUPPLEMENTARY MATERIAL of \\
``Factor Augmented Matrix Regression''}

\medskip

Elynn Y. Chen$^{\flat}$ \hspace{3ex}
Jianqing Fan$^\natural$ \hspace{3ex}
Xiaonan Zhu$^{\natural}$ \\ \normalsize

\smallskip
$^\flat$New York University \hspace{3ex}
$^{\natural}$ Princeton University
\end{center}

\section{Notations}

We use the notation $[N]$ to refer to the positive integer set $\braces{1, \ldots, N}$ for $N \in \ZZ_+$.
For any matrix $\bX$,  we use $\bx_{i \cdot}$, $\bx_{j}$, and $x_{ij}$ to refer to its $i$-th row, $j$-th column, and $ij$-th entry, respectively. 
All vectors are column vectors and row vectors are written as $\bx^\top$ for any vector $\bx$. 
For any two scalers $p$, $q$, we denote $p\wedge q := \min\{p,q\}$ and $p\vee q := \max\{p,q\}$.

We use $\vect(\bX)$ as the vectorization of $\bX$ and use $\XX$ to denote a matrix whose rows are specifically vectorized matrices. That is, the $i$-th row of $\XX$ is $\vect(\bX_i)^\top$. For two matrices $\bX_1\in\RR^{m\times n}$ and $\bX_2\in\RR^{p\times q}$, $\bX_1\otimes\bX_2\in\RR^{pm\,\times\,qn}$ is the Kronecker product. Note that $\vect(\bA\bX\bB)=(\bB^\top\otimes\bA)\vect(\bX)$. 
We denote $\diag(\bX_1,\bX_2)\in\RR^{(m+p)\times(n+q)}$ as the block-diagonal matrix with $\bX_1$ and $\bX_2$ being its diagonal blocks.

For any vector $\bx=(x_1,\ldots,x_p)^\top$, we let $\|\bx\|:= \|\bx\|_2=(\sum_{i=1}^p x_i^2)^{1/2}$ be the $\ell_2$-norm, and let $\|\bx\|_1=\sum_{i=1}^p |x_i|$ be the $\ell_1$-norm. 
Besides, we use the following matrix norms: 
$\ell_2$-norm $\norm{\bX}_2 := \sigma_{max}(\bX)$; $(2,1)$-norm $\norm{\bX}_{2,1} := \underset{\norm{\ba}_1=1}{\max} \norm{\bX\ba}_2 =\max_i \norm{\bx_{i}}_2$;
Frobenius norm $\|\bX\|_F = (\sum_{i,j}x_{ij}^2)^{1/2}$;
nuclear norm $\|\bX\|_*=\sum_{i=1}^n {\sigma_{i}(\bX)}$.
When $\bX$ is a square matrix, we denote by $\Tr \paren{\bX}$, $\lambda_{max} \paren{\bX}$, and $\lambda_{min} \paren{\bX}$ the trace, maximum and minimum singular value of $\bX$, respectively.
For two matrices of the same dimension, define the inner product $\angles{\bX_1,\bX_2} = \Tr(\bX_1^\top \bX_2)$.

\section{Computation}

\subsection{Explanation of the Accelerated Algorithm}
To elucidate the accelerated algorithm, we first consider the vanilla gradient algorithm with the gradient step performed at the approximate solution $\{(\bA^{(k)}, \bB^{(k)})\}$, which consists of a singular value thresholding (SVT) on the gradient step of estimating $\bB^*$ and a regular update of $\bA^*$. 
The partial gradient of $f_n$ with respect to $\bB$ is
\[
\nabla_B f_n(\bA, \bB) = -\frac{1}{n}\sum_{i=1}^n \paren{y_i-\angles{\bA,\bF_i}-\angles{\bB,\bU_i}} \bU_i.
\]
Both $\nabla_A f_n(\bA, \bB)$ and $\nabla_B f_n(\bA, \bB)$ are $L$-Lipschitz with given $\{(\bF_i,\bU_i)\}_{i=1}^n$.
Given a prescribed precision level $\epsilon$ and some initialization $\bA^{(0)}$ and $\bB^{(0)}$, we can estimate $\bA^*$ and $\bB^*$ alternatively. 
Given $\bA^{(k-1)}$, $\bB^{(k-1)}$ and the step size $L^{(k)}$, $\bB$ can be updated by:
\begin{equation}\label{equ_grad}
\bB^{(k)} = \argmin_{\bB} \braces{\frac{L^{(k)}}{2}\norm{\bB - \paren{\bB^{(k-1)}-\frac{1}{L^{(k)}}\nabla_B f_n(\bA^{(k-1)}, \bB^{(k-1)})}}_F^2 + \lambda_n\norm{\bB}_*},
\end{equation}    
which admits an analytic solution
\begin{equation}\label{eqn:updateB}
\bB^{(k)} = \cT_{\lam_n/L^{(k)}}\paren{\bB^{(k-1)}-\frac{1}{L^{(k)}}\nabla_B f_n(\bA^{(k-1)}, \bB^{(k-1)})}
\defeq p_{L^{(k)}}(\bA^{(k-1)}, \bB^{(k-1)}),
\end{equation}
where $\cT_{\lam}(\bC)$ is an operator defined as $\cT_{\lam}(\bC)=\bU\bD_{\lam}\bV^{\top}$ on a matrix $\bC$ with SVD $\bC=\bU\bD\bV^{\top}$ and $\bD_{\lam}$ is diagonal with the soft thresholding $(d_{\lam})_{ii} = \sign(d_{ii})\max\{0,\abs{d_{ii}} - \lam\}$.
Since $\bA$ is of low dimension, it can be updated directly by gradient descent:
\begin{equation}\label{eqn:updateA}
\bA^{(k)} = q_{L^{(k)}}(\bA^{(k-1)}, \bB^{(k-1)}) \defeq \bA^{(k-1)}-\frac{1}{L^{(k)}}\nabla_A f_n(\bA^{(k-1)}, \bB^{(k-1)}).
\end{equation}
The accelerated algorithm performs the gradient step at the search point $\{(\bQ^{(k)}, \bP^{(k)})\}$. 
So $\bA^{(k)}$ and $\bB^{(k)}$ are calculated according to equation \eqref{eqn:updateB} and \eqref{eqn:updateA} with $(\bA^{(k-1)}, \bB^{(k-1)})$ replaced by the search point $(\bQ^{(k)},\bP^{(k)})$.

\subsection{Numerical Convergence}

\begin{lemma} \label{thm:update-B}
Given $\bA^{(k-1)}$, $\bB^{(k-1)}$ and the step size $L^{(k)}$, $\bB$ can be updated by:
\begin{equation}%\label{equ_grad}
\bB^{(k)} = \argmin_{\bB} \braces{\frac{L^{(k)}}{2}\norm{\bB - \paren{\bB^{(k-1)}-\frac{1}{L^{(k)}}\nabla_B f_n(\bA^{(k-1)}, \bB^{(k-1)})}}_F^2 + \lambda_n\norm{\bB}_*},
\end{equation}    
or, equivalently, 
{
\begin{equation}%\label{eqn:updateB}
\bB^{(k)} = p_{L^{(k)}}(\bA^{(k-1)}, \bB^{(k-1)}) \defeq \cT_{\lam_n/L^{(k)}}\paren{\bB^{(k-1)}-\frac{1}{L^{(k)}}\nabla_B f_n(\bA^{(k-1)}, \bB^{(k-1)})},
\end{equation}
}
where $\cT_{\lam}(\bC)$ is an operator defined as $\cT_{\lam}(\bC)=\bU\bD_{\lam}\bV^{\top}$ on a matrix $\bC$ with SVD $\bC=\bU\bD\bV^{\top}$ and $\bD_{\lam}$ is diagonal with $(d_{\lam})_{ii} = \sign(d_{ii})\max\{0,\abs{d_{ii}} - \lam\}$.
\end{lemma}
\begin{proof}[Proof of Lemma \ref{thm:update-B}]
Lemma \ref{thm:update-B} can be prove by invoking Theorem 2.1 in \citep{cai2010singular}. 
\end{proof}

\begin{lemma}\label{thm:convergence}
Let
\[
p_{L}(\bA, \bB) = \argmin_{\bB^\prime} \big\{H_{L,B}(\bB^\prime;\bA,\bB) + \lambda_n\|\bB^\prime\|_{*} \big\},
\quad
q_{L}(\bA, \bB) = \argmin_{\bA^\prime} H_{L,A}(\bA^\prime;\bA,\bB), 
\]
where 
\begin{align}
H_{L,A}(\bA'; \bA,\bB) &\defeq \angles{\bA'-\bA,\nabla_A f_n(\bA,\bB)}+\frac{L}{2}\|\bA'-\bA\|_F^2, \\
H_{L,B}(\bB'; \bA,\bB) &\defeq \angles{\bB'-\bB,\nabla_B f_n(\bA,\bB)}+\frac{L}{2}\|\bB'-\bB\|_F^2. 
\end{align}
Assume the following inequality holds:
\vspace{0.18in}
\begin{equation}
\begin{aligned}
F_n\big(p_{L}(\bA,\bB), p_{L}(\bA,\bB)\big) \geq & f_n(\bA,\bB) + H_{L,B}(p_{L}(\bA,\bB);\bA,\bB) \\
&+ H_{L,A}(q_{L}(\bA,\bB);\bA,\bB) + \lambda_n\| p_{L}(\bA,\bB)\|_{*}.
\end{aligned}
\end{equation}
Then for any $\bB^\prime\in\RR^{p_1\times p_2}$ and $\bA^\prime\in\RR^{k_1\times k_2}$, we have
\begin{equation}
\begin{aligned}
F_n\big(\bA^\prime,\bB^\prime\big)  - F_n\big(q_{L}(\bA,\bB), & \  p_{L}(\bA,\bB) \big) \geq  \frac{L}{2} \big( \|p_{L}(\bA,\bB)-\bB\|_F^2 +  \|q_{L}(\bA,\bB)-\bA\|_F^2\big)\\
&+ L\big(\angles{\bB-\bB^\prime,p_{L}(\bA,\bB)-\bB} + \angles{\bA-\bA^\prime,q_{L}(\bA,\bB)-\bA}\big).
\end{aligned}
\end{equation}
Further, let $\{(\bA^{(k)},\bB^{(k)})\}_{k\geq 1}$ be the sequences generated by Algorithm \ref{alg:accelarated}. Then for any $k\geq 1$, we have
\begin{equation}
F_n(\bA^{(k)},\bB^{(k)}) - F_n(\bA^*,\bB^*) \leq \frac{2\gamma L \big(\|\bA^{(0)}-\bA^*\|_F^2 +\|\bB^{(0)}-\bB^*\|_F^2\big)}{(k+1)^2}.
\end{equation}
\end{lemma}

\begin{proof}
{
The $p_{L}(\bA,\bB)$ and $q_{L}(\bA,\bB)$ defined in Lemma \ref{thm:convergence} is equivalent to that defined in Section \ref{sec:comp}, which can be shown using Lemma \ref{thm:update-B} and \cite{bertsekas1999nonlinear}, see also equation (5)-(7) in \cite{ji2009accelerated}.
}
The global convergence rate of Algorithm \ref{alg:accelarated} is established by Lemma \ref{thm:convergence}, whose proves are direct extensions of Theorem 4.1 in \citet{ji2009accelerated} and are therefore omitted.
\end{proof}

\section{Technical Assumptions} \label{append:assump}

The following assumptions are needed for technical purposes. 
{Assumption \ref{asum:FU_subG} and \ref{asu:F} imposes conditions on factors and idiosyncratic components, still allowing weak correlations.} 
Assumption \ref{asu:H} is already satisfied by the data-driven method we have studied by its construction and by Proposition \ref{prop:H}. Nevertheless, since some other reasonable projection matrices can also be used, we impose Assumption \ref{asu:H} for genreality.
Assumption \ref{asu:infer} is needed for asymptotic normality. 

\begin{assumption}\label{asum:FU_subG} 
For MFM \eqref{eqn:famar}, there exists universal constants $c_0$ and $c_1$ such that
\begin{enumerate}[label={(\alph*)}]
\item The elements of $\bF_i\in \RR^{k_1\times k_2}$ are mean zero, sub-Gaussian, and weakly correlated in the sense that: $\bF_i = \bSigma_{k_1} \bZ_i \bSigma_{k_2}$, where $\bZ_i \in \RR^{k_1\times k_2}$ has independent mean zero sub-Gaussian entries, and for $\bSigma_{k_1}\in\RR^{k_1\times k_1}$, $\bSigma_{k_2}\in \RR^{k_2\times k_2}$, assume that $0 < \|\bSigma_{k_j}\|_2 <c_0$ for $j=1,2$.
\item The elements of $\bU_i\in \RR^{p_1\times p_2}$ are mean zero, sub-Gaussian, and weakly correlated in the same sense as above. Furthermore, assume that 
\[
\begin{aligned}
\text{for }  i,k\in[p_1], i\neq k, \EE\left[\sum_{j=1}^{p_2}u_{ij}u_{kj}\right]\leq c_1\sqrt{\frac{p_2}{p_1}};\\
\text{and for }  i,k\in[p_2], i\neq k, \EE\left[\sum_{j=1}^{p_1}u_{ji}u_{jk}\right]\leq c_1\sqrt{\frac{p_1}{p_2}}.
\end{aligned}
\]
\end{enumerate}
\end{assumption}

\begin{assumption} \label{asu:F}
There exist universal constants $c_1, c_2$, and $C_2$ such that
\begin{enumerate}[label={(\alph*)}]
\item  $\frac{1}{n} \sum_{i=1}^n \sum_{j=1}^n \mathbb{E}\|\bF_j\|_F\|\bF_i\|_F \Tr\left(\mathbb{E}\left[\bU_j \bU_i^{\top} \mid \bF\right]\right)<c_1.$
\item $c_2 < \lambda_{\min}(\FF^\top\FF/n) < \lambda_{\max}(\FF^\top\FF/n) <C_2$, and $\sum_{i=1}^n \tilde\bff_i \tilde\bff_i^\top \xrightarrow{P} \bSigma_F$ for some $k_1\!k_2\times k_1\!k_2$ positive definite matrix $\bSigma_F$, where $\tilde{\bff}_i := \vect(\bF_i)$ and $\FF=\big(\tilde{\bff}_1\cdots\tilde{\bff}_n\big)^\top$.
\end{enumerate}
\end{assumption}

\begin{assumption}\label{asu:H}
There exist universal constants $c_1, c_2$, $C_1$, and $C_2$ such that
\begin{enumerate}[label={(\alph*)}]
\item $\nu_{\max}(\bH_i)\leq C_i\nu_{\min}(\bH_i)$, for $i=1,2$.
\item $\lambda_{\min}(\bW_i^\top\bW_i/p_i)\geq c_i$, for $i=1,2$.
\end{enumerate}
\end{assumption}

\begin{assumption}\label{asu:infer}
\begin{enumerate}[label={(\alph*)}]
\item There exists a constant $M<\infty$ such that the $k_1\!k_2\times k_1\!k_2$ matrix satisfies
\[
\EE\left\| \frac{1}{np_1\!p_2}\sum_{i=1}^n \sum_{j=1}^{p_1\!p_2} [\bW_2^\top\otimes\bW_1^\top]_j \tilde{\bff}_i^\top \tilde{u}_{ij}\right\|^2\leq M,
\]
where $[\bW_2^\top\otimes\bW_1^\top]_j$ is the $j$-th column of $(\bW_2^\top\otimes\bW_1^\top)$, $\tilde\bff_i^\top$ is the $i$-th row of $\FF$, and $ \tilde{u}_{ij}$ is the $(i,j)$-th element of $\UU$. 
\item There exists a full-rank $k_1\!k_2\times k_1k_2$ matrix $\bLambda_H$ such that
\[
\plim_{p_1,p_2\rightarrow \infty}  (\bH_2\otimes\bH_1) = \bLambda_H.
\]
This also implies $\nu_{\min}(\bH_i), \nu_{\max}(\bH_i)=\Omega_p(1)$, for $i=1,2$.
\item There exists an $k_1\!k_2\times k_1\!k_2$ matrix $\bSigma_U$ such that
\[
\lim_{p_1,p_2\rightarrow \infty} \frac{1}{p_1\!p_2}  \sum_{j=1}^{p_1\!p_2} \sum_{k=1}^{p_1\!p_2} [\bW_2^\top\otimes\bW_1^\top]_j [\bW_2^\top\otimes\bW_1^\top]_k^\top \EE[\tilde{u}_{ij}\tilde{u}_{ik}] = \bSigma_U, \text{ for } i\in [n].
\]
\item Assume that there exist a $k_1\!k_2\times k_1\!k_2$ matrix $\bPhi$ and a $p_1\!p_2\times p_1\!p_2$ matrix $\bPsi$ such that
\[
\frac{1}{\sqrt{n}} \sum_{i=1}^n \tilde{\bff}_i \tilde{\bu}_i^\top \xrightarrow[]{ d } \cM\cN(\bfm 0, \bPhi, \bPsi),
\]
where $\tilde{\bu}_i=\vect(\bU_i)$, $\bPsi$ is divided into $p_2\times p_2$ blocks, denoted as $\bPsi_{ij}\in\RR^{p_1\times p_1}$, for $i,j\in[p_2]$, and that $\bPsi_{ij}=\op{p^{-1}}$ for all $i\neq j$.
\item There exists $M\leq \infty$ such that for all $n$, $p_1,$ and $p_2$, $\EE \tilde{u}_{ij}=0$, $\EE |\tilde{u}_{ij}|^8 \leq M$.
\end{enumerate}
\end{assumption}

\section{Theoretical Analysis of ``Pre-trained Projection'' and ``Block-wise Averaging''} \label{append:mfm-proofs}

\subsection{Proof of Proposition \ref{prop:H}}
We denote $\bS^*_1=\frac{1}{p_1}\bR^*{\bR^*}^\top$, $\bS^*_2=\frac{1}{p_2}\bC^*{\bC^*}^\top$, $\hat{\bS}_1=\frac{1}{p_1p_2n}\sum_{i=1}^n{\bX_i\bX_i^\top}$, and $\hat\bS_2=\frac{1}{p_1p_2n}\sum_{i=1}^n{\bX_i^\top\bX_i}$.
We first present a technical lemma that bounds the $\ell_2$ norm between $\hat{\bS}_1$ (resp. $\hat\bS_2$) and $\bS^*_1$ (resp. $\bS^*_2$). Its proof is provided at the end of this section. 
\begin{lemma}\label{lem:M}
Under Assumption \ref{asum:FU} and \ref{asum:FU_subG}, there exists a universal constant $c$ such that
\[
\left\|\hat{\bS}_1 - \bS^*_1\right\|_{2} \leq c\left(\sqrt{t(k_1\vee k_2)}\sqrt{\frac{\log k_1 + t}{n}}  + \frac{1}{\sqrt{p_1}} \right),
\]
and 
\[
\left\|\hat\bS_2 - \bS^*_2\right\|_{2} \leq c\left(\sqrt{t(k_1\vee k_2)}\sqrt{\frac{\log k_2 + t}{n}} + \frac{1}{\sqrt{p_2}} \right)
\]
with probability at least $1-11e^{-t}$
\end{lemma}

\begin{proof}[Proof of Proposition \ref{prop:H}]
Since the proof for $\bH_1$ and $\bH_2$ are the same, we only provide that for $\bH_1$ as follows.  The argument below is deterministic, conditioned on the event in Lemma~\ref{lem:M}.

Let $\hat{\bZ}\in\RR^{p_1\times k_1}$ be the matrix consists of the top $k_1$ top eigenvectors of $\hat{\bS}_1$. Then $\bW_1=\sqrt{p_1} \hat{\bZ}$. Since $\bS^*_1$ is symmetric, it has the eigen decomposition $\bS^*_1 = \bZ^*\bLam^*{\bZ^*}^\top$, where $\bLam^*$ is a $k_1\times k_1$ diagonal matrix, and $\bZ^*\in\RR^{p_1\times k_1}$ satisfies ${\bZ^*}^\top\bZ^*=\bI_{k_1}$. Hence, $\bR^*=\bZ^*{\bLam^*}^{1/2}$.
The upper bound is derived from
\[
\begin{aligned}
\nu_{\max}\left(\frac{1}{p_1}\bW_1^\top \bR^*\right)& = \sup_{\bu\in\SS^{k_1-1}}\left\|\frac{1}{p_1}{\bR^*}^\top \bW_1\bu\right\|_2
\leq \frac{1}{\sqrt{p_1}} \left\|{\bLam^*}^{1/2}\right\|\sup_{\bu\in\SS^{k_1-1}} \left\|{\bZ^*}^\top\hat{\bZ}\bu\right\|_2 \\
&\leq C_1 \sup_{\bu\in\SS^{k_1-1}}\left\|{\bZ^*}^\top\hat{\bZ}\bu\right\|_2
\leq C_1 \sup_{\bu\in\SS^{k_1-1}}\left\|\bu\right\|_2 =C_1
\end{aligned}
\]
for some universal positive constant $C_1$, where the second inequality is because of Assumption \ref{asm:RC}.

As for the lower bound, we have
\begin{equation}\label{eqn:H_lower}
\nu_{\min}\left(\frac{1}{p_1}\bW_1^\top\bR^*\right) \geq \frac{1}{\sqrt{p_1}} \nu_{\min} \left(\hat{\bZ}\bZ^*\right) \nu_{\min}\left({\bLam^*}^{1/2}\right) \geq c_1\nu_{\min} \left(\hat{\bZ}\bZ^*\right),
\end{equation}
where the last inequality follows from Assumption \ref{asm:RC}.
Let $\lambda_1,\ldots,\lambda_{p_1}$ be the eigenvalues of $\bS^*_1$ in the non-decreasing order, and let $\hat\lambda_1,\ldots,\hat\lambda_{p_1}$ be the eigenvalues of $\hat{\bS}_1$ in the non-decreasing order. Then by Assumption \ref{asm:RC} and that $\rank(\bS)=k_1$, we have \
\begin{equation}\label{eqn:lambda_temp2}
\lambda_{k_1}\geq c_1 \text{ for some constant } c_1, \text{ and } \lambda_{k_1+1} =0.
\end{equation}
Let $\delta := \sqrt{t(k_1\vee k_2)}\sqrt{\frac{\log k_1 + t}{n}} + \frac{1}{\sqrt{p_1}}$.
By Weyl's Theorem and Lemma \ref{lem:M}, there exists a universal constant $c_2$ such that
\begin{equation}\label{eqn:lambda_temp3}
\left|\hat{\lambda}_i - \lambda_i\right| \leq \left\|\hat{\bS}_1 - \bS^*_1\right\|_2 \leq c_2\delta, \text{ for all } i\in[p_1].
\end{equation}
with probability at least $1-11e^{-t}$.

We argue that
\[
\nu_{\min} \left(\hat{\bZ}\bZ^*\right) \geq \sqrt{1-\left(\frac{2c_2}{c_1}\delta\right)^2} \geq 1-\frac{2c_2}{c_1}\delta.
\]
It suffices to prove the above inequality for $\delta<c_1/(2c_2)$; otherwise, the inequality is trivial.
In this case, combining \eqref{eqn:lambda_temp2} and  \eqref{eqn:lambda_temp3}, we bound the eigen-gap between $\hat{\bS}_1$ and $\bS^*_1$ by
\[
\begin{aligned}
\tilde{\delta} :=& \inf\left\{|\lambda-\hat\lambda| \mid \lambda\in[\lambda_{k_1},\lambda_{1}], \hat\lambda\in(-\infty,\hat\lambda_{k_1+1}]\right\} 
={ { \lambda_{k_1}- \hat\lambda_{k_1+1} }}\\
=& {\lambda_{k_1} -\lambda_{k_1+1} + \lambda_{k_1+1}  -\hat\lambda_{k_1+1}}
\geq c_1 - |\lambda_{k_1+1}  - \hat\lambda_{k_1+1}| \geq c_1 -c_2 \delta \geq \frac{c_1}{2}.
\end{aligned}
\]
Davis-Kahan $\sin\Theta$ theorem shows that
\[
\|\sin\Theta(\hat{\bZ}, \bZ^*)\|_2 \leq \frac{\|\hat{\bS}_1-\bS^*_1\|_2}{\tilde\delta} \leq \frac{2c_2}{c_1}\delta.
\]
Together with the fact that
\[
\left[\sin\Theta(\hat{\bZ}, \bZ^*)\right]^2 + \left[\cos\Theta(\hat{\bZ}, \bZ^*)\right]^2 = \bI_{k_1},
\]
we have
\[
\nu_{\min}^2 \left(\hat{\bZ}\bZ^*\right) \geq 1- \|\sin\Theta(\hat{\bZ}, \bZ^*)\|_2^2
% \geq 1- \|\sin\Theta(\hat{\bZ}, \bZ^*)\|_F^2 
\geq 1-\left(\frac{2c_2}{c_1}\delta\right)^2.
\]
Hence,
\[
\nu_{\min} \left(\hat{\bZ}\bZ^*\right) \geq \sqrt{1-\left(\frac{2c_2}{c_1}\delta\right)^2} \geq 1-\frac{2c_2}{c_1}\delta.
\]
Substituting this to \eqref{eqn:H_lower} completes the proof.
\end{proof}
{%\blue
To prove Lemma \ref{lem:M}, we first show the following Lemma.
\begin{lemma}\label{lem:F_op}
Under Assumption \ref{asum:FU_subG} (a), there exists a universal constant $C$, such that
\[
\PP\left( \|\bF_i\| \geq C(\sqrt{k_1} + \sqrt{k_2} +\sqrt{t}) \right) \leq 2 \exp(-t).
\]
\end{lemma}
\begin{proof}[Proof of Lemma \ref{lem:F_op}]
We can find an $\varepsilon$-net $\mathcal{M}$ of the sphere $\mathcal{S}^{k_1-1}$ and $\mathcal{N}$ of the sphere $\mathcal{S}^{k_2-1}$ with cardinalities:
$$
|\mathcal{M}| \leq\left(\frac{2}{\varepsilon}+1\right)^{k_1} \quad \text { and } \quad|\mathcal{N}| \leq\left(\frac{2}{\varepsilon}+1\right)^{k_2}.
$$
Then the spectral norm of $\bF_i$ can be bounded using these nets as follows:
\begin{equation}\label{eqn:epsilon_net}
\|\bF_i\| \leq \frac{1}{1-2 \varepsilon} \cdot \max _{\bx \in \mathcal{N}, \by \in \mathcal{M}}\langle \bF_i \bx, \by\rangle, \quad \varepsilon \in\left[0, \frac{1}{2}\right)
\end{equation}

Fix $\bx \in \mathcal{N}, \by \in \mathcal{M}$. By the decomposition of $\bF_i$ in Assumption \ref{asum:FU_subG} (a) $\bF_i=\bSigma_{k_1}^{1 / 2} \bZ \bSigma_{k_2}^{1 / 2}$, we have
$$
\langle \bF_i \bx, \by\rangle=\left\langle \bZ\left( \bSigma_{k_2}^{1 / 2} \cdot \bx\right), \bSigma_{k_1}^{1 / 2} \by\right\rangle .
$$

Denote $\widetilde{\bx}=\bSigma_{k_2}^{1 / 2}  \bx$ and $\widetilde{\by}=\bSigma_{k_1}^{1 / 2}  \by$, we have
$$
\begin{aligned}
\|\langle \bF_i \bx, \by\rangle\|_{\psi_2}^2 &\leq\|\langle\boldsymbol{Z} \widetilde{\bx}, \widetilde{\by}\rangle\|_{\psi_2}^2 \leq C\sum_{i=1}^{k_1} \sum_{j=1}^{k_2} \|Z_{ij} \widetilde{\bx}_i \widetilde{\by}_j\|_{\psi_2}^2\\
&\leq CK^2 \sum_{i=1}^{k_1} \sum_{j=1}^{k_2}{\widetilde{x_j}}^2 \widetilde{y}_i^2 
\leq CK^2\|\widetilde{\bx}\|^2\|\widetilde{\by}\|^2 \leq C \left\|\bSigma_{k_1}\right\|\left\|\bSigma_{k_2}\right\|,
\end{aligned}
$$
where $K:= \max_{i,j} \|Z_{ij}\|_{\psi_2}$.
That leads to the tail bound:
$$
\mathbb{P}(\langle \bF_i \bx, \by\rangle \geq u) \leq 2 \exp \left(-\frac{c u^2} {\left\|\bSigma_{k_1}\right\|\left\|\bSigma_{k_2}\right\|}\right).
$$

By $\varepsilon$-net argument, we have
$$
\mathbb{P}\left(\max _{\bx \in \mathcal{N}, \by \in \mathcal{M}}\langle \bF_i \bx, \by\rangle \geq u\right) \leq \sum_{\bx \in \mathcal{N}, \by \in \mathcal{M}} \mathbb{P}(\langle \bF_i \bx, \by\rangle \geq u),
$$
and setting $\varepsilon=\frac{1}{4}$, we bound the probability above by
$$
9^{\left(k_1+k_2\right)} \cdot 2 \exp \left(-\frac{c u^2 }{\left\|\bSigma_{k_1}\right\|\left\|\bSigma_{k_2}\right\|}\right) .
$$

Choose $u=C\left\|\bSigma_{k_1}^{1 / 2}\right\|\left\|\bSigma_{k_2}^{1 / 2}\right\|\left(\sqrt{k_2}+\sqrt{k_1}+\sqrt{t} \right)$. If $C$ is chosen sufficiently large, such that
$\frac{c u^2 }{\left\|\bSigma_{k_1}\right\|\left\|\bSigma_{k_2}\right\|} \geq 3\left(k_1+k_2\right)+t$, we have 
\[
\mathbb{P}\left(\max _{\bx \in \mathcal{N}, \by \in \mathcal{M}}\langle \bF_i \bx, \by\rangle \geq u\right) \leq 9^{\left(k_1+k_2\right)} \cdot 2 \exp \left(-3\left(k_1+k_2\right)-t\right) \leq 2 \exp \left(-t\right).
\]

Finally, combining the above inequality with \eqref{eqn:epsilon_net} and by the assumption that $\|\bSigma_{k_1}\|$ and $\|\bSigma_{k_2}$ are bounded, we conclude that
$$
\PP\left( \|\bF_i\| \geq C(\sqrt{k_1} + \sqrt{k_2} +\sqrt{t}) \right) \leq 2 \exp(-t).
$$
\end{proof}
}

\begin{proof}[Proof of Lemma \ref{lem:M}]
We provide in details the proof for $\hat{\bS}_1$ as follows. The proof for $\hat{\bS}_2$ is similar and thus omitted here.
By definition, we can decompose $\hat{\bS}_1-\bS^*_1$ by
\begin{equation}\label{eqn:MS}
\begin{aligned}
\hat{\bS}_1-\bS^*_1 &= \frac{1}{p_1p_2n} \sum_{i=1}^n{\bX_i\bX_i^\top} - \frac{1}{p_1} \bR^*{\bR^*}^\top\\
& = \frac{1}{p_1p_2n} \sum_{i=1}^n\left(\bR^*\bF_i{\bC^*}^\top+\bU_i\right) \left(\bR^*\bF_i{\bC^*}^\top+\bU_i\right) ^\top - \frac{1}{p_1} \bR^*{\bR^*}^\top\\
&= \underbrace{\frac{1}{p_1} \bR^*\left(\frac{1}{p_2n}\sum_{i=1}^n \bF_i{\bC^*}^\top\bC^*\bF_i^\top - \bI_{k_1}\right){\bR^*}^\top}_{\bS^{(1)}} +  \underbrace{\left(\frac{1}{p_1p_2n} \sum_{i=1}^n \bU_i\bU_i^\top -\frac{1}{p_1p_2} \bfm\Sigma_U^{(1)}\right)}_{\bS^{(2)}}  \\
&\ \ +  \underbrace{\frac{1}{p_1p_2n} \bR^*\bF_i{\bC^*}^\top\bU_i^\top}_{\bS^{(3)}} + \underbrace{\frac{1}{p_1p_2n} \bU_i(\bR^*\bF_i{\bC^*}^\top)^\top}_{(\bS^{(3)})^\top} + \underbrace{\frac{1}{p_1p_2} \bfm\Sigma_U^{(1)}}_{\bS^{(4)}},
\end{aligned}
\end{equation}
where $\bfm\Sigma_U^{(1)}:=\EE{\bU\bU^\top} $. Now we establish the upper bounds of the operator norm of each terms above.

For $\bS^{(1)}$,  by Lemma \ref{lem:F_op}, there exist a constant $b$ such that 
\[
\|\bF_i\|_2\leq b\left(\sqrt{k_1}+\sqrt{k_2} + \sqrt{t}\right)
\]
with probability at lease $1-2e^{-t}$. In what follows, we analyze under this event. 
We have,
\[
\|\bF_i{\bC^*}^\top\|_2   \leq \|\bC^*\|_2 \|\bF_i\|_2\leq b^\prime \sqrt{p_2}\left(\sqrt{k_1}+\sqrt{k_2} + \sqrt{t}\right)
\]
for some universal constant $b^\prime$. Accordingly, $\frac{1}{p_2} \|\bF_i{\bC^*}^\top\bC^*\bF_i^\top\|_2\leq b^{\prime\prime}t(k_1\vee k_2)$ for some universal constant $b^{\prime\prime}$. This indicates that $\frac{1}{p_2}\bF_i{\bC^*}^\top\bC^*\bF_i^\top$  satisfies Bernstein condition with parameter $b^{\prime\prime}t(k_1\vee k_2)$. Together with { Assumption \ref{asum:FU_subG} (a)} and by Bernstein inequality (Theorem 6.17 in \citet{wainwright2019high}), we have
\[
\begin{aligned}
&\PP\left(\left\|\frac{1}{np_2}\sum_{i=1}^n \bF_i{\bC^*}^\top\bC^*\bF_i^\top -\bI_k\right\|_2\geq \epsilon \right)\\
\leq & \PP\left(\left\|\frac{1}{np_2}\sum_{i=1}^n \bF_i{\bC^*}^\top\bC^*\bF_i^\top -\bI_k\right\|_2\geq \epsilon \cond \frac{1}{p_2} \|\bF_i{\bC^*}^\top\bC^*\bF_i^\top\|_2\leq b^{\prime\prime}t(k_1\vee k_2) \right) \\
&\ +  \PP\left( \frac{1}{p_2} \|\bF_i{\bC^*}^\top\bC^*\bF_i^\top\|_2 > b^{\prime\prime}t(k_1\vee k_2) \right) \\
\leq & 2k_1\exp\left(-\frac{n\epsilon^2}{2 b^{\prime\prime}t(k_1\vee k_2)(1+\epsilon)}\right) + 2e^{-t},
\end{aligned}
\]
where we use that fact that 
\[
\begin{aligned}
\var\left(\frac{1}{p_2}\bF_i{\bC^*}^\top\bC^*\bF_i^\top -\bI_{k_1} \right) & = \EE\left[\left(\frac{1}{p_2}\bF_i{\bC^*}^\top\bC^*\bF_i^\top\right)^2\right] - \bI_{k_1}\\
&\preceq \EE\left[\left\|\frac{1}{p_2}\bF_i{\bC^*}^\top\bC^*\bF_i^\top\right\|_2 \frac{1}{p_2}\bF_i{\bC^*}^\top\bC^*\bF_i^\top \right] - \bI_{k_i} \\
&\preceq  b^{\prime\prime}t(k_1\vee k_2)\EE\left[\frac{1}{p_2}\bF_i{\bC^*}^\top\bC^*\bF_i^\top\right] - \bI_{k_1} = [ b^{\prime\prime}t(k_1\vee k_2) -1]\bI_{k_1}.
\end{aligned}
\]
Then we have 
\[
\left\|\frac{1}{np_2}\sum_{i=1}^n \bF_i{\bC^*}^\top\bC^*\bF_i^\top -\bI_k\right\|_2 \leq C_1\sqrt{t(k_1\vee k_2)}\sqrt{\frac{\log k_1 + t}{n}}
\]
with probability at lease $1-3e^{-t}$, and thus
\[
\|\bS^{(1)}\|_2 \leq C_1^\prime \sqrt{t(k_1\vee k_2)}\sqrt{\frac{\log k_1 + t}{n}}
\]
with probability at lease $1-3e^{-t}$.

For $\bS^{(2)}$, similarly to Lemma \ref{lem:F_op}, there exist a constant $b$ such that $\|\bU_i\bU_i^\top\| /(p_1p_2)\leq bt(p_1\vee p_2)/(p_1p_2)$ with probability at least $1-2e^{-t}$. Similar to the above part for $\bS^{(1)}$, we have by Bernstein inequality,
\[
\|\bS^{(2)}\|_2 \leq C_2 \sqrt{\frac{t(p_1\vee p_2)}{p_1p_2}} \sqrt{\frac{\log p_1 +t}{n}} \leq C_2\sqrt{\frac{t}{n}}
\]
with probability at least $1-3e^{-t}$.

Moreover, $\bS^{(3)}$ can be decomposed to
\[
\|\bS^{(3)}\|_2 \leq \left\|\frac{1}{p_1p_2} \EE[\bR^*\bF{\bC^*}^\top\bU^\top]\right\|_2 + \left\|\frac{1}{p_1p_2n}\sum_{i=1}^n \left(\bR^*\bF_i{\bC^*}^\top\bU_i^\top -  \EE[\bR^*\bF{\bC^*}^\top\bU^\top]\right)\right\|_2,
\]
and the first term can be bounded by $C_3/\sqrt{p_1p_2} $ by Assumption \ref{asum:FU}. The second term is bounded using the same technique as above. First, we have
\[
\|\bR^*\bF_i{\bC^*}^\top\bU_i^\top\|_2 
\leq \|\bR^*\|_2\|\bF_i\|_2\|\bC^*\|_2\|\bU_i\|_2\leq c_3 \sqrt{p_1p_2}(\sqrt{k_1}+\sqrt{k_2} + \sqrt{p_1}+\sqrt{p_2}+t)
\] 
with probability at least $1-4e^{-t}$. 
Hence $\frac{1}{p_1p_2}\bR^*\bF_i{\bC^*}^\top\bU_i$ satisfies Bernstein condition with parameter $c_3^\prime {\frac{\sqrt{k_1}+\sqrt{k_2} + \sqrt{p_1}+\sqrt{p_2}+t}{\sqrt{p_1p_2}}}$ with probability at least $1-4e^{-t}$. Therefore, by Bernstein inequality, we have 
\[
\left\|\frac{1}{p_1p_2n}\sum_{i=1}^n \left(\bR^*\bF_i{\bC^*}^\top\bU_i^\top - \EE[\bR^*\bF_i{\bC^*}^\top\bU_i^\top]\right)\right\|_2 \leq C_3^\prime \sqrt{\frac{\sqrt{k_1}+\sqrt{k_2} + \sqrt{p_1}+\sqrt{p_2}+t}{p_1p_2}} \sqrt{\frac{\log p_1 +t}{n}}
\]
with probability at least $1-5e^{-t}$. Thus, $\|\bS^{(3)} \|_2 \leq \frac{C_3}{\sqrt{p_1p_2}} + C_3^\prime  \sqrt{\frac{\log p_1 +t}{np_1p_2}}$ with probability at least $1-5e^{-t}$.

Finally, for $\bS^{(4)}$, we take use of Assumption \ref{asum:FU_subG} (b), which leads to
\[
\left\|\bfm\Sigma_U^{(1)}\right\|_F^2 = \sum_{i=1}^{p_1}\left(\EE\sum_{j=1}^{p_2} u_{ij}^2\right)^2 + \sum_{i,k\in[p_1]\\ i\neq k}\left(\EE\sum_{j=1}^{p_2} u_{ij} u_{kj}\right)^2 \leq p_1p_2^2b^4 + p_1p_2 c_4.
\]
Thus, we have $\|\bS^{(4)}\|_2\leq \frac{C_4}{\sqrt{p_1}}$.

All in all, substituting all the above results into \eqref{eqn:MS} completes the proof.

\end{proof}

A straightforward corollary of Proposition \ref{lem:Ferr_nuc} is stated as follows.
\begin{corollary}\label{cor:F}
Under the same conditions as Proposition \ref{lem:Ferr_nuc}, we have
$$\|\hat{\FF}-\FF(\bH_2\otimes\bH_1)^\top\|_2 = \mathcal \mathcal O_p\paren{\sqrt{\frac{n(k_1\wedge k_2)}{p_1 p_2}}}.$$
\end{corollary}

\subsection{Analysis of OLS after pre-trained diversification} \label{append:ols-after-diversify}

\begin{proposition}[OLS after pre-trained diversification]\label{thm:ols-pre-trained}
Suppose that the projection matrices $\bW_1$ and $\bW_2$ satisfy Definition \ref{def:W}, and Assumption \ref{asu:H}, \ref{asu:F}, and \ref{asu:infer} in Section \ref{append:assump} in the supplementary hold. Furthermore, assume that $\sqrt{n}\ll p_1p_2$. Then we have
\[
\begin{aligned}
\sqrt{n}\left(\tilde\bGamma-\mathring{\bGamma}\right) &= 
\sqrt{n}\big(\hat{\FF}^\top\hat{\FF}\big)^{-1}(\bH_2\otimes\bH_1)\FF^\top \UU + {\mathcal \mathcal O_p(	\frac{1}{\sqrt{p_1p_2}} + \frac{\sqrt{n}}{p_1p_2})} \\
&\xrightarrow{d} \cM\cN \big(\bfm 0, (\bLambda_H\bSigma_F\bLambda_H^\top)^{-1}\bLambda_H \bPhi\bLambda_H^\top (\bLambda_H\bSigma_F\bLambda_H^\top)^{-1}, \bPsi \big)
\end{aligned}
\]
where $\bSigma_F$, $\bLambda_H$, $\bPhi$, and $\bpsi$ are defined in Assumption \ref{asu:H}, \ref{asu:F}, and \ref{asu:infer} in Section \ref{append:assump} in the supplementary. 
\end{proposition}

Before providing the proof of Proposition \ref{thm:ols-pre-trained}, we first present two lemmas whose proofs will be given shortly after.

\begin{lemma}\label{lem:FF_lim}
Let $\hat{\FF}:=\big(\operatorname{vec}(\hat{\bF}_1), \ldots, \operatorname{vec}(\hat{\bF}_n)\big)^\top$ from the {\em ``pre-trained projection''}.
Suppose Assumption \ref{asm:RC} (b), \ref{asu:F} (b), and \ref{asu:infer} (a) - (c) hold, then we have
\[
\plim_{n,p_1,p_2\rightarrow \infty} \frac{1}{n} \hat{\FF}^\top\hat{\FF} = \bLambda_H^\top\bSigma_F\bLambda_H,
\]
where $ \bSigma_F := \plim_{n\rightarrow\infty}\sum_{i=1}^n \tilde\bff_i \tilde\bff_i^\top $ is positive definite, $\tilde{\bff}_i := \vect(\bF_i)$, and $ \bLambda_H := \plim_{p_1,p_2\rightarrow \infty}  (\bH_2\otimes\bH_1) $ is full rank as been assumed in Assumption \ref{asu:infer} (b) and Assumption \ref{asu:F} (b), respectively.
\end{lemma}

\begin{lemma}\label{lem: F_I}
Suppose the same condition of Lemma \ref{lem:FF_lim} holds, we have
{
\[
\big(\hat{\FF}^\top\hat{\FF}\big)^{-1} \hat{\FF}^\top{ \left[{\FF} - \hat{\FF}(\bH_2^{-1}\otimes\bH_1^{-1})^\top\right] = \mathcal \mathcal O_p\left({\frac{1}{\sqrt{n p_1p_2}}} + \frac{1}{p_1p_2}\right).
}\]
}
\end{lemma}

\begin{proof}[Proof of Proposition  \ref{thm:ols-pre-trained}]
\[
\begin{aligned}
\tilde{\bGamma}^\top  &= 	\big(\hat{\FF}^\top\hat{\FF}\big)^{-1} \hat{\FF}^\top \XX \\
& = \big(\hat{\FF}^\top\hat{\FF}\big)^{-1} \hat{\FF}^\top \left[\mathring{\FF} \mathring{\bGamma}^\top + \UU\right]\\
& =	\big(\hat{\FF}^\top\hat{\FF}\big)^{-1} \hat{\FF}^\top \left[\hat{\FF} - \hat{\FF} + \FF(\bH_2\otimes\bH_1)^\top\right]  \mathring{\bGamma}^\top +  	\big(\hat{\FF}^\top\hat{\FF}\big)^{-1} \hat{\FF}^\top\UU\\
& = \mathring{\bGamma}^\top + \big(\hat{\FF}^\top\hat{\FF}\big)^{-1} \hat{\FF}^\top \left[{\FF} - \hat\FF(\bH_2^{-1}\otimes\bH_1^{-1})^\top \right] {\bGamma}^{*\top}  \\
& \ \ \ \  + \big(\hat{\FF}^\top\hat{\FF}\big)^{-1} \left[\hat{\FF} - \FF(\bH_2\otimes\bH_1)^\top\right]^\top \UU  + \big(\hat{\FF}^\top\hat{\FF}\big)^{-1}(\bH_2\otimes\bH_1)\FF^\top \UU\\
& =: \mathring{\bGamma}^\top +\Romannum{1}+\Romannum{2}+\Romannum{3}.
\end{aligned}
\]
Lemma \ref{lem: F_I} shows that $ \Romannum{1} = \mathcal \mathcal O_p\left({\frac{1}{\sqrt{n p_1p_2}}} + \frac{1}{p_1p_2}\right).$
Also, for part $\Romannum{2},$ for every $r\in[p_1\!p_2]$, let $\UU_{\cdot\,r}$ be the $r$-th column of $\UU$. Then we have
\[
\begin{aligned}
&\big(\hat{\FF}^\top\hat{\FF}\big)^{-1} \left[\hat{\FF} - \FF(\bH_2\otimes\bH_1)^\top\right]^\top \UU_{\cdot\,r}
= \frac{1}{p_1\!p_2} \big(\hat{\FF}^\top\hat{\FF}\big)^{-1} (\bW_2\otimes\bW_1)^\top \UU^\top\UU_{\cdot\,r}\\
=& \big[n \big(\hat{\FF}^\top\hat{\FF}\big)^{-1}\big] \big[ \frac{1}{np_1\!p_2} \sum_{i=1}^n \sum_{j=1}^{p_1\!p_2} (\bW_2^\top \otimes\bW_1^\top)_j \tilde{u}_{ij} \tilde{u}_{ir}  \big]
= \mathcal \mathcal O_p(\frac{1}{ \sqrt{n p_1 p_2}}),
\end{aligned}
\]
where the last identity follows from Lemma \ref{lem:FF_lim} and Assumption \ref{asu:infer} (c). %refer to Bai  p164 middle part. 
As for part $\Romannum{3}$, by Assumption \ref{asu:infer} (b), (d), Lemma \ref{lem:FF_lim}, and Slutsky Theorem, we have 
\[
\sqrt{n}\big(\hat{\FF}^\top\hat{\FF}\big)^{-1}(\bH_2\otimes\bH_1)\FF^\top \UU \xrightarrow{d}  \cM\cN \big(\bfm 0, (\bLambda_H\bSigma_F\bLambda_H^\top)^{-1}\bLambda_H \bPhi\bLambda_H^\top (\bLambda_H\bSigma_F\bLambda_H^\top)^{-1}, \bPsi \big).
\]
Combining all the three parts gives the conclusion.
\end{proof}

\begin{proof}[Proof of Lemma  \ref{lem:FF_lim}]
By equation \eqref{eqn:FX} and \eqref{eqn:fm1}, we have
\[
\begin{aligned}
\frac{1}{n} \hat{\FF}^\top\hat{\FF} &= \frac{1}{np_1^2p_2^2} (\bW_2\otimes \bW_1)^\top \XX^\top \XX (\bW_2\otimes \bW_1)\\
&= \frac{1}{np_1^2p_2^2} (\bW_2\otimes \bW_1)^\top (\bGamma^*\FF^\top+\UU^\top) (\FF{\bGamma^*}^\top+\UU) (\bW_2\otimes \bW_1)\\
&=  \frac{1}{np_1^2p_2^2} (\bW_2\otimes \bW_1)^\top \bGamma^*\FF^\top \FF{\bGamma^*}^\top (\bW_2\otimes \bW_1) +  \frac{1}{np_1^2p_2^2} (\bW_2\otimes \bW_1)^\top \UU^\top \FF{\bGamma^*}^\top (\bW_2\otimes \bW_1) \\
&\ \ \ \  +\frac{1}{np_1^2p_2^2} (\bW_2\otimes \bW_1)^\top \bGamma^*\FF^\top \UU (\bW_2\otimes \bW_1) +  \frac{1}{np_1^2p_2^2} (\bW_2\otimes \bW_1)^\top \UU^\top \UU (\bW_2\otimes \bW_1)\\
&= \frac{1}{n} (\bH_2\otimes\bH_1)\FF^\top \FF(\bH_2\otimes\bH_1)^\top  +  \frac{1}{np_1p_2} (\bW_2\otimes \bW_1)^\top \UU^\top \FF(\bH_2\otimes\bH_1)^\top \\
&\ \ \ \  +\frac{1}{np_1p_2}(\bH_2\otimes\bH_1)\FF^\top \UU (\bW_2\otimes \bW_1) +  \frac{1}{np_1^2p_2^2} (\bW_2\otimes \bW_1)^\top \UU^\top \UU (\bW_2\otimes \bW_1)\\
&=: \Romannum{1} +  \Romannum{2} +  \Romannum{3} +  \Romannum{4}.
\end{aligned}
\]
By Assumption \ref{asu:F} (b), \ref{asu:infer} (b), and the continuous mapping theorem we know $\Romannum{1} \xrightarrow{P} \bLambda_H^\top \bSigma_F \bLambda_H$.
Moreover, Assumption \ref{asu:infer} (a) indicates that $\Romannum{2}, \Romannum{3} = \mathcal \mathcal o_p(1)$, and  Assumption \ref{asu:infer} (c) indicates that $\Romannum{4}= \mathcal \mathcal o_p(1)$. Combining all the four parts gives the conclusion.
\end{proof}

\vspace{0.2cm}
\begin{proof}[Proof of Lemma  \ref{lem: F_I}]
\begin{equation}\label{eqn:part1}
\begin{aligned}
&\left[{\FF} - \hat{\FF}(\bH_2^{-1}\otimes\bH_1^{-1})^\top\right]^\top \hat{\FF}\\
=& \left[{\FF} - \hat{\FF}(\bH_2^{-1}\otimes\bH_1^{-1})^\top\right]^\top \FF(\bH_2\otimes\bH_1)^\top +  \left[{\FF} - \hat{\FF}(\bH_2^{-1}\otimes\bH_1^{-1})^\top\right]^\top \left[\hat{\FF} - \FF(\bH_2\otimes\bH_1)^\top \right]\\
=&  - \frac{1}{p_1\!p_2} (\bH_2^{-1}\otimes\bH_1^{-1}) (\bW_2\otimes\bW_1)^\top \UU^\top\FF(\bH_2\otimes\bH_1)^\top \\
&\quad -\frac{1}{(p_1\!p_2)^2} (\bH_2^{-1}\otimes\bH_1^{-1}) (\bW_2\otimes\bW_1)^\top \UU^\top \UU (\bW_2\otimes\bW_1)\\
= &  - \frac{1}{p_1\!p_2} (\bH_2^{-1}\otimes\bH_1^{-1}) \left[\sum_{i=1}^n \sum_{j=1}^{p_1\!p_2} [\bW_2^\top\otimes\bW_1^\top]_j \tilde{\bff_i}^\top \tilde{u}_{ij}\right](\bH_2\otimes\bH_1)^\top\\
&\quad -\frac{1}{(p_1\!p_2)^2} (\bH_2^{-1}\otimes\bH_1^{-1})  \left[\sum_{i=1}^n \sum_{j=1}^{p_1\!p_2} \sum_{k=1}^{p_1\!p_2} [\bW_2^\top\otimes\bW_1^\top]_j [\bW_2^\top\otimes\bW_1^\top]_k^\top \tilde{u}_{ij}\tilde{u}_{ik}\right]\\
=& - \frac{1}{p_1\!p_2} \mathcal \mathcal O_p(1) \mathcal \mathcal O_p(\sqrt{p_1\!p_2n}) \mathcal \mathcal O_p(1) -\frac{1}{(p_1\!p_2)^2}  \mathcal \mathcal O_p(1) \mathcal \mathcal O_p(np_1\!p_2)\\
=& \mathcal \mathcal O_p\left(\sqrt{\frac{n}{p_1\!p_2}} + \frac{n}{p_1\!p_2}\right),
\end{aligned}
\end{equation}
where the second identity is by substituting all $\hat{\FF}$ by 
\begin{equation}\label{eqn:FFhat_vec}
\hat{\FF} = \FF(\bH_2\otimes \bH_1)^\top + \frac{1}{p_1\!p_2} \UU(\bW_2\otimes\bW_1)
\end{equation}
according to equation \eqref{eqn:FFhat}, and the last identity is by Assumption \ref{asu:infer} (a) - (c) and Law of Large Number.

Moreover, Lemma \ref{lem:FF_lim} indicates that $\big(\hat{\FF}^\top\hat{\FF}\big)^{-1}=\mathcal \mathcal O_p(n^{-1})$. The conclusion follows thereafter.
\end{proof}

\subsection{Proof of Theorem \ref{thm:MCMR}}
\begin{proof}
By Proposition \ref{thm:ols-pre-trained}, we have that 
\[
\tilde\bGamma = \mathring{\bGamma} + \tilde{\bM} \UU + \mathcal \mathcal O_p(\frac{1}{\sqrt{np_1p_2}} + \frac{1}{p_1p_2}),
\]
where $\tilde{\bM}:=\big(\hat{\FF}^\top\hat{\FF}\big)^{-1}(\bH_2\otimes\bH_1)\FF^\top \in \RR^{k_1\!k_2\times n}$. 
Then
\begin{equation}\label{eqn:MR}
\hat{\bR} = \frac{1}{p_2k_2} \bE_{p_1p_2}\tilde{\bGamma}^\top \bE_{k_1k_2}^\top
= \mathring{c}_s \mathring{\bR} + \frac{1}{p_2k_2} \bE_{p_1p_2} (\tilde{\bM}\UU)^\top \bE_{k_1k_2}^\top + \mathcal \mathcal O_p(\frac{1}{\sqrt{np_1p_2}} + \frac{1}{p_1p_2}),
\end{equation}
and by the proof of Proposition \ref{thm:ols-pre-trained}, 
$$\sqrt{n} \left(\tilde{\bM} \UU\right)^\top \xrightarrow{\ d\ } \cM\cN \big(\bfm 0, \bPsi, (\bLambda_H\bSigma_F\bLambda_H^\top)^{-1}\bLambda_H \bPhi\bLambda_H^\top (\bLambda_H\bSigma_F\bLambda_H^\top)^{-1}\big).$$
By linear transformation of matrix Gaussian distribution, we have
\[
\begin{aligned}
&\sqrt{\frac{n}{p_2k_2}} \bE_{p_1p_2} (\tilde{\bM}\UU)^\top \bE_{k_1k_2}^\top\\
& \xrightarrow{\ d\ } \cM\cN \big(\bfm 0, \frac{1}{p_2} \bE_{p_1p_2}\bPsi \bE_{p_1p_2}^\top, \frac{1}{k_2}\bE_{k_1k_2}(\bLambda_H\bSigma_F\bLambda_H^\top)^{-1}\bLambda_H \bPhi\bLambda_H^\top (\bLambda_H\bSigma_F\bLambda_H^\top)^{-1} \bE_{k_1k_2}^\top\big).
\end{aligned}
\]
Since $\frac{1}{p_2} \bE_{p_1p_2}\bPsi \bE_{p_1p_2}^\top = \frac{1}{p_2} \left(\sum_{i=1}^{p_2} \bPsi_{ii} + \sum_{i\neq j} \bPsi_{ij}  \right) \rightarrow \frac{1}{p_2} \sum_{i=1}^{p_2} \bPsi_{ii}$ as $p_2$ goes to infinity, by Assumption \ref{asu:infer} (d), we have
\[
\begin{aligned}
&\sqrt{\frac{n}{p_2k_2}} \bE_{p_1p_2} (\tilde{\bM}\UU)^\top \bE_{k_1k_2}^\top\\
& \xrightarrow{\ d\ } \cM\cN \big(\bfm 0, \frac{1}{p_2} \sum_{i=1}^{p_2} \bPsi_{ii}, \frac{1}{k_2}\bE_{k_1k_2}(\bLambda_H\bSigma_F\bLambda_H^\top)^{-1}\bLambda_H \bPhi\bLambda_H^\top (\bLambda_H\bSigma_F\bLambda_H^\top)^{-1} \bE_{k_1k_2}^\top\big)\\
&= \cM\cN \big(\bfm 0, \mathcal \mathcal O_p(1)_{p_1\times p_1}, \mathcal \mathcal O_p(1)_{k_1\times k_1}\big).
\end{aligned}
\]
Bring back to equation \eqref{eqn:MR}, we have
\[
\begin{aligned}
&\sqrt{np_2k_2} \left(\hat{\bR} - \mathring{c}_s \mathring{\bR}\right)\\
&= \sqrt{\frac{n}{p_2k_2}} \bE_{p_1p_2} (\tilde{\bM}\UU)^\top \bE_{k_1k_2}^\top +  \mathcal \mathcal O_p\left(\sqrt{\frac{k_2}{p_1}} + \frac{\sqrt{nk_2}}{p_1\sqrt{p_2}} \right)
& \xrightarrow{\ d\ } \cM\cN \big(\bfm 0, \mathcal \mathcal O_p(1)_{p_1\times p_1}, \mathcal \mathcal O_p(1)_{k_1\times k_1}\big).
\end{aligned}
\]
With some linear shuffle transformation, we can similarly get that 
\[
\begin{aligned}
\sqrt{np_1k_1} \left(\hat{\bC} - \mathring{r}_s \mathring{\bC}\right)
\xrightarrow{\ d\ } \cM\cN \big(\bfm 0, \mathcal \mathcal O_p(1)_{p_2\times p_2}, \mathcal \mathcal O_p(1)_{k_2\times k_2}\big).
\end{aligned}
\]
\end{proof}

\subsection{Proof of Theorem \ref{thm:U_element}}
\begin{proof}
For clear presentation of the proof, we assume that $\mathring{c}_s=1$ and let $\hat\bU_i = \bX_i - \hat\bR \hat\bF_i\hat\bC / \hat r_s$. The proof is similar under the general setting without this assumption but more tedious.
By Theorem \ref{thm:MCMR}, we obtain that
\[
\sqrt{np_2k_2} \left(\hat{\bR} - \mathring{\bR}\right) \xrightarrow{\ d\ } \cM\cN\left(\bfm 0, \mathcal \mathcal O_p(1)_{p_1\times p_1}, \mathcal \mathcal O_p(1)_{k_1\times k_1}\right),
\]
which implies that $\left\|\hat{\bR}_{i\cdot} - \mathring{\bR}_{i\cdot}\right\|_2 = \mathcal \mathcal O_p\left(\frac{1}{{ \sqrt{np_2}}}\right)$, where $\hat{\bR}_{i\cdot}$ and $\mathring{\bR}_{i\cdot}$ are columns of the $i$-th row of $\hat{\bR}$ and $\mathring{\bR}$, respectively.
Also, it indicates that element-wise, we have
\[
\sqrt{np_2k_2} \left(\hat r_{ij} - \mathring{r}_{ij}\right) \xrightarrow{\ d\ } \cN(0, b_{ij}),\ \text{\ for } i\in[p_1] \text{ and } j\in[k_1],
\]
where $b_{ij}$ are some constants. Thus, $|\hat r_{ij} - \mathring{r}_{ij}| = \mathcal \mathcal O_p(1/\sqrt{np_2})$. 
Besides, Theorem \ref{thm:MCMR} implies that ${\left\|\hat{\bC}_{i\cdot} - \mathring{r}_s \mathring{\bC}_{i\cdot}\right\|_2} = \mathcal \mathcal O_p\left(\frac{1}{\sqrt{np_1}}\right)$. 
Hence, we have
\[
{\left\|\frac{1}{\hat r_s} \hat{\bC}_{i\cdot} - \mathring\bC_{i\cdot}\right\|_2} 
\leq \frac{1}{\mathring r_s}\left(\left\|\hat{\bC} - \mathring{r}_s\mathring{\bC}\right\|_2 + |\hat r_s - r_s|\left\|\hat{\bC}\right\|_2 \right) = \mathcal \mathcal O_p\left( \frac{1}{\sqrt{np_1}} + \frac{1}{\sqrt{np_2}} \right).
\]
Now we can bound $|\hat{u}_{i,jk}-u_{i,jk}|$ as follows.
\[
\begin{aligned}
|\hat{u}_{i,jk}-u_{i,jk}| &= \left|\frac{1}{\hat r_s} (\hat{\bR}_{j\cdot})^\top \hat{\bF}_i \hat{\bC}_{k\cdot} - (\bR^*_{j\cdot})^\top \bF_i \bC^*_{k\cdot}\right|\\
&\leq \left|\frac{1}{\hat r_s} (\hat{\bR}_{j\cdot})^\top \hat{\bF}_i \hat{\bC}_{k\cdot} - (\hat\bR_{j\cdot})^\top\hat\bF_i \mathring{\bC}_{k\cdot}\right| 
+ \left|(\hat{\bR}_{j\cdot})^\top \hat{\bF}_i \mathring{\bC}_{k\cdot} - (\hat\bR_{j\cdot})^\top \bH_1\bF_i\bH_2^\top\mathring{\bC}_{k\cdot}\right| \\
& \ \ \ \  +\left|(\hat{\bR}_{j\cdot})^\top \bH_1\bF_i\bC^*_{k\cdot} - (\mathring{\bR}_{j\cdot})^\top \bH_1\bF_i \bC^*_{k\cdot}\right|\\
&\leq \left\|(\hat{\bR}_{j\cdot})^\top \hat{\bF}_i\right\|_2 \left\|\frac{1}{\hat r_s} \hat\bC_{k\cdot} - \mathring{\bC}_{k\cdot}\right\|_2 + \left\|\hat\bR_{j\cdot}\right\|_2\left\|\hat\bF_i -\bH_1\bF_i\bH_2^\top\right\|_2 \left\|\mathring{\bC}_{k\cdot}\right\|_2 \\
& \ \ \ \ + \left\|\hat{\bR}_{j\cdot} - \mathring{\bR}_{j\cdot} \right\|_2 \left\|\bH_1\right\|_2 \left\|\bF_i\bC_{k\cdot}\right\|_2 \\ 
&= \mathcal \mathcal O_p\left(\frac{1}{\sqrt{np_1}} + \frac{1}{\sqrt{np_2}} + \frac{1}{\sqrt{p_1p_2}}  \right)
\end{aligned}
\]
\end{proof}

\section{FAMAR-LowRank with Observed $\bF_t$} \label{append:famar-observed-factor}

We first establish the results of the general optimization problem \eqref{eqn:convex} with $\bF_i$ and $\bU_i$ known.
We start with introducing some notations. Let $\bG_i:=\operatorname{diag}(\bF_i,\bU_i)\in\RR^{(k_1+p_1)\times (k_2+p_2)}$ be the block-diagonal matrix of observations for $i\in[n]$. Let $\bfm{\Theta}:=\operatorname{diag}(\bA,\bB)\in\RR^{(k_1+p_1)\times (k_2+p_2)}$ be the block-diagonal matrix of matrix parameters. Then, the operator $\mG_n:\RR^{(k_1+p_1)\times (k_2+p_2)}\rightarrow \RR^{n}$ is defined via $[\mG_n(\bfm{\Theta})]_i=\angles{\bG_i,\bfm{\Theta}}$. Similarly, define mapping $\mF:\RR^{k_1\times k_2}\rightarrow\RR^n$ via  $[\mF(\bA)]_i=\angles{\bA,\bF_i}$ and mapping $\mU:\RR^{p_1\times p_2}\rightarrow\RR^n$ via  $[\mU(\bB)]_i=\angles{\bB,\bU_i}$.
With these notations, the MFR model \eqref{eqn:famar} can be re-written as
\[
\by = \mG_n(\bfm{\Theta}^*) + \bfm{\varepsilon} = \mF(\bA^*)+ \mU(\bB^*) + \bfm{\varepsilon},
\]
where $\by=(y_1,\ldots,y_n)^\top$, $\bfm{\varepsilon}=(\varepsilon_1,\ldots,\varepsilon_n)^\top$, and $\bfm{\Theta^*}=\operatorname{diag}(\bA^*, \bB^*)$ is the true parameter matrix. In addition, define the operator $\mG_n^*:\RR^n\rightarrow \RR^{(k_1+p_1)\times (k_2+p_2)}$ by $\mG_n^*(\bfm{\varepsilon}):= \sum_{i=1}^n \varepsilon_i\bG_i$.
Moreover, considering the nuclear norm regularization on $\bB$, the following notations are needed. 
For any true $\bB^*\in\RR^{p_1\times p_2}$, let $r:=\operatorname{rank}(\bB^*)$. Matrix $\bB^*$ has a singular value decomposition of the form $\bB^* = \bV_1\bD\bV_2^\top$, where $\bV_1\in\RR^{p_1\times p_1}$ and $\bV_2\in\RR^{p_2\times p_2}$ are orthonormal matrices. Let $\bV_1^r\in\RR^{p_1\times r}$ and $\bV_2^r\in\RR^{p_2\times r}$ be matrices of singular vectors associated with the top $r$ singular values of $\bB^*$. Two sub-spaces of $\RR^{p_1\times p_2}$ are defined by
\[
\mathcal{A}\left(\bV_1^r, \bV_2^r\right):=\left\{\bfm\Delta_B \in \mathbb{R}^{p_1 \times p_2} \mid \operatorname{row}(\bfm\Delta_B) \subseteq \bV_2^r \text{ and }\operatorname{col}(\bfm\Delta_B) \subseteq \bV_1^r\right\}
\]
and
\[
\mathcal{B}\left(\bV_1^r, \bV_2^r\right):=\left\{\bfm\Delta_B \in \mathbb{R}^{p_1 \times p_2} \mid \operatorname{row}(\bfm\Delta_B) \perp \bV_2^r \text{ and }\operatorname{col}(\bfm\Delta_B) \perp \bV_1^r\right\},
\]
where row $(\bfm\Delta_B) \subseteq \mathbb{R}^{p_2}$ and $\operatorname{col}(\bfm\Delta_B) \subseteq \mathbb{R}^{p_1}$ are the row space and column space, respectively, of the matrix $\bfm\Delta_B$. The shorthand notations $\mathcal{A}^r$ and $\mathcal{B}^r$ are used when $\left(\bV_1^r, \bV_2^r\right)$ are already clear from the context.

With that, we are able to introduce an important technical condition called adjusted-restricted strong convexity condition (adjusted-RSC), which is a modification of the strong convexity condition (RSC) proposed by \citet{negahban2011estimation}.
Letting $\mC\subseteq\RR^{(k_1+p_1)\times (k_2+p_2)}$ denote the restricted set, we say that the operator $\mG_n$ satisfies adjusted-RSC over the set $\mC$ if there exists some $\kappa(\mG_n)>0$ such that
\[
\frac{1}{2n}\|\mG_n(\bDelta)\|^2 \geq \kappa(\mG_n)\|\bDelta\|_F^2 \ \text{ for all } \bDelta\in\mC.
\]
To specify the set of interest, let $\Pi_{\mathcal{B}^r}$ denote the projection operator onto the subspace $\mathcal{B}^r$, and let
\begin{equation}
\bfm\Delta_B^{\prime \prime}=\Pi_{\mathcal{B}^r}(\bfm\Delta),
\quad\text{and}\quad \bfm\Delta_B^{\prime}=\bfm\Delta_B-\bfm\Delta_B^{\prime \prime}.
\end{equation}
The set of low-rank matrices is defined by:
$$
\begin{aligned}
\mathcal{C}:= & \left\{\bfm\Delta=  
\begin{pmatrix}
\bfm\Delta_A &\bf 0 \\ 
\bf 0 & \bfm\Delta_B
\end{pmatrix}\mid
\bfm{\Delta}_A\in \mathbb{R}^{k_1 \times k_2}, 
%\bfm{\Delta}_A { has the same dimension as } \bA,
\bfm\Delta_B\in \mathbb{R}^{p_1 \times p_2} ,\|\bfm\Delta_B^{\prime \prime}\|_* \leq 3\|\bfm\Delta_B^{\prime}\|_* + \|\bfm\Delta_A\|_*\right\}.
\end{aligned}
$$
The following theorem presents the Frobenius norm error bound for the FARMA estimator.

\begin{theorem}\label{thm:consistency1}
Suppose that the operator $\mG_n$ satisfies adjusted-RSC with constant $\kappa(\mG_n)>0$ over the set $\mC$, and that the tuning parameter $\lambda_n$ is chosen such that $\lambda_n\geq 2\|\mG_n^*(\bfm\varepsilon)\|_{2}/n$. Then any solution $\hat{\bfm\Theta}:=\begin{pmatrix}
\hat\bA &\bf 0 \\ 
\bf 0 & \hat\bB
\end{pmatrix}$ to \eqref{eqn:convex} satisfies
\vspace{0.25in}
\[
\|\bfm\Theta^*-\hat{\bfm\Theta}\|_F \leq \frac{\lambda_n}{\kappa(\mG_n)}\big(2\sqrt{k_1\wedge k_2} + 6\sqrt{2\operatorname{rank}(\bB^*)}\big).
\]
\end{theorem}

\begin{remark}
By taking $\lambda_n\asymp \big\|\mG_n^*(\bfm\varepsilon)\big\|_2/n$, Theorem \ref{thm:consistency1} indicates that 
\vspace{0.12in}
\[
\|\bfm\Theta^*-\hat{\bfm\Theta}\|_F = \mathcal \mathcal O_p\left(\frac{\big\|\mG_n^*(\bfm\varepsilon)\big\|_2}{n\kappa(\mG_n)}\right)
\]
with $k_1\asymp k_1 \asymp \rank(\bB^*)\asymp 1$.
The following proposition provides the convergence rate under this case with sub-Gaussian assumptions.
\end{remark}

\begin{assumption} \label{asu:subg}
For $i\in[n]$, assume that
\begin{enumerate}[label=(\roman*)]
\item $\EE \varepsilon_i\bG_i = \bfm 0$;
\item $\varepsilon_i$ is $\sigma$-sub-Gaussian, and there exists a constant $N$ such that $\max_{i=1,\ldots,n} \|\bG_i\|\leq N$.
\end{enumerate}
\end{assumption}

\begin{proposition}\label{prop:err_order0}
Suppose that the operator $\mG_n$ satisfies adjusted-RSC with constant $\kappa(\mG_n)>0$ over the set $\mC$, and that the tuning parameter $\lambda_n$ is chosen such that $\lambda_n\asymp 2\|\mG_n^*(\bfm\varepsilon)\|_{2}/n$. In addition, suppose Assumption \ref{asu:subg} holds and $k_1\asymp k_1 \asymp \rank(\bB^*)\asymp 1$.
Then we have
\[
\|\bfm\Theta^*-\hat{\bfm\Theta}\|_F = \mathcal \mathcal O_p\left(\sqrt{\frac{\log(k_1+k_2+p_1+p_2)}{n}} \right).
\]
\end{proposition}

\subsection{Proof of Theorem \ref{thm:consistency1}}
We start with providing a Lemma on the estimation error.
\begin{lemma}\label{lem:general}
Let $\bV_1^r\in\RR^{p_1\times r}$, and $\bV_2^r\in\RR^{p_2\times r}$ be matrices consisting of the top $r$ left and right (respectively) singular vectors of $\bB^*$. Then there exists $ \bfm\Delta_B = \bfm\Delta_B^\prime + \bfm\Delta_B^{\prime\prime}$ of the error $\bfm\Delta_B$ such that:
\begin{enumerate}[label={(\alph*)}]
\item $\operatorname{rank}(\bfm\Delta_B^\prime)\leq 2r$;
\item if $\lambda_n \geq 2\left\|\mG_n^*(\bfm{\varepsilon})\right\|_{2} / n$, then the nuclear norm of $\bfm\Delta_B^{\prime \prime}$ is bounded as
$$
\left\| \bfm\Delta_B^{\prime \prime}\right\|_* \leq 3\left\|\bfm\Delta_B^{\prime}\right\|_*+\left\|\bfm\Delta_A\right\|_*
$$
\end{enumerate}
\end{lemma}
The proof of the Lemma will be provided shortly after, and with that, we are ready to prove Theorem \ref{thm:consistency1}.

\begin{proof}[Proof of Theorem \ref{thm:consistency1}]
Denote $\tilde{\operatorname{vec}}(\bG_i) = (\operatorname{vec}(\bF_i)^\top, \operatorname{vec}(\bU_i)^\top)^\top$ and $\tilde{\operatorname{vec}}(\bfm{\Theta}) = (\operatorname{vec}(\bA)^\top, \operatorname{vec}(\bB)^\top)^\top$. Then $[\mG_n(\bfm{\Theta})]_i = \tilde{\operatorname{vec}}(\bG_i)^\top \tilde{\operatorname{vec}}(\bfm{\Theta})$. By optimality, we have
\[
\frac{1}{2n}\big\|\by - \mG_n(\hat{\bfm\Theta})\big\|^2 + \lambda_n\|\hat{\bB}\|_* \leq \frac{1}{2n}\big\|\by - \mG_n({\bfm\Theta}^*)\big\|^2 + \lambda_n\|{\bB}^*\|_*.
\]
Letting $\bfm\Delta_B:=\bB^*-\hat\bB$, $\bfm\Delta_A:=\bA^*-\hat\bA$, and $\bfm\Delta:=\bfm\Theta^*-\hat{\bfm\Theta}$, we have
\begin{equation}\label{eqn:temp1}
\begin{aligned}
\frac{1}{2n}\|\mG_n(\bfm\Delta)\|^2 & \leq - \frac{1}{n} \angles{\bfm\varepsilon, \mG_n(\bfm\Delta)} + \lambda_n \big(\|\hat\bB + \bfm\Delta_B\|_* - \|\hat\bB\|_*\big)\\
& = \frac{1}{n}\big|\angles{\mG_n^*(\bfm\varepsilon),\bfm\Delta}\big| + \lambda_n \big(\|\hat\bB + \bfm\Delta_B\|_* - \|\hat\bB\|_*\big) \\
&\leq \frac{1}{n}\|\mG_n^*(\bfm\varepsilon)\|_2 \|\bfm\Delta\|_* + \lambda_n \big(\|\hat\bB + \bfm\Delta_B\|_* - \|\hat\bB\|_*\big),
\end{aligned}
\end{equation}
where the first identity is by definition, and the second inequality is by Hölder's inequality. Lemma \ref{lem:general} (b) shows that $\bfm\Delta\in\mC$. Thus, by adjusted-RSC, we have 
\begin{equation}\label{eqn:temp2}
\frac{1}{2n}\|\mG_n(\bfm\Delta)\|^2 \ge \kappa(\mG_n)\|\bfm\Delta\|_F^2.
\end{equation}
Combining inequality \eqref{eqn:temp1} and \eqref{eqn:temp2} gives
\[
\begin{aligned}
\kappa(\mG_n)\|\bfm\Delta\|_F^2 & \leq \frac{1}{2n}\|\mG_n(\bfm\Delta)\|^2\\
&\leq \frac{1}{n}\|\mG_n^*(\bfm\varepsilon)\|_2 \|\bfm\Delta\|_* + \lambda_n \big(\|\hat\bB + \bfm\Delta_B\|_* - \|\hat\bB\|_*\big)\\
&\leq \lambda_n \big(\frac{1}{2}\|\bfm\Delta\|_* + \|\bfm\Delta_B\|_* \big)\\
& = \lambda_n \big(\frac{1}{2}\|\bfm\Delta_A\|_* + \frac{3}{2} \|\bfm\Delta_B\|_* \big)
\end{aligned}
\]
where the third inequality is by the condition that $\lambda_n\geq \frac{2}{n}\|\mG_n^*(\bfm\varepsilon)\|_2$ and the triangular inequality. By triangular inequality and Lemma \ref{lem:general}, we have 
\[
\|\bfm\Delta_B\|_* \leq \|\bfm\Delta_B^\prime\|_* +\|\bfm\Delta_B^{\prime\prime}\|_* \leq 4\|\bfm\Delta_B^\prime\|_* + \|\bfm\Delta_A\|_*
\]
which leads to 
\[
\kappa(\mG_n)\|\bfm\Delta\|_F^2 \leq \lambda_n \big(2\|\bfm\Delta_A\|_* + 6 \|\bfm\Delta_B^\prime\|_* \big).
\]
The proof is completed by the fact that for any matrix $\bD$, $\|\bD\|_*\leq \sqrt{\operatorname{rank}(\bD)}\|\bD\|_F$, which leads to 
\[
\begin{aligned}
\|\bfm\Delta\|_F^2 & \leq \frac{\lambda_n}{\kappa(\mG_n)} \big(2\sqrt{k_1\wedge k_2} \|\bfm\Delta_A\|_F + 6\sqrt{2\operatorname{rank}(\bB^*)} \|\bfm\Delta_B^\prime\|_F \big)\\
&{ \leq \frac{\lambda_n}{\kappa(\mG_n)}\big(2\sqrt{k_1\wedge k_2} + 6\sqrt{2\operatorname{rank}(\bB^*)} \big)\|\bfm\Delta\|_F}.
\end{aligned}
\]
\end{proof}

\begin{proof}[Proof of Lemma \ref{lem:general}]
Part (a) of the Lemma was proved in \citet{recht2010guaranteed} and \citet{negahban2011estimation}. The proof is provided here for completeness. Recall that the SVD of $\bB^*$ is $\bB^* = \bV_1\bD\bV_2^\top$, $\bV_1^r\in\RR^{p_1\times r}$ (resp. $\bV_2^r\in\RR^{p_2\times r}$) is the matrix of singular vectors associated with the top $r$ left (resp. right) singular values of $\bB^*$. 
Let $\bH =\bV_1^\top \bfm\Delta_B \bV_2 \in\RR^{p_1\times p_2}$, and write it in block form as
\[
\bH = \begin{pmatrix}
\bH_{11} &\bH_{12} \\ 
\bH_{21} & \bH_{22}
\end{pmatrix}, \text{ where } \bH_{11} \in \RR^{r\times r} \text{ and } \bH_{22} \in \RR^{(p_1-r)\times(p_2-r)}.
\]
Let $\bfm\Delta^{\prime\prime} := \bV_1\begin{pmatrix}
\bf 0 &\bf 0 \\ 
\bf 0 & \bH_{22}
\end{pmatrix} \in \cB(\bV_1^r, \bV_2^r$), and $\bfm\Delta^\prime := \bfm\Delta-\bfm\Delta^{\prime\prime}$.
Note that 
\[
\rank(\bfm\Delta^\prime) = \rank\begin{pmatrix}
\bH_{11} &\bH_{12} \\ 
\bH_{21} & \bf 0
\end{pmatrix} \leq \rank\begin{pmatrix}
\bH_{11} &\bH_{12} \\ 
\bf 0 & \bf 0
\end{pmatrix} + \rank\begin{pmatrix}
\bH_{11} &\bf 0 \\ 
\bH_{21} & \bf 0
\end{pmatrix} \leq 2r.
\]
For part (b), we have derived in \eqref{eqn:temp1} that
\begin{equation}\label{eqn:lemA1}
0\leq \frac{1}{2n}\|\mG_n(\bfm\Delta)\|^2 \leq \frac{1}{n}\|\mG_n^*(\bfm\varepsilon)\|_2 \|\bfm\Delta\|_* + \lambda_n \big(\|\bB^*\|_* - \|\hat\bB\|_*\big).
\end{equation}
By construction of $\cA^r$ and $\cB^r$, we know $\bB^* = \Pi_{\cB^r}(\bB^*)$, and
\[
\left\|\bB^*+\bfm\Delta_B^{\prime\prime}\right\|_* = \|\bB^*\|_*+\|\bfm\Delta_B^{\prime\prime}\|_*.
\]
Hence
\[
\begin{aligned}
\left\|\hat\bB\right\|_* = \left\|\bB^* + \bfm\Delta_B^{\prime\prime} + \bfm\Delta_B^{\prime} \right\|_* 
\geq \left\|\bB^* + \bfm\Delta_B^{\prime\prime} \right\|_* -\left\|\bfm\Delta_B^{\prime}\right\|_*
=\|\bB^*\|_* + \|\bfm\Delta_B^{\prime\prime}\|_* - \|\bfm\Delta_B^{\prime}\|_*.
\end{aligned}
\]
Thus
\[
\|\bB^*\|_*-\|\hat\bB\|_* \leq \|\bfm\Delta_B^{\prime}\|_* - \|\bfm\Delta_B^{\prime\prime}\|_*.
\]
Substituting that to \eqref{eqn:lemA1}, we have 
\[
0 \leq \frac{1}{n}\|\mG_n^*(\bfm\varepsilon)\|_2 \|\bfm\Delta\|_* + \lambda_n \big(\|\bfm\Delta_B^{\prime}\|_* - \|\bfm\Delta_B^{\prime\prime}\|_*\big).
\]
Since $\lambda_n\geq 2\|\mG_n^*(\bfm\varepsilon)\|_2/n$, that leads to
\[
0\leq \|\bfm\Delta_B^{\prime}\|_* - \|\bfm\Delta_B^{\prime\prime}\|_* + \frac{1}{2}\|\bfm\Delta\|_* 
= \|\bfm\Delta_B^{\prime}\|_* - \|\bfm\Delta_B^{\prime\prime}\|_* + \|\bfm\Delta_B\|_* +\|\bfm\Delta_A\|_*
\leq \frac{3}{2}\|\bfm\Delta_B^{\prime}\|_* - \frac{1}{2}\|\bfm\Delta_B^{\prime\prime}\|_* +\|\bfm\Delta_A\|_*,
\]
i.e. 
\[
\|\bfm\Delta_B^{\prime\prime}\|_* \leq 3\|\bfm\Delta_B^{\prime}\|_* + \|\bfm\Delta_A\|_*.
\]
\end{proof}

\subsection{Proof of Proposition \ref{prop:err_order0}}
\begin{proof}
Let $\bQ_i:=\diag(\bG_i,\bG_i^\top)\in\RR^{d\times d}$, where $d:=k_1+k_2+p_1+p_2$. Then $\|\bQ_i\|_2 = \|\bG_i\|_2$, and by the assumptions, $\varepsilon_i\bQ_i$ is sub-Gaussian with parameter $\bV=c^2\sigma^2 N^2 \bI_d$, where $c$ is some universal constant. Hence, by Matrix Hoeffding's inequality, we have
\begin{equation}\label{eqn:errorder_Gepsilon}
\PP\left(\left\|\frac{1}{n}\mG_n^*(\bfm\varepsilon)\right\| \geq C_1\sqrt{\frac{\log d}{n}}\right) \leq d^{-2}
\end{equation}
for $C_1$ large enough. The conclusion follows directly from Theorem \ref{thm:consistency1}.
\end{proof}

\section{FAMAR-LowRank with Estimated $\bF_i$} \label{append:famar-estimated-factor}

Let $\tilde\by := (\bI_n - \hat{\PP})\by$ denote the residuals of the response vector $\by$ after projecting onto the column space of $\hat{\FF}$, where $\hat{\PP}:=\hat{\FF}(\hat{\FF}^\top \hat{\FF})^{-1}\hat{\FF}^\top$ is the corresponding projection matrix. Then $\hat{\UU} = (\bI_n - \hat{\PP})\XX$. Thus, $\hat{\FF}^\top \hat{\UU} = \bzero$, and it is straightforward to verify that the solution of \eqref{eqn:convex} is equivalent to
\[
\begin{aligned}
&\hat\bB \in \argmin_{\bB\in\RR^{p_1\times p_2}}  \left\{\frac{1}{2n}\sum_{i=1}^n \left(\tilde y_i - \angles{\bB,\hat{\bU}_i} \right)^2 + \lambda_n\|\bB\|_{*}\right\}, \\
& \vect(\hat\bA) =  \big(\hat{\FF}^\top \hat{\FF}\big)^{-1}\hat{\FF}^\top \by. 
\end{aligned}
\]

\subsection{Consistency of $\hat\bA$}

\begin{proof}[Proof of Theorem \ref{thm:consistency2} Part I]
Let $\bPhi^*:=\bA^* - {\bR^*}^\top\bB^*\bC^*\in\RR^{k_1\times k_2}$, $\xbar{\bPhi}^*:=(\bH_1^{-1})^\top \bPhi^*\bH_2^{-1} \in\RR^{k_1\times k_2}$. 
Vectorizing them gives 
\[
\begin{aligned}
\vect(\bPhi^*) &= \vect(\bA^*) - (\bC^*\otimes\bR^*)^\top \vect(\bB^*), \\
\vect(\xbar{\bPhi}^*) &= (\bH_2^{-1}\otimes\bH_1^{-1})^{\top}\vect(\bPhi^*).
\end{aligned}
\]
Recall that the true MFR \eqref{eqn:famar} is 
\[
y_i = \angles{\bA^*, \bF_i} + \angles{\bB^*, \bU_i} + \eps_i.
\]
Thus, we have
\[
y_i = \angles{\bF_i,\bPhi^*} + \angles{\bX_i, \bB^*} + \varepsilon_i.
\]
By the definition of $\hat\UU$ and MFM \eqref{eqn:famar}, it can be equivalently rewritten as
\[
\by = {\FF}\vect(\bPhi^*) + \big(\hat{\FF}\hat{\bGamma}^\top + \hat{\UU}\big) \vect(\bB^*) + \bfm\varepsilon,
\]
where $\hat\bGamma := \hat\bC \otimes \hat\bR$.
We are now ready to calculate the error of $\hat\bA$ by
\[
\begin{aligned} 
&\vect(\hat\bA) - (\bH_2^{-1}\otimes\bH_1^{-1})^\top\vect(\bA^*)\\
=& \big(\hat{\FF}^\top\hat{\FF} \big)^{-1}\hat{\FF}^\top \left({\FF}\vect(\bPhi^*) + \hat{\FF}\hat{\bGamma}^\top\vect(\bB^*) + \bfm\varepsilon\right)\\
&~~~ - \left[(\bH_2^{-1}\otimes\bH_1^{-1})^\top\vect(\bPhi^*) -  (\bH_2^{-1}\otimes\bH_1^{-1})^\top(\bC^*\otimes\bR^*)^\top\vect(\bB^*)\right] \\
=& \left[\big(\hat{\FF}^\top\hat{\FF} \big)^{-1}\hat{\FF}^\top{{\FF}} - (\bH_2^{-1}\otimes\bH_1^{-1})^\top\right] \vect(\bPhi^*)\\
&~~~ + \left[\hat{\bGamma} - (\bH_2^{-1}\otimes\bH_1^{-1})^\top(\bC^*\otimes\bR^*)^\top\right]\vect(\bB^*)  + \big(\hat{\FF}^\top\hat{\FF} \big)^{-1}\hat{\FF}^\top\bfm\varepsilon\\
=&\big(\hat{\FF}^\top\hat{\FF} \big)^{-1}\hat{\FF}^\top \left[\FF - \hat{\FF}(\bH_2^{-1}\otimes\bH_1^{-1})^\top\right]\vect(\bPhi^*) + \left(\hat{\bGamma} - \mathring{\bfm\Gamma}\right)\vect(\bB^*) + \big(\hat{\FF}^\top\hat{\FF} \big)^{-1}\hat{\FF}^\top\bfm\varepsilon,
\end{aligned}
\]
where the last identity uses the equality $\FF = \hat\FF(\bH_2^{-1}\otimes\bH_1^{-1})^\top + \left[ \FF- \hat\FF (\bH_2^{-1}\otimes\bH_1^{-1})^\top\right]$.
Applying the $\ell_2$-norm, we have
\[
\begin{aligned}
&\big\|\vect(\hat\bA) - (\bH_2^{-1}\otimes\bH_1^{-1})^\top\vect(\bA^*)\big\| \\ &\leq \left\|\big(\hat{\FF}^\top\hat{\FF} \big)^{-1}\hat{\FF}^\top \left[\FF - \hat{\FF}(\bH_2^{-1}\otimes\bH_1^{-1})^\top\right]\right\| \big\|\vect(\bPhi^*)\big\|  + \big\|(\hat{\bGamma} - \mathring{\bGamma})\vect(\bB^*)\big\| + \big\|\big(\hat{\FF}^\top\hat{\FF} \big)^{-1}\hat{\FF}^\top\bfm\varepsilon\big\|.
\end{aligned}
\]
The first part, by Lemma \ref{lem: F_I}, under Assumption \ref{asm:RC} (b), \ref{asu:F} (b), and \ref{asu:infer} (a) - (c), can be bounded by $\mathcal \mathcal O_p\left( \big\|\bPhi^*\big\|_F/\sqrt{np_1p_2}\right)$. The second part, by the definition of operator norm and Proposition \ref{thm:ols-pre-trained}, can be bounded by
\[ 
\big\|(\hat{\bGamma} - \mathring{\bGamma})\vect(\bB^*)\big\| 
\leq \|\vect(\bB^*)\|_1 \big\|\hat{\bGamma} - \mathring{\bGamma}\big\|_{2,1}
\leq \mathcal \mathcal O_p\left(\|\vect(\bB^*)\|_1 \sqrt{\frac{(k_1\wedge k_2)k_1k_2}{p_1p_2}}\right).
\]
As for the last term, we have shown that under the condition of Lemma \ref{lem:FF_lim}, $\lambda_{\min}\big(\hat{\FF}^\top \hat{\FF}/n\big)^{1/2}\geq \Theta(1)$. Together with the fact that $\big\|\hat{\FF}^\top \bfm\varepsilon\big\| = \mathcal \mathcal O_p(\sqrt{n})$, we have
\[
\left\|\big(\hat{\FF}^\top\hat{\FF} \big)^{-1}\hat{\FF}^\top\bfm\varepsilon\right\|
\leq \frac{1}{n} \frac{1}{\lambda_{\min}\big(\hat{\FF}^\top \hat{\FF}/n)} \big\|\hat{\FF}^\top \bfm\varepsilon\big\|
\leq \mathcal \mathcal O_p\left(\frac{1}{\sqrt{n}}\right).
\]
Combining the above results together, we get
\[
\begin{aligned}
\big\|\vect(\hat\bA) - (\bH_2^{-1}\otimes\bH_1^{-1})^\top\vect(\bA^*)\big\|
\leq \mathcal \mathcal O_p\left(\frac{\big\|\bPhi^*\big\|_F}{\sqrt{np_1p_2}} +\frac{ \|\vect(\bB^*)\|_1}{\sqrt{p_1p_2}} + \frac{1}{\sqrt{n}}\right).
\end{aligned}
\]
Moreover, by the assumption that $\|\bA^*\|_F=O(p_1p_2)$ and $\|\bB^*\|_*$ is bounded (so $\|\vect(\bB^*)\|_1=\mathcal \mathcal O_p(\sqrt{p_1p_2})$), we complete the proof of the first part of the theorem by
\[
\begin{aligned}
\frac{\left\|\bH_1^\top\hat\bA\bH_2 - \bA^*\right\|_F}{\|\bA^*\|_F} 
=& \frac{\big\|\vect(\hat\bA) - (\bH_2^{-1}\otimes\bH_1^{-1})^\top\vect(\bA^*)\big\|}{\|\bA^*\|_F}\\
=& \mathcal \mathcal O_p\left(\frac{\big\|\bPhi^*\big\|_F}{p_1p_2\sqrt{np_1p_2}} +\frac{ \|\vect(\bB^*)\|_1}{p_1p_2\sqrt{p_1p_2}} + \frac{1}{p_1p_2\sqrt{n}}\right)\\
=& \mathcal \mathcal O_p\left(	\frac{1}{\sqrt{np_1p_2}}\right)
\end{aligned}
\]
\end{proof}

\subsection{Consistency of $\hat\bB$} \label{subsec:B}
\begin{proof}[Proof of Theorem \ref{thm:consistency2} Part II]
Replacing the true $(\bF_i,\bU_i)$ in MFM \eqref{eqn:famar} by the estimated $(\hat\bF_i. \hat{\bU}_i)$, we have
\[
\begin{aligned}
y_i &= \angles{\bF_i,\bA^*} +  \angles{\bU_i,\bB^*} + \varepsilon_i\\
& = \angles{(\bH_1^{-1})^\top\hat\bF_i\bH_2^{-1},\bA^*} +  \angles{\hat\bU_i,\bB^*} + \xbar\varepsilon_i 
= \angles{\xbar\bG_i,\bTheta^*} + \xbar\varepsilon_i,
\end{aligned}
\]
where $\xbar\varepsilon_i := y_i - \angles{\xbar\bG_i,\bTheta^*}$ and $\xbar\bG_i = \operatorname{diag}(\bH_1^{-1} \hat\bF_i (\bH_2^{-1})^\top, \hat\bU_i)$.
Consider the optimization problem \eqref{eqn:convex}.
Let
\begin{equation}\label{eqn:opt_estFU_modify}
\big(\xbar{\bA},\hat{\bB}\big)\in\argmin_{\bA\in\RR^{k_1\times k_2},\\ \bB\in\RR^{p_1\times p_2}} \left\{\frac{1}{2n}\sum_{i=1}^n \left( y_i-\angles{\bA,\bH_1^{-1}\hat{\bF}_i(\bH_2^{-1})^\top}-\angles{\bB,\hat{\bU}_i} \right)^2 + \lambda_n\|\bB\|_{*}\right\}.
\end{equation}
Let $\xbar{\bfm\Theta}:= \diag(\xbar{\bA},\hat{\bB})$, and recall that $\bTheta^* = \operatorname{diag}(\bA^*, \bB^*)$ and the operator $\xbar\mG_n: \RR^{(k_1+p_1)\times(k_2+p_2)} \rightarrow \RR$ is defined such that $\xbar\mG_n(\bTheta) = \angles{\xbar\bG_i, \bTheta}$.
Letting $\bfm\Delta:={\bTheta}^*- \xbar{\bfm\Theta}$, $\bfm\Delta_B:=\bB^*-\hat\bB$, we have $\|\bDelta_B\|_F \leq \|\bDelta\|_F$.
Optimality of $\xbar\bTheta$ gives \[
\frac{1}{2n}\big\|\by - \xbar\mG_n(\xbar{\bfm\Theta})\big\|^2 + \lambda_n\|\hat{\bB}\|_* \leq \frac{1}{2n}\big\|\by - \xbar\mG_n({\bTheta}^*)\big\|^2 + \lambda_n\|{\bB}^*\|_* = \frac{1}{2n} \|\xbar\bepsilon\|^2+ \lambda_n\|{\bB}^*\|_*.
\] 
By elementary algebra, we get
\begin{equation}\label{eqn:bound_delta}
\frac{1}{2n}\|\xbar\mG_n(\bfm\Delta)\|^2 \leq \frac{1}{n} \angles{\xbar{\bfm\varepsilon}, \mG_n(\bfm\Delta)} + \lambda_n \big(\|\hat\bB + \bfm\Delta_B\|_* - \|\hat\bB\|_*\big).
\end{equation}
Based on the proof of Theorem \ref{thm:consistency1} and Lemma \ref{lem:general}, to upper bound the error of $\hat\bB$, we only need to show that all the regularization conditions in Theorem \ref{thm:consistency1} hold for optimization problem \eqref{eqn:opt_estFU_modify} with respect to the error $\bfm\Delta$ and the residual ${\bfm\varepsilon}$.

There are two parts to be verified. The first is to show that the error $\bfm\Delta\in\mC$, where the set $\mC$ is defined following Definition \ref{def:rsc}, with the projections of $\bDelta_b$ defined based on the SVD of $\bB^*$, the lower-right part of $\bTheta^*$. 
This is straightforward from Lemma \ref{lem:general} and the condition that $\lambda_n\geq 2\|\xbar\mG_n^*({{\xbar{\bfm\varepsilon}}})\|_2 /n$, where $\xbar\mG_n^*(\xbar\bepsilon)$, as a reminder, is defined as $\sum_{i=1}^n \xbar\varepsilon_i \xbar\bG_i$.  

What left is to show that the adjust-RSC holds for $\xbar\mG_n$ with some constant ${\kappa}(\xbar\mG_n)$ for all matrices in $\mC$ based on the assumption that $\mG_n$ satisfied adjusted-RSC with constant $\kappa(\mG_n)$ over the set $\mC$. That is, to show  for all $ \bfm\Delta\in\mC$,
\[
\|\mG_n(\bfm\Delta)\|_2^2/(2n)\geq \kappa_2(\mG_n)\|\bfm\Delta\|_F^2 \Rightarrow \|\xbar\mG_n(\bfm\Delta)\|_2^2/(2n)\geq \kappa(\xbar\mG_n)\|\bfm\Delta\|_F^2, \text{ for some } \kappa(\xbar\mG_n).
\]
By triangular inequality, 
\[
\|\xbar\mG_n(\bfm\Delta)\|_2 \geq \|\mG_n(\bfm\Delta)\|_2 - \|\xbar\mG_n(\bfm\Delta)-\mG_n(\bfm\Delta)\|_2.
\]
By the definition of $\mG_n$ and Cauchy-Schwartz inequality, we have
\begin{equation}\label{eqn:mG}
\begin{aligned}
\|\xbar\mG_n(\bfm\Delta)-\mG_n(\bfm\Delta)\|_2 &= \left(\sum_{i=1}^n \angles{\xbar\bG_i - \bG_i, \bfm\Delta}^2\right)^{1/2}\\
&\leq \|\bfm\Delta\|_F \left(\sum_{i=1}^n \|\xbar{\bG}_i -\bG_i\|_F^2\right)^{1/2}.
\end{aligned}
\end{equation}
By proof of Theorem \ref{thm:U_element}, we have
\[
\left\|\frac{1}{\hat c_s}{\hat{\bR}} -\mathring{\bR}\right\|_2 = \mathcal \mathcal O_p\left( \sqrt{\frac{p_1}{np_2}}\right),\  \text{ and }\ 
\left\|\frac{1}{\hat r_s} {\hat{\bC}} - \mathring{\bC}\right\|_2= \mathcal \mathcal O_p\left(\sqrt{\frac{p_2}{np_1}} + \frac{1}{\sqrt{n}} \right).
\]
Moreover, by the Assumption \ref{asum:FU_subG}, we know that for any $j\in[k_1]$ and $\ell\in[k_2]$, $F_{1,j\ell},\ldots,F_{n,j\ell}$ are independent sub-Gaussian random variables, where $F_{i,j\ell}$ is the $(j,\ell)$-th element of $\bF_i$. Therefore, $\frac{1}{n} \sum_{i=1}^n F_{i,j\ell}^2$ is sub-Exponential. By Maximal tail inequality, we have 
\[
\PP\left( \max_{j,\ell} \frac{1}{n} \sum_{i=1}^n F_{i,j\ell}^2 \geq C\log k_1k_2 + Cu \right) \leq e^{-u}.
\]
Thus, 
$\frac{1}{n}\sum_{i=1}^n\|\bF_i\|_F^2 = \mathcal \mathcal O_p(1)$. 
With the assumption $\|\bR\|_2 = O(p_1)$ and  $\|\bC\|_2 = O(p_2)$, by Cauchy-Schwartz inequality, we obtain the following bound.
\[
\begin{aligned}
\frac{1}{n}\sum_{i=1}^n\|\hat{\bU}_{i}-\bU_{i}\|_F^2&= \frac{1}{n}\sum_{i=1}^n\left\|\frac{1}{\hat r_s \hat c_s} \hat{\bR} \hat{\bF}_i \hat{\bC}^\top - \bR^{*} \bF_i \bC^{*\top}\right\|_F^2\\
& \leq 7\left\|\frac{1}{\hat c_s}\hat{\bR} - \mathring c_s \mathring{\bR} \right\|_2^2 \|\bH_1\|_2^2 \|\bH_2\|_2^2 \frac{1}{n}\sum_{i=1}^n\|\bF_i\|_F^2 \left\|\frac{1}{\hat c_s}\hat{\bC} -  \mathring r_s\mathring{\bC} \right\|_2^2\\
&\qquad +7\left\|\frac{1}{\hat c_s}\hat{\bR} - \mathring{\bR} \right\|_2^2 \frac{1}{n}\sum_{i=1}^n\|\hat\bF_i-\bH_1\bF_i\bH_2^\top\|_F^2 \left\|\frac{1}{\hat r_s}\hat{\bC} -  \mathring{\bC} \right\|_2^2\\
&\qquad +7\left\|\mathring{\bR} \right\|_2^2 \|\bH_1\|_2^2 \|\bH_2\|_2^2 \frac{1}{n}\sum_{i=1}^n\|\hat\bF_i\|_F^2 \left\|\frac{1}{\hat r_s}\hat{\bC} - \mathring{\bC} \right\|_2^2\\
&\qquad +7\left\|\mathring{\bR} \right\|_2^2 \frac{1}{n}\sum_{i=1}^n\|\hat\bF_i-\bH_1\bF_i\bH_2^\top\|_F^2 \left\|\frac{1}{\hat r_s}\hat{\bC} - \mathring{\bC} \right\|_2^2\\
&\qquad +7\left\|\frac{1}{\hat c_s}\hat{\bR} - \mathring{\bR} \right\|_2^2 \frac{1}{n}\sum_{i=1}^n\|\hat\bF_i-\bH_1\bF_i\bH_2^\top\|_F^2 \left\|\mathring{\bC} \right\|_2^2\\
&\qquad +7\left\|\mathring{\bR} \right\|_2^2 \frac{1}{n}\sum_{i=1}^n\|\hat\bF_i-\bH_1\bF_i\bH_2^\top\|_F^2 \left\|\mathring{\bC} \right\|_2^2\\
&\qquad +7\left\|\frac{1}{\hat c_s}\hat{\bR} - \mathring{\bR} \right\|_2^2 \|\bH_1\|_2^2 \|\bH_2\|_2^2 \frac{1}{n}\sum_{i=1}^n\|\bF_i\|_F^2 \left\|\mathring{\bC} \right\|_2^2\\
& = \mathcal \mathcal O_p\left(\frac{p_1}{n} +\frac{p_2}{n} + 1 \right)
\end{aligned}
\]
Together with inequality \eqref{eqn:mG} and Proposition \ref{lem:Ferr_nuc}, we have
\[
\begin{aligned}
\frac{\|\xbar\mG_n(\bDelta) - \mG_n(\bDelta)\|_2^2}{n\|\bDelta\|_F^2} & \leq \frac{1}{n}\sum_{i=1}^n \|\xbar\bG_i - \bG_i\|_F^2\\
&= \frac{1}{n}\sum_{i=1}^n \left[\|\xbar\bF_i - \bF_i\|_F^2 + \|\hat\bU_i - \bU_i\|_F^2 \right]
\leq C_1 \left(\frac{p_1 \vee p_2}{n} + 1 \right)
\end{aligned}
\]
for some constant $C_1$.
Together with the { assumption $\frac{p_1\vee p_2}{n}<C_0$ and that $\kappa(\mG_n) > 2C_1(C_0+1)$}, we have
\[
\frac{\|\xbar\mG_n(\bfm\Delta)\|_2^2}{n\|\bfm\Delta\|_F^2}\geq \kappa(\mG_n) - C_1(C_0+1) \geq \frac{\kappa(\mG_n)}{2} =: \kappa(\xbar\mG_n).
\]

Finally, we can prove following the idea in the proof of Theorem \ref{thm:consistency1}, as briefly outlined the the follows. From inequality \eqref{eqn:bound_delta}, we have
\[
\begin{aligned}
\frac{1}{2n}\|\xbar\mG_n(\bfm\Delta)\|^2 & \leq \frac{1}{n} \angles{\xbar{\bfm\varepsilon}, \mG_n(\bfm\Delta)} + \lambda_n \big(\|\hat\bB + \bfm\Delta_B\|_* - \|\hat\bB\|_*\big)\\
&  = \frac{1}{n} \angles{\xbar\mG_n^*(\xbar\bepsilon),\bDelta} +  \lambda_n \big(\|\hat\bB + \bfm\Delta_B\|_* - \|\hat\bB\|_*\big)\\
& \leq \frac{1}{n} \|\xbar\mG_n^*(\xbar\bepsilon)\|_2 \|\bDelta\|_* + \lambda_n \big(\|\hat\bB + \bfm\Delta_B\|_* - \|\hat\bB\|_*\big)\\
& \leq \lambda_n \big(\frac{\|\bDelta\|_*}{2} + \|\bDelta_B\|_*\big) = \lambda_n  \big(\frac{1}{2} \|\bDelta_A\|_* + \frac{3}{2} \|\bDelta_B\|_*\big),
\end{aligned}
\]
where the last inequality is by the assumption that $\lambda_n\geq 2\|\xbar\mG_n^*(\xbar\bepsilon)\|_2/n$, and $\bDelta_A$ is defined as $\bA^* - \xbar\bA$.
Combining the above bound with the fact that $\bDelta\in\mC$ and $\|\xbar\mG_n(\bfm\Delta)\|_2^2/(2n)\geq \kappa(\xbar\mG_n)\|\bfm\Delta\|_F^2$, we derive that
\[
\|\bB^* - \hat\bB\|_F \leq \|\xbar{\bfm\Theta} - {\bTheta}^*\|_F \leq \frac{2\lambda_n}{\kappa_2(\mG_n)}\big(2\sqrt{k_1\wedge k_2} + 6\sqrt{2\operatorname{rank}(\bB^*)}\big).
\]
\end{proof}

\subsection{Proof of Proposition \ref{prop:err_order}}

In cases where FAMAR simplifies to matrix regression with MFM covariates \eqref{eqn:mr}, specifically when $\bA^* = \bR^\top\bB^*\bC$, it becomes possible to eliminate $\bar\bepsilon$ and directly utilize the true error $\bepsilon$ in determining the tuning parameter $\lambda_n$.

\begin{corollary}\label{cor:consistency_B}
Suppose the assumptions inTheorem \ref{thm:consistency2} hold.
When $\bA^* = \bR^\top\bB^*\bC$, when $\lambda_n\geq 2\|\xbar\mG_n^*(\bepsilon)\|_2 /n$, the result for $\|\hat{\bB}-\bB^*\|_F$ still holds.
\end{corollary}

\begin{proof}[Proof of Corollary \ref{cor:consistency_B}]
%============when A*=RBC======
Consider the same optimization problem as \eqref{eqn:convex}:
\begin{equation}
\big(\xbar{\bA},\hat{\bB}\big)\in\argmin_{\bA\in\RR^{k_1\times k_2},\\ \bB\in\RR^{p_1\times p_2}} \left\{\frac{1}{2n}\sum_{i=1}^n \left( y_i-\angles{\bA,\bH_1^{-1}\hat{\bF}_i(\bH_2^{-1})^\top}-\angles{\bB,\hat{\bU}_i} \right)^2 + \lambda_n\|\bB\|_{*}\right\}.
\end{equation}
Let $\xbar{\bfm\Theta}:= \diag(\xbar{\bA},\hat{\bB})$.
Different from the proof of Theorem \ref{thm:consistency2}, we define $$\xbar\bTheta^* = \operatorname{diag}(\bH_1^{-1}\hat\bR^\top \bB^*\hat\bC(\bH_2^{-1})^\top, \bB^*).$$ Recall that $\xbar\bG_i = \operatorname{diag}((\bH_1^{-1})^\top \hat\bF_i \bH_2^{-1}, \hat\bU_i)$, and the operator $\xbar\mG_n: \RR^{(k_1+p_1)\times(k_2+p_2)} \rightarrow \RR$ is defined that $\xbar\mG_n(\bTheta) = \angles{\xbar\bG_i, \bTheta}$. By the decomposition $\bX_i = \hat\bR\hat\bF_i \hat\bC^\top + \hat\bU_i$, we have
\[
\varepsilon_i = y_i - \angles{\bX_i,\bB^*} = y_i - \angles{\xbar\bG_i, \xbar\bTheta^*}.
\]
Letting $\bfm\Delta:={\xbar\bTheta}^*- \xbar{\bfm\Theta}$, $\bfm\Delta_B:=\bB^*-\hat\bB$, we have $\|\bDelta_B\|_F \leq \|\bDelta\|_F$.
Optimality of $\xbar\bTheta$ gives \[
\frac{1}{2n}\big\|\by - \xbar\mG_n(\xbar{\bfm\Theta})\big\|^2 + \lambda_n\|\hat{\bB}\|_* \leq \frac{1}{2n}\big\|\by - \xbar\mG_n({\xbar\bTheta}^*)\big\|^2 + \lambda_n\|{\bB}^*\|_* = \frac{1}{2n} \|\bepsilon\|^2+ \lambda_n\|{\bB}^*\|_*.
\] 
By elementary algebra, we get
\begin{equation}
\frac{1}{2n}\|\xbar\mG_n(\bfm\Delta)\|^2 \leq \frac{1}{n} \angles{{\bfm\varepsilon}, \mG_n(\bfm\Delta)} + \lambda_n \big(\|\hat\bB + \bfm\Delta_B\|_* - \|\hat\bB\|_*\big).
\end{equation}
Similar to Theorem \ref{thm:consistency2}, to upper bound the error of $\hat\bB$, we only need to show that all the regularization conditions in Theorem \ref{thm:consistency1} hold for optimization problem \eqref{eqn:opt_estFU_modify} with respect to the error $\bfm\Delta$ and the residual ${\bfm\varepsilon}$.

Lemma \ref{lem:general} and the condition that $\lambda_n\geq 2\|\xbar\mG_n^*({{\bfm\varepsilon}})\|_2 /n$ shows that  $\bfm\Delta\in\xbar\mC$, where the set $\xbar\mC$ is defined following Definition \ref{def:rsc}, with the projections of $\bDelta_b$ defined based on the SVD of $\bB^*$, the lower-right part of $\xbar\bTheta^*$.  
It is straightforward to show that $\xbar\mC$ = $\mC$, where $\xbar\mC$ is defined based on $\xbar\bTheta^*$ while $\mC$ is defined based on $\bTheta^*$, for both $\xbar\bTheta^*$ and $\bTheta^*$ have the structure $\operatorname{diag}(*, \bB^*)$.
Thus, we have $\bDelta:= {\xbar\bTheta}^*- \xbar{\bfm\Theta} \in \mC$.

The rest of the proof is exactly the same as that for Theorem \ref{thm:consistency2} by replacing $\xbar\bepsilon$ with $\bepsilon$.

\end{proof}
%--------------------

\begin{proof}[Proof of Proposition \ref{prop:err_order}]
By taking $\lambda_n\asymp 2\big\|\xbar\mG_n^*(\bfm\varepsilon)\big\|_2/n$, Corollary \ref{cor:consistency_B} indicates that 
\[
\frac{\|\hat{\bB}-\bB^*\|_F}{\big\|\xbar\mG_n^*({\bfm\varepsilon})\big\|_2 /n} =\Op{1},
\]
when all other quantities are finite.
By triangular inequality,
\begin{equation}\label{eqn:errorder1}
\big\|\xbar{\mG}_n^*(\bfm\varepsilon)\big\|_2 
\leq \big\|{\mG}_n^*(\bfm\varepsilon)\big\|_2 + \big\|\xbar{\mG}_n^*(\bfm\varepsilon)-{\mG}_n^*(\bfm\varepsilon)\big\|_2 = \big\|\sum_{i=1}^n \varepsilon_i \bG_i\big\|_2 + \big\|\sum_{i=1}^n \varepsilon_i (\xbar{\bG}_i-\bG_i)\big\|_2.
\end{equation}

By inequality \eqref{eqn:errorder_Gepsilon} in the proof of Proposition \ref{prop:err_order0}, we already have
\begin{equation}\label{eqn:errorder2}
\PP\left(\left\|\frac{1}{n}\mG_n^*(\bfm\varepsilon)\right\| \geq C_1\sqrt{\frac{\log d}{n}}\right) \leq d^{-2}
\end{equation}
for $C_1$ large enough.

For the second term, by Cauchy-Schwartz inequality, we have
\[
\big\|\sum_{i=1}^n \varepsilon_i (\xbar{\bG}_i-\bG_i)\big\|_2 \leq \sum_{i=1}^n |\varepsilon_i|\big\|\xbar\bG_i-\bG_i\big\|_2 \leq \left(\sum_{i=1}^n \varepsilon_i^2\right)^{1/2} \left(\sum_{i=1}^n \big\|\xbar\bG_i-\bG_i\big\|_2^2\right)^{1/2}.
\]
Since $\varepsilon_i$ is $\sigma-$sub-Gaussian, we know $\varepsilon_i^2$ is sub-exponential with parameters $(4\sqrt{2}\sigma^2,\, 4\sigma^2)$, i.e.
\[
\EE\left[e^{\lambda \varepsilon_i^2}\right] \leq e^{16\sigma^4\lambda^2},\ \forall |\lambda|<\frac{1}{4\sigma^2}.
\]
Hence $\sum_{i=1}^n \varepsilon_i^2 /n$ is sub-Exponential with parameters $(4\sqrt{2}\sigma^2,\, 4\sqrt{n}\sigma^2)$, and therefore
\begin{equation}\label{eqn:errorder3}
\PP\left(\frac{1}{n}\sum_{i=1}^n \varepsilon_i^2 \geq C_2\sqrt{\frac{\log n}{n}}\right) \leq n^{-1}
\end{equation}
for large enough $C_2$.
In addition, by Corollary \ref{cor:F} and Theorem \ref{thm:U_element}, we have
\begin{equation}\label{eqn:errorder4}
\begin{aligned}
\sum_{i=1}^n \big\|\xbar\bG_i - \bG_i\big\|_2^2 & \leq \sum_{i=1}^n \big\|\xbar\bG_i - \bG_i\big\|_F^2 \leq \sum_{i=1}^n \big\|\xbar\bF_i - \bF_i\big\|_F^2 + \sum_{i=1}^n \big\|\xbar\bU_i - \bU_i\big\|_F^2 \\
& = \mathcal \mathcal O_p\left(p_1p_2\big(\sqrt{\frac{n(k_1\wedge k_2)k_1k_2}{p_1p_2}} + 1\big)^2\right).
\end{aligned}
\end{equation}

Plugging \eqref{eqn:errorder2}, \eqref{eqn:errorder3}, and \eqref{eqn:errorder4}  into \eqref{eqn:errorder1}, we get
\[
\frac{\big\|\xbar{\mG}_n^*(\bfm\varepsilon)\big\|_2}{n} = \mathcal \mathcal O_p\left(\sqrt{\frac{\log(k_1+k_2+p_1+p_2)}{n}} + \frac{(\log n)^{1/4}}{n^{1/4}} \sqrt{(k_1\wedge k_2)k_1k_2} + \frac{(\log n)^{1/4}}{n^{3/4}}\sqrt{p_1p_2} \right),
\]
where the second term is negligible.
The proof is completed by Corollary \ref{cor:consistency_B}.
\end{proof}

\section{FAMAR-Sparse with Observed $\bF_i$} 
\label{append:sparse-B-observed-F}

When the true $\bB^*$ is sparse, the results from \cite{fan2020factor} can be similarly applied by vectorizing the matrices and applying sparse regularization. 
Here, we take Lasso as example.
Let $\btheta:= (\vect (\bB)^\top, \vect (\bA)^\top)^\top\in\RR^{p_1p_2+K_1K_2}$, and $\bbeta:= \btheta_{[p_1p_2]} = \vect (\bB)$.
Denote $S:= \supp (\btheta^*)$, $S_1 = \supp (\bbeta^*)$, and $S_2 = [p_1p_2+K_1K_2]\setminus S$.
Further, let $\bw_i:= (\vect (\bU_i)^\top, \vect(\bF_i)^\top)^\top$, $\hat\bw_i:= (\vect(\hat\bU_i)^\top, \vect(\hat\bF_i)^\top)^\top$, and $\xbar\bw_i:= (\vect(\hat\bU_i)^\top, \vect(\bH_1^{-1}\hat\bF_i(\bH_2^\top)^{-1})^\top)^\top$,
Let $\bW := (\bw_1,\ldots,\bw_n)^\top$, $\hat\bW := (\hat\bw_1,\ldots,\hat\bw_n)^\top$, and $\xbar\bW := (\xbar\bw_1,\ldots,\xbar\bw_n)^\top$. 
By adding the subscript $S$, we only take the $j$-th column for all $j\in S$. 
Define the loss function 
\[
L_n(\by, \bW\btheta) =\frac{1}{2n} \|\by - \bW\btheta\|^2.
\]
We solve
\begin{equation*} %\label{eqn:sparse_opt}
\big(\hat{\bA}^*,\hat{\bB}^*\big)\in\argmin_{\bA,\bB} \left\{\frac{1}{2n}\sum_{i=1}^n \left( y_i-\angles{\bA,{\bF}_i}-\angles{\bB,{\bU}_i} \right)^2 + \lambda_n\|\vect (\bB)\|_{1}\right\}.
\end{equation*}	
After vectorization, this equals to
\[
\hat\btheta^*\in\argmin_{\btheta\in\RR^{p_1p_2+K_1K_2}} \left\{L_n(\by,\bW\btheta)+ \lambda_n\|\btheta_{[p_1p_2]}\|_{1}\right\},
\]	
where $\hat\btheta^*= (\vect (\hat\bB^*)^\top, \vect (\hat\bA^*)^\top)^\top$.

We make the following assumptions.
\begin{assumption}\label{asu:Baparse_W}
All elements of $\bW$ are sub-Gaussian.
\end{assumption}
\begin{proposition}\label{prop:Bsparse_Fknown}
\begin{itemize}
\item[(i)] Error bounds : Under Assumptions \ref{asu:Bsparse_Fknown}, if
$
\lambda_n>\frac{7}{\tau}\left\|\frac{1}{n} (\bW^\top \bW)\right\|_{\infty},
$
then $\operatorname{supp}(\widehat{\boldsymbol{\theta}^*}) \subseteq S$ and
$$
\begin{aligned}
& \left\|\widehat{\boldsymbol{\theta}}^*-\boldsymbol{\theta}^*\right\|_{\infty} \leq \frac{3}{5 \kappa_{\infty}}\left(\left\|\nabla_S L_n\left(\by, \bW\boldsymbol{\theta}^*\right)\right\|_{\infty}+\lambda_n\right), \\
& \left\|\widehat{\boldsymbol{\theta}}^*-\boldsymbol{\theta}^*\right\|_2 \leq \frac{2}{\kappa_2}\left(\left\|\nabla_S L_n\left(\by, \bW\boldsymbol{\theta}^*\right)\right\|_2+\lambda_n \sqrt{\left|S_1\right|}\right), \\
& \left\|\widehat{\boldsymbol{\theta}}^*-\boldsymbol{\theta}^*\right\|_1 \leq \min \left\{\frac{3}{5 \kappa_{\infty}}\left(\left\|\nabla_S L_n\left(\by, \bW\boldsymbol{\theta}^*\right)\right\|_1+\lambda_n\left|S_1\right|\right), \frac{2 \sqrt{|S|}}{\kappa_2}\left(\left\|\nabla_S L_n\left(\by, \bW\boldsymbol{\theta}^*\right)\right\|_2+\lambda_n \sqrt{\left|S_1\right|}\right)\right\} .
\end{aligned}
$$
\item[(ii)] Sign consistency : In addition, if the following condition
$$
\begin{aligned}
& \min \left\{\left|\boldsymbol{\beta}_j^*\right|: \boldsymbol{\beta}_j^* \neq 0, j \in[p]\right\}>\frac{C}{\kappa_{\infty} \tau}\left\|\nabla L_n\left(\boldsymbol{\theta}^*\right)\right\|_{\infty}
\end{aligned}
$$
holds for some $C \geq 5$, then by taking $\lambda \in\left(\frac{7}{\tau}\left\|\nabla L_n\left(\boldsymbol{\theta}^*\right)\right\|_{\infty}, \frac{1}{\tau}\left(\frac{5 C}{3}-1\right)\left\|\nabla L_n\left(\boldsymbol{\theta}^*\right)\right\|_{\infty}\right)$, the estimator achieves the sign consistency $\operatorname{sign}(\widehat{\boldsymbol{\beta}}^*)=\operatorname{sign}\left(\boldsymbol{\beta}^*\right)$.

\end{itemize}

\end{proposition}

\begin{proof}[Proof of Proposition \ref{prop:Bsparse_Fknown}]
%\noindent\textbf{Error bounds:} 

Taking gradient w.r.t. $\btheta_S$, we get
\[
\begin{aligned}
&\triangledown_S L_n(\by, \bW\btheta) = \frac{1}{n} \bW_S^\top(\bW\btheta-\by) \in \RR^{|S|},\\
&\triangledown_{SS}^2 L_n(\by, \bW\btheta) = \frac{1}{n} \bW_S^\top\bW_S\in \RR^{|S|\times |S|},\\
&\triangledown_{jS}^2 L_n(\by, \bW\btheta) = \frac{1}{n} \sum_{i=1}^n w_{ij} \bw_{iS}.
\end{aligned}
\]
Thus, for any $\btheta\in\RR^{p_1p_1+K_1K_2}$, $\|\triangledown_{\cdot S}^2 L_n(\by, \bW\btheta) - \triangledown_{\cdot S}^2 L_n(\by, \bW\btheta^*)\|_\infty =0.$
Assumption 4.1 in \cite{fan2020factor} holds with $M=1$ and $A=\infty$.
Assumption 4.2 and 4.3 in \cite{fan2020factor} are satisfied by Assumption \ref{asu:Bsparse_Fknown}.
Therefore, by Theorem 4.1 in \cite{fan2020factor}, we get the results.
\end{proof}

\begin{corollary}\label{cor:Bsparse_Fknown}
Suppose all the conditions in Proposition \ref{prop:Bsparse_Fknown} hold, and moreover, Assumption \ref{asu:Baparse_W} holds. Then, by taking $\lambda_n>\frac{7}{\tau}\left\|\frac{1}{n} (\bW^\top \bW)\right\|_{\infty}$ and $\lambda_n\asymp\frac{7}{\tau}\left\|\frac{1}{n} (\bW^\top \bW)\right\|_{\infty}$,
we have $\operatorname{supp}(\widehat{\boldsymbol{\theta}^*}) \subseteq S$ and
$$
\begin{aligned}
& \left\|\widehat{\boldsymbol{\theta}}^*-\boldsymbol{\theta}^*\right\|_{\infty} = \mathcal O_p\left(\sqrt{\frac{|S|}{ n}} + \frac{(\log(dn))^2}{n}\right) \\
& \left\|\widehat{\boldsymbol{\theta}}^*-\boldsymbol{\theta}^*\right\|_2 =\mathcal O_p\left(\sqrt{\frac{|S|}{n}} + \frac{\sqrt{|S|}(\log(dn))^2}{n}\right) \\
& \left\|\widehat{\boldsymbol{\theta}}^*-\boldsymbol{\theta}^*\right\|_1 =\mathcal O_p\left(\frac{|S|}{\sqrt n}+ \frac{{|S|}(\log(dn))^2}{n}\right),
\end{aligned}
$$
where $d= p_1p_2+K_1K_2$.
\end{corollary}

\begin{proof}[Proof of Corollary \ref{cor:Bsparse_Fknown}]
Recall that $\bepsilon = \by- \bW\btheta^*$. We have
\[
\begin{aligned}
&\|\triangledown_S L_n(\by, \bW\btheta^*)\|_\infty = \frac{1}{n} \|\bW_S^\top \bepsilon\|_\infty =  \mathcal O_p\left(\sqrt{\frac{|S|}{n}}\right),\\
& \|\triangledown_S L_n(\by, \bW\btheta^*)\|_2 = \frac{1}{n} \|\bW_S^\top \bepsilon\|_2 = { \mathcal O_p\left(\sqrt{\frac{|S|}{n}}\right)},\\
& \|\triangledown_S L_n(\by, \bW\btheta^*)\|_1 = \frac{1}{n} \|\bW_S^\top \bepsilon\|_1 =  \mathcal O_p\left({\frac{|S|}{\sqrt n}}\right).
\end{aligned}
\]

We also have to determine the order of $\lambda_n$. 
By taking $d=p_1p_2+K_1K_2$ and by Assumption \ref{asu:Baparse_W}, we have 
\[
\begin{aligned}
\PP( |w_{ij}| \|\bw_{i}\| \leq t) 
&\geq \PP(|w_{ij}|\leq \sqrt{t},  \|\bw_{i}\| \leq \sqrt{t}) \\
&\geq \PP(\|\bw_{i}\| \leq \sqrt{t} \wedge C\sqrt d)- \PP(|w_{ij}|> \sqrt{t}) \\
&\geq 1-e^{-c_1t} - e^{-c_2t}.
\end{aligned}
\]
Thus, 
\[
\begin{aligned}
\PP(\sum_{i=1}^n |w_{ij}| \|\bw_{i}\| \leq t) &
= 1-\PP(\sum_{i=1}^n |w_{ij}| \|\bw_{i}\| > t)\\
&\geq 1- \sum_{i=1}^n \PP( |w_{ij}| \|\bw_{i}\| > t)\\
& \geq 1- n\cdot (e^{-c_{1i}t}+e^{-c_{2i}t}).
\end{aligned}
\]
Therefore,
\[
\begin{aligned}
\PP( \|\bW^\top\bW\|_\infty \leq t) 
& = \PP(\max_{j\in[d]} \| \sum_{i=1} w_{ij} \bw_i \| \leq t )\\
& \geq \PP(\max_{j\in[d]} \sum_{i=1} |w_{ij}| \|\bw_i \| \leq t )\\
& =  \PP(\forall j\in[d], \sum_{i=1} |w_{ij}| \|\bw_i \| \leq t )\\
& \geq \sum_{i=1}^d \PP(\sum_{i=1} |w_{ij}| \|\bw_i \| \leq t) - (d-1)\\
& \geq 1-dn(e^{-c_1t}+ e^{-c_2t}).
\end{aligned}
\]
By taking $t=(\log (dn))^2$, we have 
\[
\lambda_n = C\cdot \frac{7}{\tau}\left\|\frac{1}{n} (\bW^\top \bW)\right\|_{\infty} = \mathcal O_p\left(\frac{(\log (dn))^2}{n}\right).
\]
\end{proof}

\section{FAMAR-Sparse with Estimated $\bF_i$} \label{append:sparse-B-estimated-F}

%\subsection{Estimate Sparse $\bB$ with Approximate Factor Model}
Recall that we defined $\bW := (\bw_1,\ldots,\bw_n)^\top$, $\hat\bW := (\hat\bw_1,\ldots,\hat\bw_n)^\top$, and $\xbar\bW := (\xbar\bw_1,\ldots,\xbar\bw_n)^\top$, where $\xbar\bw_i:= (\vect(\hat\bU_i)^\top, \vect(\bH_1^{-1}\hat\bF_i(\bH_2^\top)^{-1})^\top)^\top$.
The estimators solves the following optimization problem: 
\[
\big(\hat{\bA},\hat{\bB}\big)\in\argmin_{\bA,\bB} \left\{\frac{1}{2n}\sum_{i=1}^n \left( y_i-\angles{\bA,\hat{\bF}_i}-\angles{\bB,\hat{\bU}_i} \right)^2 + \lambda_n\|\vect (\bB)\|_{1}\right\}.
\]	
The analysis on $\hat{\bA}$ is similar to that of Theorem \ref{thm:consistency2} based on the fact that $\hat{\FF}^\top \hat{\UU} = \bzero$, which leads to
\[
\vect(\hat\bA) =  \big(\hat{\FF}^\top \hat{\FF}\big)^{-1}\hat{\FF}^\top \by.
\]
The consistency of $\hat\bB$ is established in Theorem \ref{thm:Bsparse_Funknown}. Before proving Theorem \ref{thm:Bsparse_Funknown}, we first introduce a technical lemma and proposition.
\begin{lemma}\label{lem:inv_sub}
Suppose $\mathbf{A} \in \mathbb{R}^{q \times r}$ and $\mathbf{B}, \mathbf{C} \in \mathbb{R}^{r \times r}$ and $\left\|\mathbf{C B}^{-1}\right\|<1$, where $\|\cdot\|$ is an induced norm. Then $\left\|\mathbf{A}\left[(\mathbf{B}+\mathbf{C})^{-1}-\mathbf{B}^{-1}\right]\right\| \leq \frac{\left\|\mathbf{A} \mathbf{B}^{-1}\right\| \cdot\left\|\mathbf{C B}^{-1}\right\|}{1-\left\|\mathbf{C} \mathbf{B}^{-1}\right\|}$.
\end{lemma}

{
\begin{proposition}\label{prop:u-element}
For the element-wise error of $\hat\bU_i$'s, we have the following results.
\begin{itemize}
\item[(a)] $\max_{j\in[p_1],k\in[p_2]}\frac{1}{n} \sum_{i=1}^n |u_{i,jk} - \hat u_{i,jk}|  = \mathcal O_p\left(\frac{log p_1p_2}{n} + \sqrt{\frac{\log p_1p_2}{n}} \left(\frac{p_1}{p_2} \vee \frac{p_2}{p_1} \right)^{1/2} \right).$
\item[(b)] $\max_{j\in[p_1],k\in[p_2]}\frac{1}{n} \sum_{i=1}^n |u_{i,jk} - \hat u_{i,jk}|^2  = \mathcal O_p\left(\frac{log^2 p_1p_2}{n^2} + \frac{\log p_1p_2}{n} \left(\frac{p_1}{p_2} \vee \frac{p_2}{p_1} \right) \right)$;
\item[(c)] $\max_{j\in[p_1],k\in[p_2]} |u_{i,jk} - \hat u_{i,jk}|  = \mathcal O_p\left(\frac{log p_1p_2}{n} + \left[\frac{\log p_1p_2}{n} \left(\frac{p_1}{p_2} \vee \frac{p_2}{p_1} \right)\right]^{1/2} + \sqrt{\frac{\log np_1p_2}{p_2p_2}} \right)$.
\end{itemize}
\end{proposition}

\begin{proof}[Proof of Proposition \ref{prop:u-element}]
To start with, By Definition \ref{def:W} (a), we have
\[
\left\|\hat\FF - \FF(\bH_2\otimes\bH_1)^\top\right\|_F^2 = \frac{1}{p_1^2p_2^2} \|\UU(\bW_2\otimes\bW_1)\|_F^2 \leq \frac{1}{p_1^2p_2^2} \|\UU\|_F^2 \|\bW_2\|_2^2 \|\bW_1\|_2^2 = \mathcal O_p\left( \frac{n}{p_1p_2}\right).
\]
Moreover, similar to \cite{fan2013large}, we have
\[
\frac{1}{n} \|\FF^\top\UU\|_{\max} = \max_{j\in[p_1p_2], k\in[k_1k_2]} \left|\frac{1}{n} \sum_{i=1}^n u_{ij}f_{ik}\right| = \mathcal O_p\left( \sqrt{\frac{\log p_1p_2}{n}}\right).
\]
Thus, by Proposition \ref{thm:ols-pre-trained}, we have
\[
\begin{aligned}
\left\|\frac{1}{\hat r_s}\hat{\bC} - \mathring{\bC}\right\|_{\max} &= \left\|\frac{1}{p_1k_1} \bE_{p_2p_1} \bS_{p_2p_1} \left[\big(\hat{\FF}^\top\hat{\FF}\big)^{-1}(\bH_2\otimes\bH_1)\FF^\top \UU \right]^\top \bS_{k_2k_1}^\top \bE_{k_2k_1}^\top\right\|_{\max} + \mathcal O_p\left(\sqrt{\frac{p_2}{np_1}} + \frac{1}{\sqrt{n}}\right)\\
& \leq \frac{C}{np_1} \sqrt{p_2}\|\bE_{p_2p_1}\|_2 \|\bS_{p_2p_1}\|_2 \dfrac{1}{\lambda_{\min}(\hat\FF^{\top}\hat\FF/n)} \|\bH_2\otimes\bH_1\|_2 \|\FF^\top\UU\|_{\max} \|\bS_{k_2k_1}\|_2 \|\bE_{k_2k_1}\|_2 \\
&\qquad \qquad \qquad\qquad \qquad \qquad\qquad \qquad \qquad\qquad \qquad \qquad + \mathcal O_p\left(\sqrt{\frac{p_1}{np_2}} + \frac{1}{\sqrt{n}}\right)\\
\vspace{-0.2in}
& =  \mathcal O_p\left(\sqrt{\frac{p_1\log p_1p_2}{np_2}} + \frac{1}{\sqrt{n}}\right).
\end{aligned}
\]
Similarly, we have $\|\frac{1}{\hat c_s}\hat\bR - \mathring \bR\|_{\max} = \mathcal O_p\left(\sqrt{\frac{p_2\log p_1p_2}{np_1}} + \frac{1}{\sqrt{n}}\right)$.
Therefore, by triangle inequality, we have
{\allowdisplaybreaks
% \[
\begin{align*}
&\max_{j\in[p_1], k\in[p_2]} \frac{1}{n}\sum_{i=1}^n|\hat{u}_{i,jk}-u_{i,jk}|\\
&\leq  \max_{j\in[p_1]} \left\|\frac{1}{\hat c_s}\hat{\bR}_{j\cdot} - \mathring{\bR}_{j\cdot} \right\| \|\bH_1\|_2 \|\bH_2\|_2 \frac{1}{n}\sum_{i=1}^n\|\bF_i\|_F \max_{k\in[p_2]} \left\|\frac{1}{\hat r_s}\hat{\bC}_{k\cdot} - \mathring{\bC}_{k\cdot} \right\|\\
&\qquad + \max_{j\in[p_1]} \left\|\frac{1}{\hat c_s}\hat{\bR}_{j\cdot} - \mathring{\bR}_{j\cdot} \right\| \frac{1}{n}\sum_{i=1}^n\|\hat\bF_i-\bH_1\bF_i\bH_2^\top\|_F \max_{k\in[p_2]} \left\|\frac{1}{\hat c_s}\hat{\bC}_{k\cdot} - \mathring{\bC}_{k\cdot} \right\|\\
&\qquad +\max_{j\in[p_1]}\left\|\mathring{\bR}_{j\cdot} \right\| \|\bH_1\|_2 \|\bH_2\|_2 \frac{1}{n}\sum_{i=1}^n\|\hat\bF_i\|_F \max_{k\in[p_2]} \left\|\frac{1}{\hat r_s}\hat{\bC}_{k\cdot} - \mathring{\bC}_{k\cdot} \right\|\\
&\qquad +\max_{j\in[p_1]}\left\|\mathring{\bR}_{j\cdot} \right\| \frac{1}{n}\sum_{i=1}^n\|\hat\bF_i-\bH_1\bF_i\bH_2^\top\|_F  \max_{k\in[p_2]} \left\|\frac{1}{\hat r_s}\hat{\bC}_{k\cdot} - \mathring{\bC}_{k\cdot} \right\|\\
&\qquad + \max_{j\in[p_1]} \left\|\frac{1}{\hat c_s}\hat{\bR}_{j\cdot} - \mathring{\bR}_{j\cdot} \right\| \frac{1}{n}\sum_{i=1}^n\|\hat\bF_i-\bH_1\bF_i\bH_2^\top\|_F \max_{k\in[p_2]}\left\|\mathring{\bC}_{k\cdot} \right\|\\
&\qquad +\max_{j\in[p_1]}\left\|\mathring{\bR}_{j\cdot} \right\| \frac{1}{n}\sum_{i=1}^n\|\hat\bF_i-\bH_1\bF_i\bH_2^\top\|_F  \max_{k\in[p_2]}\left\| \mathring{\bC}_{k\cdot} \right\|\\
&\qquad + \max_{j\in[p_1]} \left\|\frac{1}{\hat c_s}\hat{\bR}_{j\cdot} - \mathring{\bR}_{j\cdot} \right\| \|\bH_1\|_2 \|\bH_2\|_2 \frac{1}{n}\sum_{i=1}^n\|\bF_i\|_F \max_{k\in[p_2]}\left\|\mathring{\bC}_{k\cdot} \right\|\\
& = \mathcal O_p\left(\frac{log p_1p_2}{n} + \sqrt{\frac{\log p_1p_2}{n}} \left(\frac{p_1}{p_2} \vee \frac{p_2}{p_1} \right)^{1/2} \right).
\end{align*}
% \]
}
Similarly, by Cauchy-Schwartz inequality, we have
{\allowdisplaybreaks
\begin{align*}
&\max_{j\in[p_1], k\in[p_2]} \frac{1}{n}\sum_{i=1}^n|\hat{u}_{i,jk}-u_{i,jk}|^2\\
&\leq 7 \max_{j\in[p_1]} \left\|\frac{1}{\hat c_s}\hat{\bR}_{j\cdot} - \mathring{\bR}_{j\cdot} \right\|^2 \|\bH_1\|_2^2 \|\bH_2\|_2^2 \frac{1}{n}\sum_{i=1}^n\|\bF_i\|_F^2 \max_{k\in[p_2]} \left\|\frac{1}{\hat r_s}\hat{\bC}_{k\cdot} - \mathring{\bC}_{k\cdot} \right\|^2\\
&\qquad +7 \max_{j\in[p_1]} \left\|\frac{1}{\hat c_s}\hat{\bR}_{j\cdot} - \mathring{\bR}_{j\cdot} \right\|^2 \frac{1}{n}\sum_{i=1}^n\|\hat\bF_i-\bH_1\bF_i\bH_2^\top\|_F^2 \max_{k\in[p_2]} \left\|\frac{1}{\hat r_s}\hat{\bC}_{k\cdot} -  \mathring{\bC}_{k\cdot} \right\|^2\\
&\qquad +7\max_{j\in[p_1]}\left\| \mathring{\bR}_{j\cdot} \right\|^2 \|\bH_1\|_2^2 \|\bH_2\|_2^2 \frac{1}{n}\sum_{i=1}^n\|\hat\bF_i\|_F^2 \max_{k\in[p_2]} \left\|\frac{1}{\hat r_s}\hat{\bC}_{k\cdot} - \mathring{\bC}_{k\cdot} \right\|^2\\
&\qquad +7\max_{j\in[p_1]}\left\|\mathring{\bR}_{j\cdot} \right\|^2 \frac{1}{n}\sum_{i=1}^n\|\hat\bF_i-\bH_1\bF_i\bH_2^\top\|_F^2  \max_{k\in[p_2]} \left\|\frac{1}{\hat r_s}\hat{\bC}_{k\cdot} -  \mathring{\bC}_{k\cdot} \right\|^2\\
&\qquad +7 \max_{j\in[p_1]} \left\|\frac{1}{\hat c_s}\hat{\bR}_{j\cdot} - \mathring{\bR}_{j\cdot} \right\|^2 \frac{1}{n}\sum_{i=1}^n\|\hat\bF_i-\bH_1\bF_i\bH_2^\top\|_F^2 \max_{k\in[p_2]}\left\|\mathring{\bC}_{k\cdot} \right\|^2\\
&\qquad +7\max_{j\in[p_1]}\left\|\mathring{\bR}_{j\cdot} \right\|^2 \frac{1}{n}\sum_{i=1}^n\|\hat\bF_i-\bH_1\bF_i\bH_2^\top\|_F^2  \max_{k\in[p_2]}\left\|\mathring{\bC}_{k\cdot} \right\|^2\\
&\qquad +7 \max_{j\in[p_1]} \left\|\frac{1}{\hat c_s}\hat{\bR}_{j\cdot} - \mathring{\bR}_{j\cdot} \right\|^2 \|\bH_1\|_2^2 \|\bH_2\|_2^2 \frac{1}{n}\sum_{i=1}^n\|\bF_i\|_F^2 \max_{k\in[p_2]}\left\|\mathring{\bC}_{k\cdot} \right\|^2\\
& = \mathcal O_p\left(\frac{log^2 p_1p_2}{n^2} + \frac{\log p_1p_2}{n} \left(\frac{p_1}{p_2} \vee \frac{p_2}{p_1} \right) \right).
\end{align*}
}

Furthermore, replacing the average over $i$ in the above inequality with maximum over $i$ and $n^{-1}\sum_{i=1}^n\|\hat\bF_i - \bH_1\bF_i\bH_2^\top\|_F^2$ with $k_1k_2\|\hat\FF - \FF(\bH_2\otimes\bH_1)^\top\|_{\max}$, we can can also derive
\[
\max_{i\in[n], j\in[p_1], k\in[p_2]}|\hat{u}_{i,jk}-u_{i,jk}|^2 =  \mathcal O_p\left(\frac{log^2 p_1p_2}{n^2} + \frac{\log p_1p_2}{n} \left(\frac{p_1}{p_2} \vee \frac{p_2}{p_1} + \frac{\log np_1p_2}{p_1p_2} \right) \right).
\]
\end{proof}
}

\subsection{Proof of Theorem \ref{thm:Bsparse_Funknown}}

\begin{proof}[Proof of Theorem \ref{thm:Bsparse_Funknown}]
Replacing the true $(\bF_i,\bU_i)$ in MFM \eqref{eqn:famar} by the estimated $(\hat\bF_i. \hat{\bU}_i)$, we have
\[
\begin{aligned}
y_i &= \angles{\bF_i,\bA^*} +  \angles{\bU_i,\bB^*} + \varepsilon_i\\
& = \angles{(\bH_1^{-1})^\top\hat\bF_i\bH_2^{-1},\bA^*} +  \angles{\hat\bU_i,\bB^*} + \xbar\varepsilon_i 
= \xbar\bw_i^\top \btheta^* + \xbar\varepsilon_i,
\end{aligned}
\]
where $\xbar\varepsilon_i := y_i -  \xbar\bw_i^\top \btheta^*$.
Consider the optimization problem
\[
\xbar\btheta\in\argmin_{\btheta\in\RR^{p_1p_2+K_1K_2}} \left\{L_n(\by,\xbar\bW\btheta)+ \lambda_n\|\vect (\bB)\|_{1}\right\}.
\]	
We have $\|\vect(\hat\bB)-\vect(\bB^*)\|\leq\|\xbar\btheta-\btheta^*\|$ for any norm $\|\cdot\|$.

To apply Proposition \ref{prop:Bsparse_Fknown}, we need to show that the conditions on $\bW$ needed to hold for Proposition \ref{prop:Bsparse_Fknown} also hold for $\xbar\bW$.
We start with checking the restricted strong convexity. 
By Proposition \ref{prop:H}, \ref{lem:Ferr_nuc}, and \ref{prop:u-element} (c), we have
\[
\begin{aligned}
\max_{i\in[n], j\in[d]}	|\xbar{w}_{ij} - w_{ij}| =\mathcal O_p\left(\frac{log p_1p_2}{n} + \left[\frac{\log p_1p_2}{n} \left(\frac{p_1}{p_2} \vee \frac{p_2}{p_1} \right)\right]^{1/2} + \sqrt{\frac{\log np_1p_2}{p_2p_2}} \right).
\end{aligned}
\]
Moreover, by Assumption \ref{asu:Baparse_W} and maximal tail inequality, we have 
\[
\max_{j\in[p_1],k\in[p_2]}\frac{1}{n} \sum_{i=1}^n |\xbar w_{i,jk} - \hat w_{i,jk}|  = \mathcal O_p\left(\frac{log p_1p_2}{n} + \sqrt{\frac{\log p_1p_2}{n}} \left(\frac{p_1}{p_2} \vee \frac{p_2}{p_1} \right)^{1/2} \right).
\]
Also, by Assumption \ref{asu:Baparse_W} and maximal tail inequality, we have $\max_{i\in[n],j\in[d]}|w_{ik}| = \mathcal O_p(\sqrt{\log (np_1p_2)})$. %in ORF550 lecture note p149%
Therefore, for $j\in[d]$, 
\begin{equation}\label{eqn:W_row_sub}
\begin{aligned}
&\max_{j\in[d]} \|\frac{1}{n} \sum_{i=1}^n \xbar w_{ij} \xbar\bw_{iS} - w_{ij} \bw_{iS}  \|_1 \\
&  \leq \max_{j\in[d]}\frac{1}{n} \sum_{i=1}^n \left[\|\xbar w_{ij} \xbar\bw_{iS} - \xbar w_{ij} \bw_{iS} \|_1 + \|\xbar w_{ij} \bw_{iS} - w_{ij} \bw_{iS} \|_1\right] \\
& \leq \max_{j\in[d]}\frac{1}{n} \sum_{i=1}^n \left[|\xbar w_{ij} | \sum_{k\in S}|\xbar w_{ik} - w_{ik}|+  |\xbar w_{ij} - w_{ij}| \sum_{k\in S}  |w_{ik}|\right]\\
&\leq \max_{i\in[n],j\in[d]} |\xbar w_{ij}| \cdot \sum_{k\in S} \left[\frac{1}{n}\sum_{i=1}^n |\xbar w_{ik}- w_{ik}| \right] + \max_{i\in[n],j\in[d]} |\xbar w_{ij} - w_{ij}| \cdot \sum_{k\in S} \left[\frac{1}{n} \sum_{i=1}^n |w_{ik}|\right]\\
& =\mathcal O_p\left\{ \sqrt{\log (np_1p_2)} \cdot |S| \left[\frac{log p_1p_2}{n} + \left[\frac{\log p_1p_2}{n} \left(\frac{p_1}{p_2} \vee \frac{p_2}{p_1} \right)\right]^{1/2} + \left(\frac{\log np_1p_2}{p_2p_2}\right)^{1/2} \right] \right\}\\
& = o_p(1).
\end{aligned}
\end{equation}

Let $\alpha_{\infty} = \|(\bW_S^\top\bW_S)^{-1}(\xbar\bW_S^\top\xbar\bW_S - \bW_S^\top\bW_S)\|_\infty$. With the assumption $ \|n(\bW_S^\top\bW_S)^{-1}\|_\infty \leq \kappa_{\infty}^{-1}$ and bound \eqref{eqn:W_row_sub}, we have
\begin{equation}\label{eqn:alpha}
\begin{aligned}
\alpha_{\infty}& \leq \|n(\bW_S^\top\bW_S)^{-1}\|_\infty \|\frac{1}{n} \xbar\bW_S^\top\xbar\bW_S - \frac{1}{n}\bW_S^\top\bW_S)\|_\infty \\
&=  \|n(\bW_S^\top\bW_S)^{-1}\|_\infty \max_{j\in S} \|\frac{1}{n} \sum_{i=1}^n \xbar w_{ij} \xbar\bw_{iS} - w_{ij} \bw_{iS} \|_1 \\
&\leq \frac{1}{\kappa_\infty} \cdot C\sqrt{\log (np_1p_2)} \cdot |S| \left[\frac{log p_1p_2}{n} + \left[\frac{\log p_1p_2}{n} \left(\frac{p_1}{p_2} \vee \frac{p_2}{p_1} \right)\right]^{1/2} + \left(\frac{\log np_1p_2}{p_2p_2}\right)^{1/2} \right]\\
&< \frac{1}{2}
\end{aligned}
\end{equation}
with high probability for some constant $C$ when $n, p_1,p_2$ are large enough.
Hence, by Lemma \ref{lem:inv_sub}, we have
\[
\begin{aligned}
\|n(\xbar\bW_S^\top \xbar\bW_S)^{-1} - n(\bW_S^\top \bW_S)^{-1} \|_\infty 
\leq \|n(\bW_S^\top \bW_S)^{-1} \|_\infty \cdot \frac{\alpha_{\infty}}{1-\alpha_{\infty}}
\leq \frac{2\alpha_{\infty}}{\kappa_{\infty}} 
\leq \frac{1}{2\kappa_{\infty}}.
\end{aligned}
\]
Together with Assumption \ref{asu:Bsparse_Fknown} and triangle inequality, we have
\[
\|n(\xbar\bW_S^\top \xbar\bW_S)^{-1} \|_\infty \leq 	\|n(\bW_S^\top \bW_S)^{-1} \|_\infty \leq \| + \frac{1}{2\kappa_\infty} \leq\frac{1}{\kappa_\infty},
\]
\[
\begin{aligned}
\left\|\frac{1}{n} \xbar\bW^\top\xbar\bW - \frac{1}{n}\bW^\top\bW\right\|_2 
& \leq \sqrt{|S|}   \|\frac{1}{n} \xbar\bW^\top\xbar\bW - \frac{1}{n}\bW^\top\bW\|_\infty.
\end{aligned}
\]
Meanwhile, for matrix 2-norm, let $\alpha_{2} = \|(\bW_S^\top\bW_S)^{-1}(\xbar\bW_S^\top\xbar\bW_S - \bW_S^\top\bW_S)\|_2$. Then
\[
\begin{aligned}
\alpha_{2}& \leq \|n(\bW_S^\top\bW_S)^{-1}\|_2 \|\frac{1}{n} \xbar\bW_S^\top\xbar\bW_S - \frac{1}{n}\bW_S^\top\bW_S\|_2 \\
&\leq  \|n(\bW_S^\top\bW_S)^{-1}\|_2  \sqrt{|S|}   \|\frac{1}{n} \xbar\bW_S^\top\xbar\bW_S - \frac{1}{n}\bW_S^\top\bW_S\|_\infty\\
&\leq \frac{1}{\kappa_2} \cdot C\sqrt{\log (p_1p_2)} |S|^{3/2} \left[\frac{log p_1p_2}{n} + \left[\frac{\log p_1p_2}{n} \left(\frac{p_1}{p_2} \vee \frac{p_2}{p_1} \right)\right]^{1/2} + \left(\frac{\log np_1p_2}{p_2p_2}\right)^{1/2} \right]\\
&< \frac{1}{2}
\end{aligned}
\]
%	when $p_1,p_2$ are large enough.
By Lemma \ref{lem:inv_sub}, we have
\[
\begin{aligned}
\left\|n(\xbar\bW_S^\top \xbar\bW_S)^{-1} - n(\bW_S^\top \bW_S)^{-1} \right\|_2 
\leq \|n(\bW_S^\top \bW_S)^{-1} \|_2 \cdot \frac{\alpha_{2}}{1-\alpha_{2}}
\leq \frac{1}{2\kappa_{2}}.
\end{aligned}
\]
when $p_1,p_2$ are large enough. 
Together with Assumption \ref{asu:Bsparse_Fknown} and triangle inequality, we have
\[
\|n(\xbar\bW_S^\top \xbar\bW_S)^{-1} \|_2 \leq 	\|n(\bW_S^\top \bW_S)^{-1} \|_2 + \frac{1}{2\kappa_2} \leq\frac{1}{\kappa_2}.
\]
Therefore, the restricted strong convexity holds for $\xbar\bW$.

We also need to show that the irrepresentable condition holds for $\xbar\bW$. 
\[
\begin{aligned}
&\left\|\xbar\bW_{S_2}^\top \xbar\bW_S (\xbar\bW_S^\top\xbar\bW_S)^{-1}- \bW_{S_2}^\top \bW_S (\bW_S^\top\bW_S)^{-1}\right\|_\infty\\
\leq& \big\| \frac{1}{n}\xbar\bW_{S_2}^\top \xbar\bW_S - \frac{1}{n}\bW_{S_2}^\top \bW_S \big\|_\infty \|n(\xbar\bW_S^\top\xbar\bW_S)^{-1}\|_\infty  
+  \left\| \frac{1}{n} \bW_{S_2}^\top \bW_S\left[n(\xbar\bW_S^\top\xbar\bW_S)^{-1}-n(\bW_S^\top\bW_S)^{-1}\right]\right\|_\infty\\
=:& I+II.
\end{aligned}
\]
By \eqref{eqn:W_row_sub} and Assumption \ref{asu:Bsparse_Fknown} (a), we have
\[
\begin{aligned}
I &=   \|n(\bW_S^\top\bW_S)^{-1}\|_\infty \max_{j\in S_2} \|\frac{1}{n} \sum_{i=1}^n \xbar w_{ij} \xbar\bw_{iS} - w_{ij} \bw_{iS} \|_1\\
& \leq \frac{1}{\kappa_\infty} \cdot C\sqrt{\log (np_1p_2)} \cdot |S| \left[\frac{log p_1p_2}{n} + \left[\frac{\log p_1p_2}{n} \left(\frac{p_1}{p_2} \vee \frac{p_2}{p_1} \right)\right]^{1/2} + \left(\frac{\log np_1p_2}{p_2p_2}\right)^{1/2} \right].
\end{aligned}
\]
with high probability.
By Lemma \ref{lem:inv_sub} and Assumption \ref{asu:Bsparse_Fknown} (b), we have
\[
\begin{aligned}
II& \leq \| \bW_{S_2}^\top\bW_S(\bW_S^\top \bW_S)^{-1} \|_\infty \frac{\|(\frac{1}{n} \xbar\bW_S^\top\xbar\bW_S - \frac{1}{n} \bW_S^\top\bW_S) \cdot n(\bW_S^\top\bW_S)^{-1}\|_\infty}{1-\|(\frac{1}{n} \xbar\bW_S^\top\xbar\bW_S - \frac{1}{n} \bW_S^\top\bW_S) \cdot n(\bW_S^\top\bW_S)^{-1}\|_\infty}\\
& \leq\frac{\|(\frac{1}{n} \xbar\bW_S^\top\xbar\bW_S - \frac{1}{n} \bW_S^\top\bW_S)\|_\infty \|n(\bW_S^\top\bW_S)^{-1}\|_\infty}{1-\|(\frac{1}{n} \xbar\bW_S^\top\xbar\bW_S - \frac{1}{n} \bW_S^\top\bW_S)\|_\infty \| n(\bW_S^\top\bW_S)^{-1}\|_\infty}\\
&\leq \frac{\frac{1}{\kappa_\infty} \cdot C\sqrt{\log (np_1p_2)} \cdot |S| \left[\frac{log p_1p_2}{n} + \left[\frac{\log p_1p_2}{n} \left(\frac{p_1}{p_2} \vee \frac{p_2}{p_1} \right)\right]^{1/2} + \left(\frac{\log np_1p_2}{p_2p_2}\right)^{1/2} \right]}{1-1/2}\\
&\leq \frac{2C}{\kappa_\infty} \cdot C\sqrt{\log (np_1p_2)} \cdot |S| \left[\frac{log p_1p_2}{n} + \left[\frac{\log p_1p_2}{n} \left(\frac{p_1}{p_2} \vee \frac{p_2}{p_1} \right)\right]^{1/2} + \left(\frac{\log np_1p_2}{p_2p_2}\right)^{1/2} \right]
\end{aligned}
\]
with high probability, where the third inequality is by \eqref{eqn:W_row_sub}.

Combining these two bounds, we have
\[
\begin{aligned}
&\|\xbar\bW_{S_2}^\top \xbar\bW_S (\xbar\bW_S^\top\xbar\bW_S)^{-1}- \bW_{S_2}^\top \bW_S (\bW_S^\top\bW_S)^{-1}\|_\infty\\
\leq& C \sqrt{\log(p_1p_2)} |S| \left[\frac{log p_1p_2}{n} + \left[\frac{\log p_1p_2}{n} \left(\frac{p_1}{p_2} \vee \frac{p_2}{p_1} \right)\right]^{1/2} + \left(\frac{\log np_1p_2}{p_2p_2}\right)^{1/2} \right]
\leq  \tau/2
\end{aligned}
\]
with high probability when $n,p_1,p_2$ are large enough.
Therefore, by triangle inequality and Assumption \ref{asu:Bsparse_Fknown} (b), we have
\[
\|\xbar\bW_{S_2}^\top \xbar\bW_S (\xbar\bW_S^\top\xbar\bW_S)^{-1}\|_\infty \leq 1-\tau + \frac{\tau}{2} = 1-\frac{\tau}{2}.
\]

To sum up, we have verified that the conditions needed for Proposition \ref{prop:Bsparse_Fknown} also holds for $\xbar\bW$, and thus, by Proposition \ref{prop:Bsparse_Fknown}, we have
\begin{equation}\label{eqn:theta_temp}
\begin{aligned}
& \left\|\widehat{\boldsymbol{\theta}}^*-\boldsymbol{\theta}^*\right\|_{\infty} \leq \frac{6}{5 \kappa_{\infty}}\left(\left\|\nabla_S L_n\left(\by, \xbar\bW\boldsymbol{\theta}^*\right)\right\|_{\infty}+\lambda\right), \\
& \left\|\widehat{\boldsymbol{\theta}}^*-\boldsymbol{\theta}^*\right\|_2 \leq \frac{4}{\kappa_2}\left(\left\|\nabla_S L_n\left(\by, \xbar\bW\boldsymbol{\theta}^*\right)\right\|_2+\lambda \sqrt{\left|S_1\right|}\right), \\
& \left\|\widehat{\boldsymbol{\theta}}^*-\boldsymbol{\theta}^*\right\|_1 \leq \min \left\{\frac{6}{5 \kappa_{\infty}}\left(\left\|\nabla_S L_n\left(\by, \xbar\bW\boldsymbol{\theta}^*\right)\right\|_1+\lambda\left|S_1\right|\right), \frac{4 \sqrt{|S|}}{\kappa_2}\left(\left\|\nabla_S L_n\left(\by, \xbar\bW\boldsymbol{\theta}^*\right)\right\|_2+\lambda \sqrt{\left|S_1\right|}\right)\right\} .
\end{aligned}
\end{equation}
Notice that
\[
\nabla_S L_n\left(\by, \xbar\bW\boldsymbol{\theta}^*\right) =  \frac{1}{n} \xbar\bW_S^\top(\xbar\bW\btheta^*-\by).
\]
Thus 
\[
\left\|\nabla_S L_n\left(\by, \xbar\bW\boldsymbol{\theta}^*\right)\right\|_{\infty} \leq \xbar\delta,
\quad \left\|\nabla_S L_n\left(\by, \xbar\bW\boldsymbol{\theta}^*\right)\right\|_{2} \leq \xbar\delta\sqrt{|S|},
\quad \left\|\nabla_S L_n\left(\by, \xbar\bW\boldsymbol{\theta}^*\right)\right\|_{1} \leq \xbar\delta |S|.
\]
In addition, $\lambda > 14\delta/\tau \geq \xbar\delta$. We conclude the proof by substituting the above bounds to \eqref{eqn:theta_temp}.

\end{proof}

\subsection{Proof of Theorem \ref{thm:Bsparse_Funknown2}}

\begin{proof}[Proof of Theorem \ref{thm:Bsparse_Funknown2}]
Define $\xbar\btheta^* = (\vect(\bB^*)^\top, [(\bH_2\otimes\bH_1)^\top (\hat\bC\otimes\hat{\bR}) \vect(\bB^*)]^\top)^\top$. Then we still have $\|\vect(\hat\bB)-\vect(\bB^*)\|\leq\|\xbar\btheta-\xbar\btheta^*\|$ for any norm $\|\cdot\|$. Also, $\bw_i^\top \btheta^* = \xbar\bw_i^\top\xbar\btheta^*$.
Let 
$
\delta = \max_{j\in S} \delta_j, 
$
where
\[
\begin{aligned}
\delta_j &= \left|\frac{1}{n} \sum_{i=1}^n \xbar w_{ij} (\xbar\bw_i\xbar\btheta^* - y_i)\right| 
= \left|\frac{1}{n} \sum_{i=1}^n \xbar w_{ij} (\bw_i\btheta^* - y_i)\right|= \left|\frac{1}{n} \sum_{i=1}^n \xbar w_{ij} \varepsilon_i\right|\\
&= \left|\frac{1}{n} \sum_{i=1}^n  w_{ij} \varepsilon_i + \frac{1}{n} \sum_{i=1}^n (\xbar w_{ij} - w_{ij}) \varepsilon_i\right|\\
&\leq \left|\frac{1}{n} \sum_{i=1}^n  w_{ij} \varepsilon_i \right| +  \left( \frac{1}{n} \sum_{i=1}^n (\xbar w_{ij} - w_{ij})^2 \right)^{1/2} \left( \frac{1}{n}\sum_{i=1}^n\varepsilon_i^2 \right)^{1/2}\\
&= \mathcal O_p\left(\frac{1}{\sqrt{n}} + \frac{log p_1p_2}{n} + \sqrt{\frac{\log p_1p_2}{n}} \left(\frac{p_1}{p_2} \vee \frac{p_2}{p_1} \right) ^{1/2} \right),
\end{aligned}
\]
where the inequality is by triangle inequality, Cauchy-Schwartz inequality, and Proposition \ref{prop:u-element} (b).
By union bound, we have 
\[
\delta = \mathcal O_p\left(\frac{|S|\log p_1p_2}{n} + |S|\sqrt{\frac{\log p_1p_2}{n}} \left(\frac{p_1}{p_2} \vee \frac{p_2}{p_1} \right) ^{1/2} \right).
\]

%	{\red $\mathcal O_p$: probability..}
\end{proof}

%=========================
\section{Extra Experiment Results} \label{append:extra-experiments}

\subsection{FAMAR Coefficient Estimation and Prediction}
{%\blue
As an extension of the simulations in Section \ref{sec:simul}, we present three more settings with a different way of generating true $\bA^*$. Recall that previously in Section \ref{sec:simul}, $\bA^*$ is generated as $\bA^* = 0.5\cdot \bR^\top \bB* \bC$. As mentioned in Section \ref{sec:intro}, we can also deal with cases where $\bA^*$ is not related with $\bR^\top \bB^* \bC$ at all. To show this in simulations, Figure \ref{fig:Aarb2} to \ref{fig:Aarb3} present the result for matrix factor and idiosyncratic component estimation, FAMAR coefficient estimation, and dependent variable prediction. The detailed settings for each group of simulations are as follows.
\begin{itemize}
\item Figure \ref{fig:Aarb1}: Fix $n=1000$, increase $p_1=p_2$ from $20$ to $100$. All other settings are the same as that for simulations in Figure \ref{fig:pFAMAR_f2r2} except that $\bA^*$ is generated such that each elements are i.i.d from $\cN(1000,1000^2)$.
\item Figure \ref{fig:Aarb2}: All other settings are the same as that for simulations in Figure \ref{fig:Aarb1} except that the elements of $\bA^*$ are i.i.d. from $\cN(p_1p_2, (p_1p_2)^2)$, with mean and standard deviation growing with the dimension $p_1,p_2$.
\item Figure \ref{fig:Aarb3}: Fix $(p_1,p_2) = (70, 50)$, increase $n$ from $500$ to $4000$. All other settings are the same as that for simulations in Figure \ref{fig:nFAMAR_f2r2} except that $\bA^*$ is generated such that each elements are i.i.d from $\cN(1000,1000^2)$.
\end{itemize}

The results are mostly similar to that in Section \ref{sec:simul}, showing the efficiency of FAMAR.
Notice that the trend of relative errors on $y$ is different from what we see in Figure \ref{fig:pFAMAR_f2r2} as $p_1$ and $p_2$ increases. {%\red 
Since the dimension of factors is fixed as (2,2) and does not change with $p_1$ and $p_2$, the contribution of $\angles{\bA^*,\bF_i}$ decreases as the dimension increases. To verify that, we make the elements in $\bA^*$ to have mean and std grow with problem dimensions with order $O(\sqrt{p_1p_2})$.} The results are shown in Figure \ref{fig:Aarb2}, where the errors on $y$ present similar patterns as that in Figure \ref{fig:pFAMAR_f2r2} .
}

\begin{figure}[ht!]
\centering
\begin{subfigure}[b]{0.45\textwidth}
\centering
\includegraphics[width=\textwidth]{figs/Farb7U4_2_all_A_err_p_k22.pdf}
\end{subfigure}
\begin{subfigure}[b]{0.45\textwidth}
\centering
\includegraphics[width=\textwidth]{figs/Farb7U4_2_all_B_err_p_k22.pdf}
\end{subfigure}
\hfill
\begin{subfigure}[b]{0.45\textwidth}
\centering
\includegraphics[width=\textwidth]{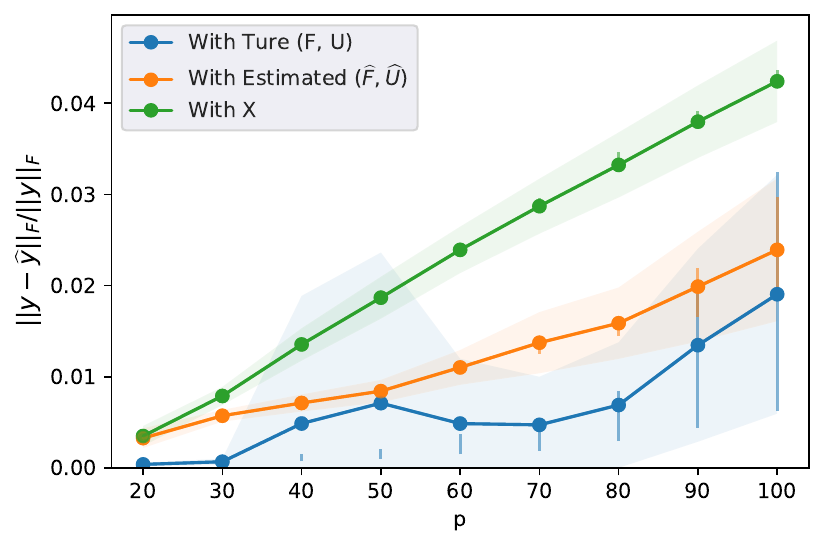}
\end{subfigure}
\begin{subfigure}[b]{0.45\textwidth}
\centering
\includegraphics[width=\textwidth]{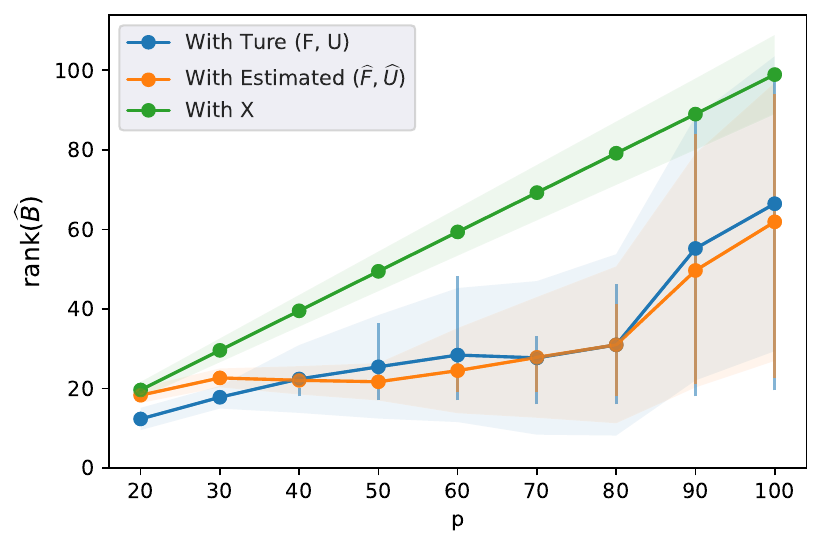}
\end{subfigure}
\hfill
\begin{subfigure}[b]{0.45\textwidth}
\centering
\includegraphics[width=\textwidth]{figs/Farb7U4_2_all_F_err_p_k22.pdf}
\end{subfigure}
\begin{subfigure}[b]{0.45\textwidth}
\centering
\includegraphics[width=\textwidth]{figs/Farb7U4_2_all_U_err_p_k22.pdf}
\end{subfigure}
\caption{Fix $n=1000$, increase $p_1=p_2$ from $20$ to $100$.  $\bA^*$ is generated such that each elements are i.i.d from $\cN(1000,1000^2)$. The blue, orange, and green lines correspond to doing the regression with true $(\bF,\bU)$, estimated $(\hat{\bF},\hat{\bU})$, and $\bX$, respectively. The experiments are repeated 100 times. The solid lines are the means, the shadows show the standard deviation, and the vertical lines present the quartiles.}
\label{fig:Aarb1}
\end{figure}

\begin{figure}[ht!]
\centering
\begin{subfigure}[b]{0.45\textwidth}
\centering
\includegraphics[width=\textwidth]{figs/Farb7U4_2_all_A_err_Ap_k22.pdf}
\end{subfigure}
\begin{subfigure}[b]{0.45\textwidth}
\centering
\includegraphics[width=\textwidth]{figs/Farb7U4_2_all_B_err_Ap_k22.pdf}
\end{subfigure}
\hfill
\begin{subfigure}[b]{0.45\textwidth}
\centering
\includegraphics[width=\textwidth]{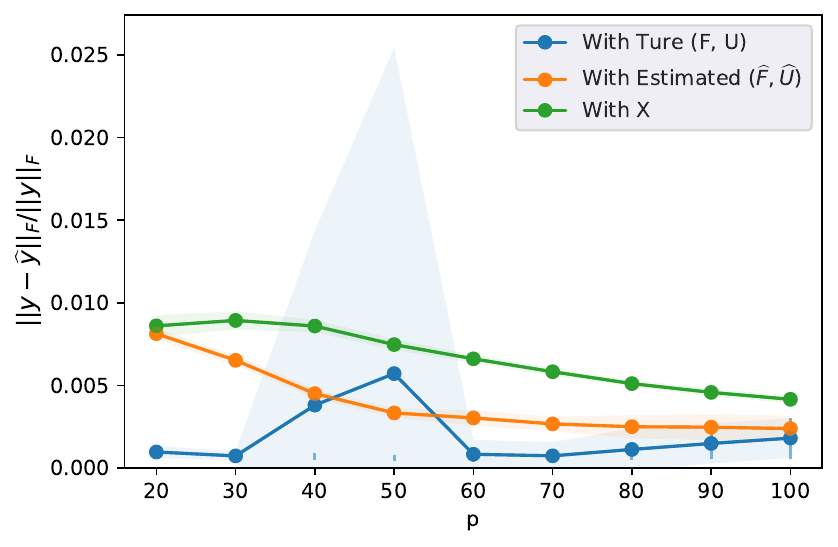}
\end{subfigure}
\begin{subfigure}[b]{0.45\textwidth}
\centering
\includegraphics[width=\textwidth]{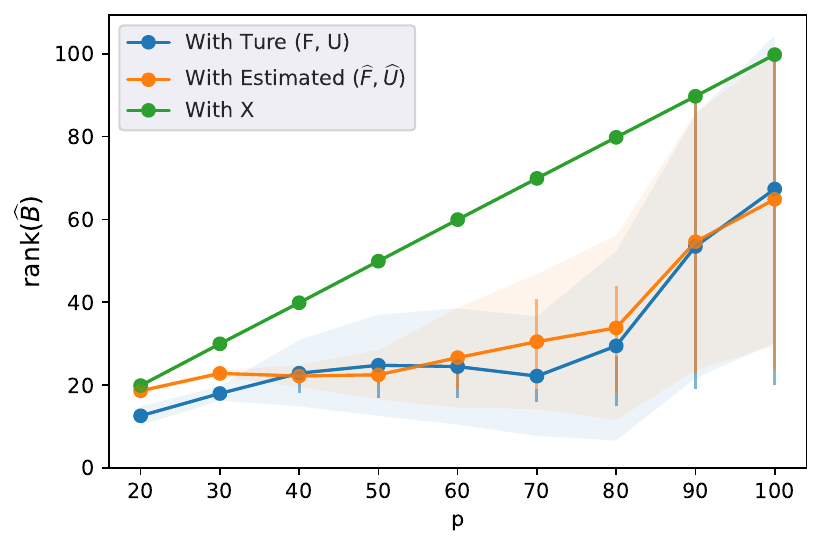}
\end{subfigure}
\hfill
\begin{subfigure}[b]{0.45\textwidth}
\centering
\includegraphics[width=\textwidth]{figs/Farb7U4_2_all_F_err_Ap_k22.pdf}
\end{subfigure}
\begin{subfigure}[b]{0.45\textwidth}
\centering
\includegraphics[width=\textwidth]{figs/Farb7U4_2_all_U_err_Ap_k22.pdf}
\end{subfigure}
\caption{Fix $n=1000$, increase $p_1=p_2$ from $20$ to $100$.  $\bA^*$ is generated such that each elements are i.i.d from $\cN(p_1p_2,(p_1p_2)^2)$. The blue, orange, and green lines correspond to doing the regression with true $(\bF,\bU)$, estimated $(\hat{\bF},\hat{\bU})$, and $\bX$, respectively. The experiments are repeated 100 times. The solid lines are the means, the shadows show the standard deviation, and the vertical lines present the quartiles.}
\label{fig:Aarb2}
\end{figure}

\begin{figure}[ht!]
\centering
\begin{subfigure}[b]{0.45\textwidth}
\centering
\includegraphics[width=\textwidth]{figs/Farb7U4_2_all_A_err_n_k24.pdf}
\end{subfigure}
\begin{subfigure}[b]{0.45\textwidth}
\centering
\includegraphics[width=\textwidth]{figs/Farb7U4_2_all_B_err_n_k24.pdf}
\end{subfigure}
\hfill
\begin{subfigure}[b]{0.45\textwidth}
\centering
\includegraphics[width=\textwidth]{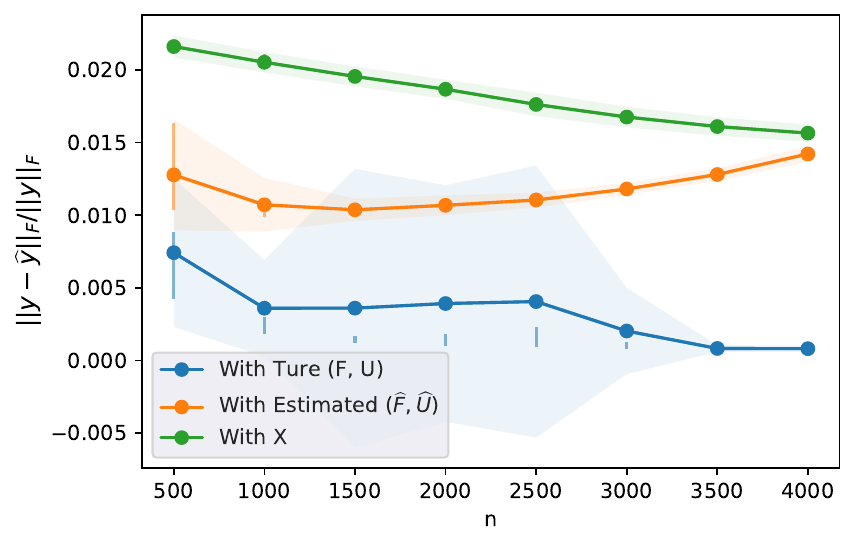}
\end{subfigure}
\begin{subfigure}[b]{0.45\textwidth}
\centering
\includegraphics[width=\textwidth]{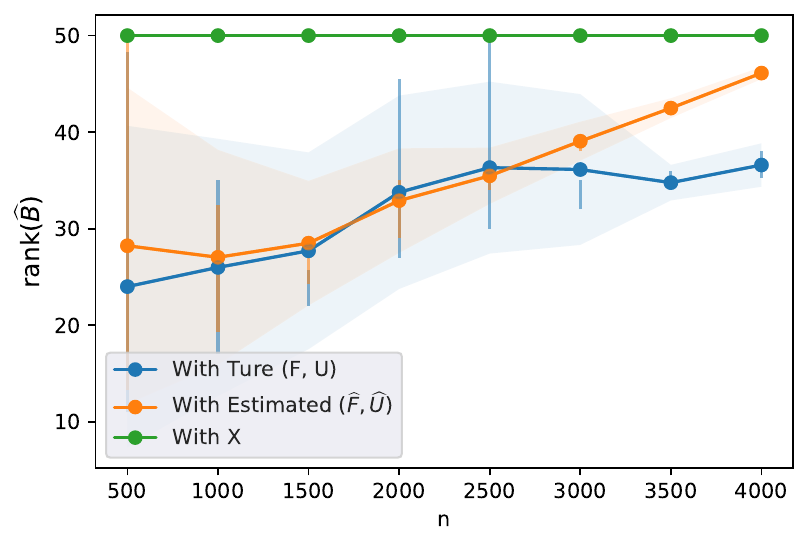}
\end{subfigure}
\hfill
\begin{subfigure}[b]{0.45\textwidth}
\centering
\includegraphics[width=\textwidth]{figs/Farb7U4_2_all_F_err_n_k24.pdf}
\end{subfigure}
\begin{subfigure}[b]{0.45\textwidth}
\centering
\includegraphics[width=\textwidth]{figs/Farb7U4_2_all_U_err_n_k24.pdf}
\end{subfigure}
\caption{Fix $(p_1,p_2) = (70,50)$, increase $n$ from $500$ to $4000$. $\bA^*$ is generated such that each elements are i.i.d from $\cN(1000,1000^2)$. The blue, orange, and green lines correspond to doing the regression with true $(\bF,\bU)$, estimated $(\hat{\bF},\hat{\bU})$, and $\bX$, respectively. The experiments are repeated 100 times. The solid lines are the means, the shadows show the standard deviation, and the vertical lines present the quartiles.}
\label{fig:Aarb3}
\end{figure}

\subsection{FAMAR Matrix Loading Asymptotically Normality.}\label{app:RC_normality}
We present the full graph of the asymptotic normality check of  the estimated loadings $\hat{\bR}$ and $\hat{\bC}$ according to Theorem  \ref{thm:MCMR} in Figure \ref{RC_normality}.
\begin{figure}[ht!]
\centering
\begin{subfigure}[b]{\textwidth}
\centering
\includegraphics[width=\textwidth]{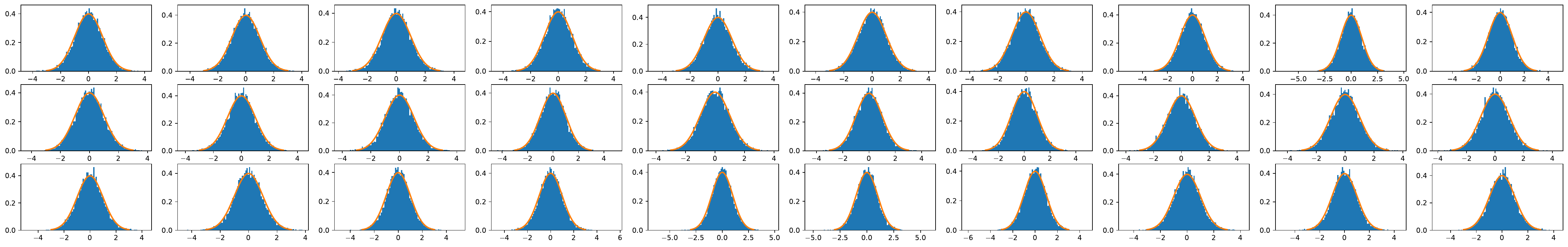}
\end{subfigure}
\begin{subfigure}[b]{\textwidth}
\centering
\includegraphics[width=\textwidth]{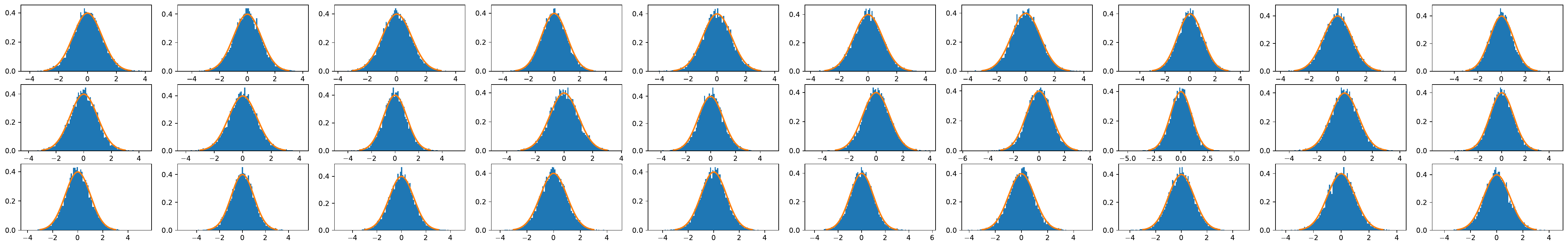}
\caption{Elements for $\sqrt{np_2k_2} \left(\hat{\bR} - \mathring{c}_s \mathring{\bR^*}\right)^\top/\hat{s}_1\in\RR^{3\times 20}$}
\end{subfigure}
\hfill
\begin{subfigure}[b]{\textwidth}
\centering
\includegraphics[width=\textwidth]{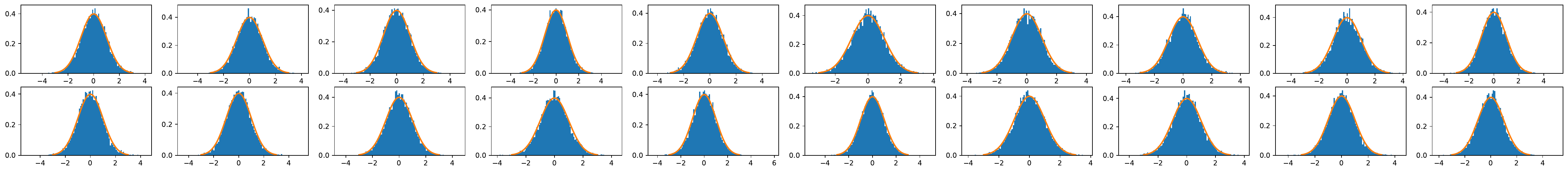}
\end{subfigure}
\begin{subfigure}[b]{\textwidth}
\centering
\includegraphics[width=\textwidth]{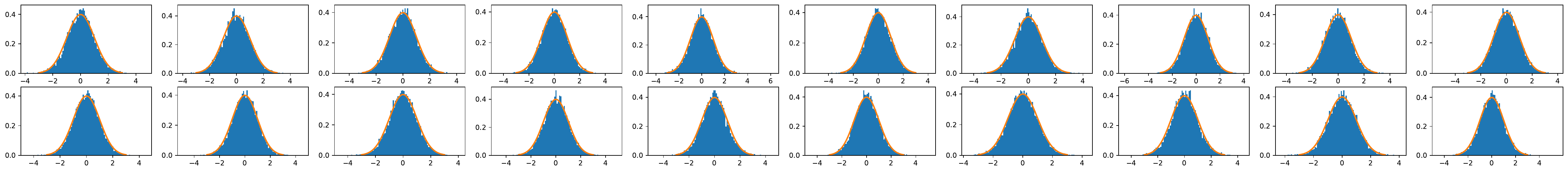}
\end{subfigure}
\begin{subfigure}[b]{\textwidth}
\centering
\includegraphics[width=\textwidth]{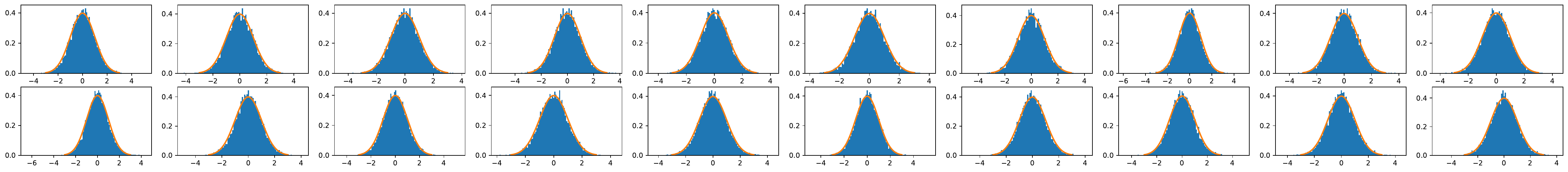}
\caption{Elements for $\sqrt{np_1k_1} \left(\hat{\bC} - \mathring{r}_s \mathring{\bC^*}\right)^\top/\hat{s}_2\in \RR^{2\times 30}$}
\end{subfigure}
\caption{{Elementwise normality check for the estimated loadings $\hat{\bR}$ and $\hat{\bC}$ according to Theorem \ref{thm:MCMR}. The blue histograms are empirical densities, and the orange lines are the standard Gaussian density functions.}}
\label{RC_normality}
\end{figure}

\end{appendices}

\end{document}